\documentclass[acmsmall,screen,nonacm,acmthm=false]{acmart}

\usepackage{algorithm}
\usepackage{algpseudocode}
\usepackage{bbm}

\usepackage{aliascnt}
\usepackage{cleveref}

\newcommand{\capac}{C}
\newcommand{\inst}{\mathrm{I}}
\newcommand{\opt}{\mathrm{Opt}}
\newcommand{\sched}{\mathrm{S}}
\newcommand{\rate}{\mathrm{rate}}
\newcommand{\greedy}{\mathrm{G}}
\newcommand{\gr}{\mathrm{gr}}
\newcommand{\expect}[1]{\operatorname{E}\!\left[#1\right]}
\newcommand{\prob}[1]{\Pr\!\left[#1\right]}
\newcommand{\on}{\mathrm{ON}}
\newcommand{\algname}{\mathrm{PrefixLookahead}}
\newcommand{\bad}{\mathcal{B}}
\newcommand{\advname}{\mathrm{RandomShift}}
\newcommand{\batch}{\mathrm{batch}}
\newcommand{\wait}{\mathrm{wait}}
\newcommand{\off}{\mathrm{off}}

\newcommand{\act}{\mathrm{act}}
\newcommand{\ign}{\mathrm{ign}}

\newcommand{\cycjobs}{\mathrm{N}}
\newcommand{\cycflow}{\mathrm{F}}
\newcommand{\cycdur}{\mathrm{D}}
\newcommand{\cyc}{\mathcal{C}}
\newcommand{\makespan}{\mathrm{Makespan}}
\newcommand{\early}{\mathrm{early}}

\newcommand{\pred}{\mathrm{pred}}

\newtheorem{theorem}{Theorem}[section]

\newaliascnt{lemma}{theorem}
\newtheorem{lemma}[lemma]{Lemma}
\aliascntresetthe{lemma}

\newaliascnt{informaltheorem}{theorem}
\newtheorem{informaltheorem}[informaltheorem]{Informal Theorem}
\aliascntresetthe{informaltheorem}

\newaliascnt{corollary}{theorem}
\newtheorem{corollary}[corollary]{Corollary}
\aliascntresetthe{corollary}

\newaliascnt{fact}{theorem}
\newtheorem{fact}[fact]{Fact}
\aliascntresetthe{fact}
\crefname{fact}{fact}{facts}
\Crefname{fact}{Fact}{Facts}

\newaliascnt{claim}{theorem}

\aliascntresetthe{claim}

\newaliascnt{definition}{theorem}
\newtheorem{definition}[definition]{Definition}
\aliascntresetthe{definition}

\newaliascnt{proposition}{theorem}
\newtheorem{proposition}[proposition]{Proposition}
\aliascntresetthe{proposition}

\newenvironment{proofsketch}
  {\begin{proof}[Proof sketch]}
  {\end{proof}}

\usepackage{thm-restate}
\usepackage{titletoc}

\AtBeginDocument{%
}
\usepackage{tikz}
\usetikzlibrary{decorations.pathreplacing}
\usepackage{subcaption}

\newcommand{\sanjay}[1]{{\color{brown}{\bfseries Sanjay: #1}}}

\newcommand{\aditya}[1]{{\color{red}{Aditya: #1}}}

\begin{document}

\pagestyle{plain}

\hypersetup{
    pdftitle={Online Batching with Lookahead: Adversarial Limits and Stochastic Gains},
    pdfauthor={Dimitrios Christou, Jiayi Chen, Nihal Sharma, Aditya Akella,
               Shuchi Chawla, and Sanjay Shakkottai}
}

\begin{center}

{\LARGE\bfseries
Online Batching with Lookahead:\\[0.15em]
Adversarial Limits and Stochastic Gains \footnote{This research was supported by the National Science Foundation under Grant No. 2326576 and by the InfraAI Center at the University of Texas at Austin.}
\par}

\vspace{1.25em}

{\large
\setlength{\tabcolsep}{4pt}
\begin{tabular}{ccc}
\begin{tabular}{c}
Dimitrios Christou\textsuperscript{$\dagger$}\\[-0.05em]
{\small\texttt{christou@cs.utexas.edu}}
\end{tabular}
&
\begin{tabular}{c}
Jiayi Chen\textsuperscript{$\dagger$}\\[-0.05em]
{\small\texttt{janechen@cs.utexas.edu}}
\end{tabular}
&
\begin{tabular}{c}
Nihal Sharma\textsuperscript{$\ddagger$}\\[-0.05em]
{\small\texttt{nihal.sharma@capitalone.com}}
\end{tabular}
\\[0.8em]
\begin{tabular}{c}
Aditya Akella\textsuperscript{$\dagger$}\\[-0.05em]
{\small\texttt{akella@cs.utexas.edu}}
\end{tabular}
&
\begin{tabular}{c}
Shuchi Chawla\textsuperscript{$\dagger$}\\[-0.05em]
{\small\texttt{shuchi@cs.utexas.edu}}
\end{tabular}
&
\begin{tabular}{c}
Sanjay Shakkottai\textsuperscript{$\dagger$}\\[-0.05em]
{\small\texttt{sanjay.shakkottai@utexas.edu}}
\end{tabular}
\end{tabular}
\par}

\vspace{1em}

{\normalsize
\textsuperscript{$\dagger$}The University of Texas at Austin 
\par}

\vspace{0.25em}

{\normalsize
\textsuperscript{$\ddagger$}Capital One 
\par}

\end{center}

\vspace{1.25em}

\begin{center}
\begin{minipage}{0.92\textwidth}

\noindent\textbf{Abstract.}
Modern inference systems rely on batching to exploit parallel execution
and amortize serial execution overheads, but batching creates a
fundamental online scheduling tradeoff: serving requests immediately
reduces request queuing delay, while waiting can form larger batches
that better utilize parallel hardware and improve serving efficiency.
We study this tradeoff in a single-server online batching model with
bounded batch capacity, fixed batch processing time, and the objective
of minimizing total flow time. Our focus is on the value of finite
lookahead into future arrivals over a purely online setting for
scheduling.

\vspace{3pt}

We show that the benefits of lookahead depend qualitatively on the
arrival model. In the adversarial setting, finite perfect lookahead has
inherent limitations: for any bounded lookahead window, every online
algorithm remains separated from the offline optimum by a constant
factor. Thus, worst-case arrival sequences can prevent perfect
knowledge of future events from fully overcoming online uncertainty.
In contrast, under Poisson arrivals in the stable regime, finite
lookahead is much more powerful. Although the best distribution-aware
online algorithm has a nonzero asymptotic gap from the offline optimum,
lookahead closes this gap exponentially fast in the window length. We
prove matching upper and lower bounds up to constants in the exponent.
Together, these results give a tight characterization of when lookahead
can close the online--offline gap, incorporating insights from the
competitive analysis of online algorithms into queueing-theoretic
steady-state bounds.

\vspace{3pt}

We also study imperfect predictions, showing that the benefit of
lookahead degrades gracefully with prediction error. Simulations on
synthetic arrivals and experiments on real request traces support the
theory: modest lookahead substantially reduces flow time and brings
online batching close to the offline benchmark.

%\textcolor{red}{To move: The paper was written and submitted on July 10, 2026. The authors used AI-based tools to improve the presentation of the manuscript, including polishing the writing and revising the organization, layout, and presentation order of proofs. The authors developed the underlying algorithmic ideas and constructed and sketched the proofs themselves; AI tools were additionally used to help present these arguments in a clearer and more polished form. The authors take full responsibility for the content of the paper.}

\end{minipage}
\end{center}
\thispagestyle{empty}

\newpage

\makeatletter
\begingroup
\let\addcontentsline\@gobblethree
\tableofcontents*
\endgroup
\makeatother

\thispagestyle{empty}
\newpage
\pagenumbering{arabic}
\section{Introduction}

% The design of algorithms based on the knowledge of future events has emerged as a powerful approach for going beyond the classical limitations of online algorithms. 
%The design of algorithms based on learned predictions has emerged as a powerful approach for going beyond the classical limitations of online algorithms.
%In many optimization problems, there is a fundamental gap between the performance of the offline optimum, which knows the entire input in advance, and that of an online algorithm, which observes the input over time and must make decisions without knowing the future. 
% Access to a ``lookahead'' into the future provides a way to partially bridge this gap: even imperfect information about short-term future inputs can help an algorithm avoid costly online decisions and substantially improve over purely online performance, sometimes approaching the hindsight optimum.
%Predictions provide a way to partially bridge this gap: even imperfect information about future inputs can help an algorithm avoid costly online decisions and substantially improve over purely online performance, sometimes approaching the hindsight optimum.

The design of algorithms with learned predictions~\cite{PurohitSvitkinaKumar2018,LykourisVassilvitskii2021,BamasMaggioriSvensson2020,MitzenmacherVassilvitskii2020} has emerged as a powerful approach for overcoming the classical limitations of online algorithms. In many optimization problems, there is a fundamental gap between the performance of the offline optimum, which knows the entire input in advance, and that of an online algorithm, which observes the input over time and must make irrevocable decisions without knowing the future. Predictions offer a way to partially bridge this gap: even limited information about future inputs can help an algorithm avoid costly online decisions and substantially improve over purely online performance, sometimes approaching the hindsight optimum.

In this paper, we study the power of predictions in a fundamental online batched scheduling problem. %~\cite{nvidiatritondynamicbatching}.
Jobs arrive over time and must be processed by a single non-preemptive server in batches of size at most $C$. Each batch, regardless of the number of jobs it contains, takes exactly $s$ time units to process. The objective is to minimize total flow-time, namely the total time jobs spend in the system between arrival and completion. The problem captures a basic tension in online request-processing systems: processing jobs immediately reduces delay, while waiting may allow the scheduler to form larger batches and amortize the fixed processing time of a batch~\cite{GfellerPeetersWeberWidmayer2009,ImMoseley2013,DoolyGoldmanScott2001,KarlinKenyonRandall2001}.

For request-processing problems such as this one, a particularly natural form of prediction is partial lookahead, a classical semi-online model in which the algorithm observes a bounded portion of the future input~\cite{LiYuanCaoBu2009,ImrehNemeth2007,DunkeNickel2016,DunkeNickel2021}. Rather than prescribing the algorithm’s actions, such predictions provide limited information about future arrivals, for example by revealing which jobs will arrive next. To isolate the algorithmic value of this information, we study the clean idealized model of {\bf perfect but finite lookahead}: at every point in time, the algorithm observes all arrivals in the next $L$ time units, in addition to the arrivals it has already seen. Our goal is to understand, through competitive analysis, how much this lookahead can close the gap between online and offline performance, and how its power depends on the structure of the input.

We begin our investigation in the context of adversarial arrivals. One might expect perfect lookahead to be most valuable precisely in this setting, where the online algorithm has no reliable structure in the input to exploit. Under this intuition, the competitive ratio should improve with the lookahead window $L$ and converge to $1$ as $L\to\infty$. This intuition is correct when the server has unbounded batch capacity. With bounded capacity, however, the picture is fundamentally different: we show that for every finite $L$, the competitive ratio of any online algorithm with lookahead remains bounded away from $1$. The reason is that the value of lookahead depends not only on knowing the near future, but also on the scheduler's ability to recover from decisions that are locally optimal within the next $L$ time units but globally suboptimal. In the bounded-capacity setting, such decisions can have persistent consequences from which the online algorithm may never fully recover.

We complement our theoretical results with an empirical study using real-world compute request traces drawn from Microsoft~\cite{shahrad2020serverless, zhang2021faster} and Alibaba~\cite{weng2022mlaas, yan2025Understanding}. For each trace, we evaluate a simple online algorithm with lookahead: at each decision point, the algorithm computes the hindsight-optimal schedule for the arrivals observed so far together with those revealed by the lookahead window, and then follows this schedule until new information becomes available. We find that even modest lookahead enables this simple algorithm to perform nearly optimally. In contrast, without lookahead, every online algorithm from a broad and natural class performs noticeably worse. These empirical findings stand in sharp contrast to our adversarial lower bound: on real-world traces, even limited lookahead can yield dramatic improvements over prediction-free online algorithms.

Motivated by these empirical observations, we ask whether a stochastic model of the input can account for the performance we see on real-world traces. At first glance, one might expect lookahead to be less valuable in such a setting. An algorithm that knows the input distribution can already optimize against the likely pattern of future arrivals, rather than against worst-case realizations. Lookahead into the realized sample path should therefore provide only limited additional information. Surprisingly, our theoretical results show the opposite.

We study a stochastic version of the model in which arrivals follow a Poisson process of rate $\lambda>0$. This places the problem within the classical literature on controlled bulk-service queues, where several customers may be served simultaneously and the server may wait to accumulate a larger batch~\cite{Bailey1954,DebSerfozo1973,YadinNaor1963}. We focus on the stable regime $\lambda s<C$, where the long-run arrival rate is smaller than the maximum processing rate of the server. In this regime, the expected offline optimum flow-time grows linearly with the number of arrivals, and the central question is whether online algorithms can match the offline asymptotic rate, and whether finite lookahead can close any remaining rate gap.

We first show that, without lookahead, a provable constant gap remains between the performance of any online algorithm and the offline optimum, even when the online algorithm knows the arrival distribution. In contrast, with a finite window of perfect lookahead, the algorithm’s performance converges to the offline optimum exponentially fast as the lookahead window length grows; we give matching upper and lower bounds on this convergence rate. Together, these results help explain the dramatic gains from lookahead observed in real-world traces.

%We also study a stochastic version of the model in which arrivals follow a Poisson process of rate $\lambda>0$. We focus on the stable regime $\lambda s<C$, where the long-run arrival rate is smaller than the maximum processing rate of the server. In this regime, the expected offline optimum flow-time grows linearly with the number of arrivals, and the central question is whether online algorithms can match the offline asymptotic rate, and whether finite lookahead can close any remaining rate gap.

Finally, while most of the paper studies the idealized model of perfect lookahead, we also consider a more realistic prediction model in which the algorithm receives an imperfect forecast of near-future arrivals. In this setting, the lookahead window may contain errors in the predicted arrival times. Following the learning-augmented framework, we seek consistency under accurate predictions, robustness under arbitrary predictions, and graceful degradation with prediction error~\cite{PurohitSvitkinaKumar2018,LykourisVassilvitskii2021,WeiZhang2020,ImMoseleyXuZhang2023}. We show that the benefits of lookahead are robust: even imperfect predictions can be used to obtain performance close to the offline optimum, with the guarantee degrading smoothly as the prediction error increases.

%Finally, although most of our focus is on the power of perfect lookahead, we show that our observations extend also to the more realistic setting of erroneous predictions of near-future arrivals. We show that imperfect lookahead can still be used to approximate the offline optimum, with performance degrading smoothly as prediction error increases.

To summarize, our results show that the value of lookahead for online batching is highly sensitive to the input model: gains from finite lookahead are inherently limited under adversarial arrivals with bounded batch capacity, yet even modest lookahead, perfect or erroneous, can yield near-optimal performance on real-world and stochastic inputs, with guarantees that degrade gracefully under prediction error. 

We now place the online batching model in context and describe our results and techniques.
%We consider a single-server batching problem. Jobs arrive over time at release times $r_1,\ldots,r_n$, and each job must eventually be assigned to a batch. The server is non-preemptive: once it starts processing a batch, it cannot interrupt it to add newly arriving jobs. Each batch can contain at most $B$ jobs and, regardless of its size, takes exactly $s$ time units to process. Equivalently, the model has a fixed per-batch processing/setup time and no additional per-job processing time. A schedule specifies when each batch is started and which available jobs it contains. If job $j$ is completed at time $C_j$, its flow-time is $C_j-r_j$, and the objective is to minimize total flow-time over all jobs, that is, $\sum_{j=1}^n (C_j-r_j)$. The offline optimum knows the entire arrival sequence in advance and chooses the schedule minimizing this objective. An online algorithm, by contrast, observes jobs only as they arrive and must decide when to start a batch and which currently available jobs to include, without knowing future arrivals. 

\paragraph{Request-level Batching: Motivation and Scope.}
Our model is motivated by the central tension of request-level batching in inference serving~\cite{nvidiatritondynamicbatching}. Modern inference systems batch requests to amortize fixed overheads, improve accelerator utilization, and reduce per-request compute cost. The scheduler must therefore decide when to stop waiting and launch a batch: launching too early sacrifices batching efficiency, while waiting too long increases latency for requests already in the queue. This tension appears even in homogeneous, non-generative inference workloads, such as image classification, embedding generation, ranking, and other fixed-cost neural inference tasks. These workloads are often served on GPUs or other accelerators, where a batch incurs a fixed launch or setup cost and is constrained by a maximum feasible batch size determined by memory, tensor shape compatibility, or latency targets.

We study this tradeoff using the stylized online batching model described above. The model isolates the value of future arrival information for the batching decision itself, abstracting away multi-replica routing, heterogeneous models, request-dependent service times, and other system-level features. It is therefore most directly applicable to homogeneous non-LLM inference, where requests in a scheduling class have comparable service demands and the main decision is when to launch a batch. LLM serving introduces additional structure: autoregressive generation has separate prefill and decode phases, continuous batching, variable output lengths, and a KV cache that stores intermediate attention states for previously processed tokens and therefore creates sequence-dependent memory constraints. Modeling these features requires a richer state space and is outside the scope of this paper.

In this model, finite lookahead captures settings in which short-term demand forecasts are available, for example from recent traffic patterns, application-level structure, or learned predictors, while long-horizon predictions remain unreliable. In the idealized model the lookahead information is exact; in our prediction-error model it is an imperfect forecast of near-future arrivals. The window length $L$ is thus the natural resource: larger $L$ gives the scheduler more information about whether to launch a batch immediately or wait for imminent arrivals, but still leaves uncertainty about the rest of the input.

\paragraph{Our Results.}
We evaluate online algorithms, with and without lookahead, using competitive analysis. An algorithm is $\alpha$-competitive if, on every input sequence, its total flow-time is at most $\alpha$ times that of the offline optimum. Equivalently, its average per-job flow-time is within the same factor of the minimum possible average per-job flow-time. The offline optimum has competitive ratio $1$ by definition, and ratios closer to $1$ indicate better performance. Our first result characterizes the power and limitations of finite lookahead in the adversarial model.

\begin{informaltheorem}[Adversarial arrivals]
For adversarial arrivals, purely online batched scheduling admits a constant competitive ratio bounded away from $1$ (\Cref{thm:adv-online-bounds}). Perfect finite lookahead is substantially more powerful when batch capacity is unbounded: for $C=\infty$, the optimal competitive ratio with lookahead $L$ converges polynomially to $1$ as $L$ grows (\Cref{thm:adv-lookahead-uncap-upper,thm:adv-lookahead-uncap-lower}). With bounded capacity, however, this convergence fails. For every finite capacity $C<\infty$ and every finite lookahead window $L<\infty$, every $L$-lookahead algorithm remains separated from the offline optimum by a factor of $1+\Omega(1)$ (\Cref{thm:adv-lookahead-cap-lb}). 
\end{informaltheorem}

%\begin{informaltheorem}[Adversarial arrivals]
%For adversarial arrivals, purely-online batching admits a constant competitive ratio of $1+\Theta(1)$. Perfect finite lookahead helps substantially when capacity is unbounded: for $C=\infty$, the optimal competitive ratio with lookahead $L$ is $1+\Theta\left(\mathrm{poly}(s/L)\right)$ and approaches $1$ as $L$ grows. However, this phenomenon disappears under bounded capacity. For every finite capacity $C<\infty$ and every finite lookahead window $L<\infty$, every $L$-lookahead algorithm remains separated from the offline optimum by a $1+\Omega(1)$ factor.
%\end{informaltheorem}

%We prove the adversarial online bounds in \Cref{thm:adv-online-bounds}, the uncapacitated lookahead upper and lower bounds in \Cref{thm:adv-lookahead-uncap-upper,thm:adv-lookahead-uncap-lower} respectively, and the bounded-capacity lower bound in \Cref{thm:adv-lookahead-cap-lb}. 
Together, these results show that batch capacity fundamentally changes the value of lookahead. When $C=\infty$, a lookahead scheduler can repair local mistakes by absorbing leftover jobs into a large future batch. When $C<\infty$, this repair mechanism breaks down: an early incorrect batching decision can leave the scheduler permanently out of phase with future full batches. Thus, in the worst case, no finite amount of lookahead suffices to recover the offline optimum.

By contrast, in the Poisson model, finite lookahead is enough to close the offline-optimal gap exponentially fast. 

\begin{informaltheorem}[Poisson arrivals]
For Poisson arrivals of rate $\lambda$, the per-job flow-time of the offline-optimal scheduler is finite exactly in the stable regime $\lambda s<C$ (\Cref{sec:stoch-rate}). In this regime, every online policy remains separated from the offline-optimal 
flow-time by a competitive ratio of $1+\Theta_{\lambda,s,\capac}(1)$ (\Cref{thm:online-gap}). 
%rate by a positive constant $\Theta_{\lambda,s,B}(1)$. 
Perfect finite lookahead closes this stochastic gap exponentially fast (\Cref{thm:poisson-upper}): there is an $L$-lookahead scheduler with competitive ratio at most
$$
%\rate(\sched^*)
1+\exp(-\Theta_{\lambda,s,\capac}(L)).
$$
Moreover, this exponential dependence on $L$ is tight up to constants in the exponent (\Cref{thm:poisson-lb}).
\end{informaltheorem}

%We establish the rate definition and the stability threshold in \Cref{sec:stoch-rate}. The online--offline separation is proven in \Cref{thm:online-gap}, and the lookahead upper bound is in \Cref{thm:poisson-upper}. The matching exponential lower bound is proved in \Cref{thm:poisson-lb}. 

The contrast between the two settings is the main message of the paper: under adversarial arrivals, finite lookahead provides only local information and bounded capacity can lead to irreversible local mistakes; under Poisson arrivals, memorylessness creates regeneration opportunities, and the probability of failing to find one inside the lookahead window decays exponentially in $L$.

Next, we treat the case of imperfect predictions where the instance observed through the lookahead window doesn't correspond to the true sequence of arrivals. We assume that the true and predicted instances contain the same number of jobs and measure prediction error by the maximum discrepancy $\epsilon$ between rank-matched arrival times. We show that total flow-time is Lipschitz with respect to this error: any schedule computed for the predicted instance can be transferred to the true instance at an additional cost of at most $O(\epsilon)$ per job. Using this property, we extend our perfect-lookahead upper bounds to erroneous predictions, incurring only an additive $O(\epsilon)$ loss in the competitive ratio, while retaining a constant robust bound~(\Cref{thm:imperfect-prediction}).

Finally, simulations on synthetic and real-world traces validate our theoretical findings across a range of practical settings (see \Cref{sec:eval}). We also note that all of the algorithms we design can be implemented efficiently. 
%with respect to the size of the job arrival sequence.}

%\paragraph{Our Techniques.}
%A key technical contribution of our work is to combine algorithmic design with queueing-theoretic steady-state analysis in order to obtain tight competitive-ratio bounds. This is especially important in the stochastic setting, where the steady-state behavior of an online algorithm can be difficult to analyze: the algorithm may carry state across time, and its decisions may depend on the entire history of arrivals and past actions.

\paragraph{Our Techniques.}
A key technical contribution of our work is to adapt the tools of competitive analysis to a setting where performance is governed by long-run stochastic behavior. Our upper bounds require designing online algorithms whose decisions are robust to worst-case uncertainty within the lookahead model, while our lower bounds require constructing local input patterns that force a gap relative to the offline optimum. In the stochastic setting, turning these local comparisons into asymptotic competitive-ratio bounds requires controlling the steady-state behavior of the algorithm. This is challenging because online algorithms may carry state across time and make decisions based on the entire history of arrivals and past actions.

For our upper bounds, we control this history dependence by designing algorithms whose behavior regenerates at carefully chosen times. In the stochastic setting with lookahead, for example, the algorithm uses only a prefix of the lookahead window, stopping at the first point at which the system can be certified to become idle for a short interval. We show that, with high probability, such a prefix exists and the algorithm can schedule it optimally. This regeneration structure allows us to analyze performance cycle by cycle and derive the desired competitive-ratio bounds.

Our lower bounds require different arguments. In the adversarial setting, we show that any local ``mistake'' made by the online algorithm can be amplified into a constant performance gap by keeping the server fully loaded for an arbitrarily long time. The stochastic setting is more delicate: although we cannot prescribe a single bad arrival sequence, we identify local arrival patterns that create a gap between the online algorithm and the offline optimum, and occur with positive probability in the regenerative structure of the process. In particular, in the absence of lookahead, we use the memorylessness of Poisson arrivals to reduce the analysis of online algorithms to a single busy cycle of some threshold policy. On a positive-probability event within such a cycle, the online schedule can be locally replaced by a better one that completes no later than the original cycle endpoint. These local improvements concatenate into a feasible global schedule, and the renewal-reward theorem turns constant improvements per cycle into a positive gap between the online and offline asymptotic rates.

For the lookahead results, the central idea is to decompose the offline optimum itself into regenerative pieces. We define offline cycles that begin with an arrival into an empty system and end once the jobs in the cycle can be scheduled optimally and cleared before the next arrival. This yields a cycle-by-cycle benchmark that a lookahead algorithm can attempt to imitate. The upper bound shows that, except with probability exponentially small in $L$, the lookahead window contains enough of the current offline cycle for the algorithm to recognize the cycle and execute its optimal schedule. The lower bound shows that this exponential dependence is inherent: with exponentially small but unavoidable probability, the relevant offline cycle extends beyond the lookahead window, leaving any finite-lookahead algorithm unable to distinguish between continuations that require different optimal decisions. Thus, the same offline decomposition both powers the algorithm and identifies the rare events that limit its convergence rate.

\subsection{Related Work}

Finite lookahead has been studied widely in online optimization
\cite{DunkeNickel2016,DunkeNickel2021} with applications in bin packing, routing, packet scheduling, caching and several other queuing environments~\cite{Grove1995,JailletWagner2006,BohmChrobakJezLiSgallVesely2016,MotwaniSaraswatTorng1998,ZhengXuZhang2008,MandelbaumShabtay2011,ZhengChengLiuXu2013} to model side information about future inputs provided to the learner. Particularly relevant are lookahead algorithms for batch and parallel-batch scheduling~\cite{LiYuanCaoBu2009,LiZhangYang2012,LiYuanYang2014,DwibedyMohanty2023,Jiao2024} and for online data acknowledgment~\cite{ImrehNemeth2007} that address the tradeoff between idling to batch jobs together and processing jobs as they arrive. While this literature primarily studies worst-case competitive performance~\cite{SleatorTarjan1985,BorodinElYaniv1998,Albers2003}, we study both the adversarial and stochastic settings with the goal of comparing the benefit from lookahead in the two. 

%Competitive analysis is the traditional framework for evaluating algorithms that make irrevocable decisions without knowing future inputs. 
Online scheduling has a rich history, with algorithms designed to optimize metrics such as makespan~\cite{Graham1966}, total or weighted completion time~\cite{HallSchulzShmoysWein1997}, and total or weighted flow-time~\cite{PhillipsSteinTorngWein1997,ChekuriKhannaZhu2001,BansalPruhs2014}; see also the survey~\cite{PruhsSgallTorng2004}. Batching~\cite{PottsKovalyov2000,FowlerMonch2022,TianFuYuan2014} introduces a further layer of complexity in vanilla scheduling but is applicable to modern serving systems that can process jobs in parallel. The works closest to our adversarial model study online single-machine batching with release times and total flow-time~\cite{GfellerPeetersWeberWidmayer2009} and online batch scheduling under several flow-time objectives~\cite{ImMoseley2013}; related bounded-batch models consider weighted completion-time objectives~\cite{ChenDengZang2004,MaWanWeiYuan2014}. Another closely related  problem is that of online aggregated TCP acknowledgment~\cite{DoolyGoldmanScott2001,KarlinKenyonRandall2001}, where packets accumulate delay until the algorithm sends a fixed-cost acknowledgment covering all pending packets. Variants incorporate delay penalties, communication constraints, offline optimization, and general delay functions~\cite{AlbersBals2005,FrederiksenLarsenNogaUthaisombut2003,NogaSeidenWoeginger2002,BhorePawlowskiUmboh2026}. Both TCP acknowledgment and batch scheduling trade-off immediate service against waiting to aggregate additional requests. In the standard acknowledgment model, however, actions are instantaneous and uncapacitated, whereas our batches have finite capacity and occupy the server for a nonzero processing interval.

Another related area is learning-augmented online algorithms~\cite{ALPS,MitzenmacherVassilvitskii2020}, which use predictions to obtain near-optimal performance when forecasts are accurate while retaining worst-case guarantees as prediction error increases. This framework has been applied to scheduling with predictions of job weights, processing loads, and job orders~\cite{LattanziLavastidaMoseleyVassilvitskii2020,ImKumarQaemPurohit2023,LindermayrMegow2025}, as well as to interval scheduling and dynamic TCP acknowledgment~\cite{BoyarFavrholdtKamaliLarsen2026,ImMoseleyXuZhang2023}. Most of this literature is framed in the adversarial-input model, where predictions are used to strengthen worst-case guarantees. Our work asks how such guarantees behave in both adversarial and stochastic arrival models. We show that perfect finite lookahead has sharply different effects in the two settings: under adversarial arrivals, every finite lookahead window leaves a constant gap from the offline optimum, whereas under Poisson arrivals this gap shrinks exponentially fast with the window length. This contrast places our work in the broader beyond-worst-case perspective~\cite{Roughgarden2020}.

Finally, the literature on the control of queueing systems studies dynamic service decisions, such as when to initiate service and how many waiting jobs to process in each batch~\cite{YadinNaor1963,DebSerfozo1973}. These problems are typically formulated as average-cost Markov or semi-Markov decision processes, yielding optimal causal policies under a known arrival distribution~\cite{FedergruenHordijkTijms1979,Puterman1994,Sennott1999}. Our perspective is different: after characterizing the optimal causal policy, we compare with a clairvoyant sample-path optimum and quantify how finite lookahead into the realized arrival process reduces the resulting steady-state optimality gap.

\section{The Adversarial Batching Model}\label{sec:adversarial}

In this section we study our online batching problem under adversarial arrivals. An instance $\inst$ consists of an arbitrary, potentially adversarial sequence of $n$ release times $r_1\le \cdots \le r_n$. A scheduler $\sched$ must process these jobs in batches subject to the following constraints: job $j$ may only be processed after time $r_j$; each batch may contain at most $\capac$ jobs; every nonempty batch, regardless of its size, requires some fixed processing time $s>0$; preemption is not allowed; and all jobs in a batch complete simultaneously. 

Let $T_\sched (j)$ denote the completion time of job $j$ under a scheduling algorithm $\sched$. We measure performance in terms of the total {\em flow-time} of $\sched$ on the instance $\inst$, which we denote as
$$
\sched(\inst)
:=
\sum_{j=1}^n (T_\sched(j)-r_j).
$$
If the scheduler $\sched$ uses internal randomness, we denote with $\sched(\inst)$ the expected flow-time of $\sched$ on $\inst$. Our benchmark is the offline optimum: a scheduler that observes the entire instance $\inst$ in advance and then chooses a feasible batching schedule minimizing total flow-time. We denote this value by
$$
\opt(\inst)
:=
\min_{\sched}\sched(\inst).
$$

We evaluate the performance of a scheduling algorithm using the tools of competitive analysis. We say that a scheduler $\sched$ is {\it $\alpha$-competitive} if, for every adversarial instance $\inst$, its total flow-time satisfies $\sched(\inst)\le \alpha\cdot \opt(\inst)$.

\subsection{Structural Properties of the Offline Optimum}\label{sec:adv-basic-properties}

We begin by recording two key properties of the offline optimum that will be useful throughout this paper. First, we show that the optimal flow-time objective is super-additive: if an instance is partitioned into two sub-instances, then restricting any feasible schedule to each part separately gives feasible schedules for the two parts, so the optimal cost of the full instance must be at least the sum of the two separate optima.

\begin{restatable}{lemma}{superadditivelemma}
\label{lemma:superadditive}
Let $\inst$ be any instance, and let $\inst_1,\inst_2$ be any partition of the jobs in $\inst$ into two disjoint sub-instances such that $\inst_1\cup \inst_2=\inst$ and $\inst_1\cap\inst_2=\emptyset$. Then, $\opt(\inst) \ge \opt(\inst_1)+\opt(\inst_2)$.
\end{restatable}

The above statement holds path-wise, and therefore if instance $\inst$ is random it holds with probability one over every instantiation of $\inst$ and every partition of that instantiation. We next adapt the standard notion of work-conservation for our batched scheduling model.

\begin{definition}[Work conservation]\label{def:work-conserving}
A schedule $\sched$ is \emph{work-conserving} if it satisfies the following two properties: (1) whenever $\sched$ starts a batch while $q$ jobs are pending, the batch contains $\min\{q,\capac\}$ jobs; and (2) whenever the processor is idle and at least $\capac$ jobs are pending, $\sched$ starts a batch immediately.
\end{definition}

Since every nonempty batch has the same processing time $s$, work-conservation is without loss when minimizing flow-time: an optimal schedule never benefits from leaving capacity unused in a started batch, nor from idling when a full batch is already available. In other words, every flow-time optimal schedule is work-conserving. Work-conservation also rules out artificial delays that could make a schedule arbitrarily slow to finish without improving flow-time. This lets us compare the final completion time (a.k.a. the \textit{makespan}) of a flow-time optimum to that of any other feasible schedule.

\begin{restatable}{lemma}{makespanlemma}\label{lemma:opt-makespan}
Let $\inst$ be any instance, let $\sched^*$ be any schedule minimizing flow-time for $\inst$, and let $\sched$ be any feasible schedule for $\inst$. Then,  
$$
    \max_{j=1,\ldots,n} T_{\sched^*}(j)
    \le 
    \max_{j=1,\ldots,n} T_{\sched}(j) + 2s . 
$$
In particular, the offline optimum completes the last job in $\inst$ within $2s$ of any other schedule.
\end{restatable}

We defer the proofs of~\Cref{lemma:superadditive,lemma:opt-makespan} to Appendix~\ref{app:superadditive} and~\ref{app:makespan} respectively.

\subsection{Online Algorithms for Adversarial Batching}
Our primary interest is in online schedulers, which must make decisions without observing the entire job sequence. At any time $t$, an online scheduler knows only the arrivals that have occurred up to time $t$ and its own past decisions; it has no information about future release times. The batching problem captures a fundamental online tradeoff. Starting a batch immediately reduces the waiting time of currently pending jobs, but may be wasteful if many additional jobs arrive soon after. Conversely, waiting may allow the scheduler to form larger batches and amortize the setup time, but increases the flow-time of jobs already in the queue. In the adversarial model, as we show, this tradeoff cannot be resolved perfectly: online schedulers can achieve constant competitive ratios, but no online scheduler can match the offline optimum on all instances.

We begin by analyzing what is arguably the simplest online algorithm. Let $\greedy$ denote the greedy scheduler, which starts a batch whenever the processor is idle and at least one job is pending; the batch includes as many pending jobs as possible, up to the capacity $C$. We show that the greedy scheduler can never fall too far behind an optimal schedule: in the worst case, it requires only an additional $s$ time units to catch up to the optimum in terms of the number of jobs processed. This implies that greedy incurs at most twice the flow-time of the optimal schedule. A similar technique was used in \cite{GfellerPeetersWeberWidmayer2009} to show that the greedy scheduler is $2$-competitive in a related setting where processing time is linear in the batch size, but batch size is unbounded. 

We further show that greedy is optimal among all deterministic schedulers in our setting, and we prove constant lower bounds even for randomized schedulers. The intuition is that once an online algorithm falls behind the optimal schedule, an adversary can force it to schedule at capacity for an arbitrarily long interval, leaving it no opportunity to recover and match the queue state of the optimum. These results are summarized below as \Cref{thm:adv-online-bounds}; a formal proof is presented in Appendix~\ref{appendix:adv-online-bounds}. We assume throughout that $C\ge 2$, since the case $C=1$ is degenerate: every batch contains a single job, so the greedy scheduler is offline optimal. 

%This scheduler was shown to be $2$-competitive by \textcolor{blue}{XX} in the uncapacitated setting ($B=\infty$), even under the more challenging setting where each job has an individual processing time $p>0$ and batching $m$ jobs takes time $pm+s$. We adapt their arguments to show that the same is true in our simplified setting (corresponding to $p=0$) even in the presence of the batching capacity $B$. Furthermore, the greedy algorithm is optimal across all deterministic schedulers for our setting, with constant lower bounds holding even for randomized schedulers.

\begin{restatable}{theorem}{adversarialonline}
\label{thm:adv-online-bounds}
Assume $\capac\ge 2$. The greedy scheduler $\greedy$ is $2$-competitive for adversarial online batching. Furthermore, for any $\epsilon>0$, no deterministic scheduler is $(2-\epsilon)$-competitive and no randomized scheduler is $(11/9-\epsilon)$-competitive. 
\end{restatable}

%The main message is that the adversarial model admits constant-competitive online scheduling, but also that purely online information is intrinsically insufficient to recover or even approach the offline optimum to arbitrary precision.

\subsection{Schedulers with Perfect Lookahead}

We next consider algorithms with partial information about the future. The prediction model we study is \emph{time lookahead}: at every time $t$, the scheduler knows all arrivals in the interval $[0,t+L]$, where $L\ge 0$ is the lookahead window. Thus $L=0$ recovers the standard online model, while $L\to\infty$ recovers the offline optimum. Throughout this section, we assume that the predictions are perfect: the scheduler observes the true arrivals in every lookahead window.

\subsubsection{The Uncapacitated Setting}

We begin with the uncapacitated case, where $C=\infty$. In this setting, finite lookahead can be extremely powerful. Intuitively, because batch capacity is unbounded, an online algorithm can transition from one queue state to a more desirable one while incurring only $O(s)$ additional delay per job. Thus, after an early suboptimal decision, an adversary cannot force the algorithm to remain in a bad state indefinitely. To fully exploit this flexibility, however, the algorithm must use lookahead carefully so as to avoid repeatedly making suboptimal decisions.

We propose an algorithm, called $\advname$ (\Cref{alg:adv_uncap_lookahead}), that partitions the time horizon into phases of length $L+2s$. The start times of these phases are shifted by a random amount chosen uniformly from a short interval. Each phase is further divided into active jobs, which arrive during the first $L$ time units, and ignored jobs, which arrive during the final $2s$ time units. At the beginning of each phase, $\advname$ uses its lookahead to compute an offline optimal schedule for the active jobs and then executes that schedule; ignored jobs are scheduled with the first batch of the following phase. Since batch capacity is unbounded, this schedule necessarily completes within the same phase: in the worst case, the algorithm starts a batch just before the last active arrival, processes it for $s$ time units, and then processes any remaining active jobs together in one additional batch of length $s$. Consequently, the schedules executed in different phases do not interfere with one another.

\begin{comment}
    
\end{comment}
\begin{algorithm}[h]
\caption{$\advname(L)$}
\label{alg:adv_uncap_lookahead}
\begin{algorithmic}[1]
%\Require Lookahead window $L>0$
\State Sample an offset $\tau\sim \mathrm{Unif}[0,L+2s]$.
\State Define active intervals $A_0:=[0,(\tau-2s)^+]$ and $A_{\ell+1}:=[\tau+\ell(L+2s),\ \tau+\ell(L+2s)+L]$ $\;\forall\ell\ge 0$.
\State Define the ignored intervals
$ I_\ell :=  (\tau-2s + \ell (L+2s), \tau + \ell(L+2s))$ for all $\ell \ge 0$.
   
\For{$\ell=0,1,2,\ldots$}
    \State At the beginning of $A_\ell$, use the $L$-lookahead window to observe all arrivals in $A_\ell$.
    \State Compute an offline-optimal schedule for the sub-instance induced by jobs released in $A_\ell$.
    \State Execute this offline-optimal schedule, without changing any of its batching times or sizes.
    \State Process each ignored job in the first batch after its arrival defined by these schedules. If no batch starts within the first $s$ units of the next active interval, start a fresh one at its beginning.
\EndFor
\end{algorithmic}
\end{algorithm}

Hence, our algorithm processes all active jobs optimally. Moreover, because of the random phase shift, the probability that any given job is ignored, and hence not processed optimally, is precisely $2s/(L+2s)$. Thus, if we can obtain a constant approximation for the ignored jobs without interfering with the schedule already computed for the active jobs, then the competitive ratio of $\advname$ is $1+O(s/L)$. This is indeed the case. Each ignored job waits at most $2s$ time units for its ignored window to end and at most an additional $s$ time units for the first active batch to start; if no such batch starts in time, the algorithm inserts one. Thus each ignored job incurs flow-time at most $4s$. Since the offline optimum also incurs flow-time at least $s$ for each ignored job, the desired bound follows. This yields the following guarantee, proved formally in Appendix~\ref{appendix:adv-lookahead-upper}.

\begin{restatable}{theorem}{randomshift}
\label{thm:adv-lookahead-uncap-upper}
Assume $C=\infty$. For every $L\ge 0$, $\advname(L)$ is $\alpha(L)$-competitive for
\[\alpha(L):=1 + \frac{6s}{2s+L}.\]
\end{restatable}

\Cref{thm:adv-lookahead-uncap-upper} demonstrates that in the uncapacitated setting, perfect lookahead suffices to overcome the information-theoretic lower bounds of~\Cref{thm:adv-online-bounds}. In particular, it shows that for any $\rho>0$, one can achieve a competitive ratio of at most $1+\rho$ provided a lookahead window of $L \approx s/\rho$. We complement this result by showing that a polynomial dependency between the lookahead window and the competitive ratio is necessary, via the following lower bound, proven in Appendix~\ref{appendix:adv-uncap-lookahead-lb}. 

\begin{restatable}{theorem}{advuncaplb}
\label{thm:adv-lookahead-uncap-lower}
Assume $C=\infty$ and $L\ge s$. For every $L$-lookahead scheduler $\sched$, there exists an adversarial instance $\inst$ such that $\sched(\inst)\ge \beta(L)\cdot \opt(\inst)$  for
$$
\beta(L):= 1 + \Omega\bigg( \frac{s^2}{L^2}\bigg).
$$
\end{restatable}

\subsubsection{The Capacitated Setting}
We now turn to the capacitated setting, where each batch may contain at most $C$ jobs. The analysis of $\advname$ already indicates why capacity changes the picture. In the uncapacitated case, the scheduler can optimize over the jobs visible in its lookahead window, and any jobs left between consecutive windows can later be absorbed into a single large stitching batch. Thus boundary mistakes can be repaired locally. With finite capacity, this repair mechanism breaks down: once the scheduler commits capacity to a batch, it may be impossible to undo that decision when additional jobs arrive just beyond the lookahead horizon. Recovering from such a mistake may require several future batches, creating persistent flow-time loss. As it turns out, this obstruction is not an artifact of $\advname$; it is inherent to the capacitated problem. Specifically, we show that even perfect finite lookahead does not eliminate the online--offline separation. 

\begin{restatable}{theorem}{advcaplb}\label{thm:adv-lookahead-cap-lb}
    Fix any finite capacity $C\ge 2$ and any finite lookahead $L\ge 0$. For any $L$-lookahead scheduler $\sched$ and any $\epsilon>0$, there exists an adversarial instance $\inst$ such that $\sched(\inst)\ge (\alpha_C - \epsilon)\cdot \opt(\inst)$ where
    \[
        \alpha_C
        =
        1+\frac{\lfloor C/2\rfloor}{C}
        \quad\text{if $\;\sched$ is deterministic,}
    \]
    and
    \[
        \alpha_C
        =
        1+\frac{\lfloor C/2\rfloor\lceil C/2\rceil}{C^2}
        \quad\text{if $\;\sched$ is randomized.}
    \]
    In particular, when $C$ is even, these lower bounds are $3/2-\epsilon$ for deterministic schedulers and $5/4-\epsilon$ for randomized schedulers.
%    $\alpha=1.5$ if $\sched$ is deterministic and $\alpha=1.25$ is $\sched$ randomized.
\end{restatable}

We include a proof sketch for the theorem here. The formal proof is deferred to Appendix~\ref{sec:adv-cap-lb}.

% Preamble:
% \usepackage{tikz}
% \usepackage{graphicx}
% \usetikzlibrary{decorations.pathreplacing}

\begin{figure}[h]
\centering
\resizebox{0.97\linewidth}{!}{%
\begin{tikzpicture}[
    x=1cm,
    y=1cm,
    tick/.style={black, line width=0.8pt},
    axis/.style={black, line width=0.8pt},
    arr/.style={font=\small},
    timelab/.style={font=\small},
    title/.style={font=\large},
    prefixstyle/.style={draw=red!70, text=red!70},
    early/.style={draw=green!55!black, text=green!55!black},
    late/.style={draw=blue!70, text=blue!70},
    brace/.style={
        decorate,
        decoration={brace, amplitude=5pt},
        line width=0.8pt
    }
]

% Optional light grid, to mimic the hand-drawn background.
% Comment out these two lines if you want a clean white background.
%\draw[step=0.75cm, very thin, blue!12] (-0.4,-4.6) grid (13.7,1.9);

% Common x-coordinates
\coordinate (x0)   at (1.0,0);
\coordinate (xd)   at (1.6,0);
\coordinate (xs)   at (2.4,0);
\coordinate (x2s)  at (3.5,0);
\coordinate (xmid) at (5.0,0);
\coordinate (xks)  at (7.0,0);

% Top timeline coordinates
\coordinate (xkp1) at (8.1,0);
\coordinate (xkp2) at (9.0,0);
\coordinate (xkM)  at (11.4,0);
\coordinate (xkM1) at (12.6,0);

% Bottom timeline coordinates
\coordinate (xk1d) at (8.4,0);
\coordinate (xk2d) at (9.3,0);
\coordinate (xkMd) at (12.0,0);

% y-coordinates
\def\yE{0}
\def\yL{-3.0}

%%%%%%%%%%%%%%%%%%%%%%%%%%%%%%%%%%%%%%%%%%%%%%%%%%%%%%%%%%%%
% Early instance
%%%%%%%%%%%%%%%%%%%%%%%%%%%%%%%%%%%%%%%%%%%%%%%%%%%%%%%%%%%%

\node[title, anchor=west] at (-1.0,{\yE+1.75})
    {\small The early instance $\inst_{\mathrm E}$:};

% Axis
\draw[axis] (0.7,\yE) -- (13.1,\yE);

% Left labels
\node[arr, red!70, anchor=east] at (0.65,{\yE+0.36}) {arrivals};
\node[timelab, anchor=east] at (0.65,{\yE-0.35}) {time};

% Ticks and time labels
\foreach \x/\lab in {
    1.0/{\scriptsize $0$},
    1.3/{\scriptsize $\delta s$},
    2.0/{\scriptsize $s$},
    3.0/{\scriptsize $2s$},
    7.0/{\scriptsize $ks$},
    8.0/{\scriptsize $(k+1)s$},
    9.0/{\scriptsize $(k+2)s$},
    11.0/{\tiny $(k+M)s$},
    12.1/{\tiny $(k+M+1)s$}
}{
    \draw[tick] (\x,{\yE-0.12}) -- (\x,{\yE+0.12});
    \node[timelab, below=3pt] at (\x,\yE) {\lab};
}

% Ellipses
\node at (5.0,{\yE+0.22}) {\large $\cdots$};
\node at (10.1,{\yE+0.22}) {\large $\cdots$};

% Arrival labels: prefix
\node[arr, red!70] at (1.0,{\yE+0.43}) {\tiny $q\quad$};
\node[arr, red!70] at (1.3,{\yE+0.43}) {\tiny $C-q$};
\node[arr, red!70] at (2.0,{\yE+0.43}) {\scriptsize $C$};
\node[arr, red!70] at (3.0,{\yE+0.43}) {\scriptsize $C$};
\node[arr, red!70] at (7.0,{\yE+0.43}) {\scriptsize $C$};

% Arrival labels: early continuation
\node[arr, green!55!black] at (8.0,{\yE+0.43}) {\scriptsize $q$};
\node[arr, green!55!black] at (9.0,{\yE+0.43}) {\scriptsize $C$};
\node[arr, green!55!black] at (11.0,{\yE+0.43}) {\scriptsize $C$};
\node[arr, green!55!black] at (12.0,{\yE+0.43}) {\scriptsize $C$};

% Braces
\draw[brace, prefixstyle]
    (1.0,{\yE+0.82}) -- (7.0,{\yE+0.82})
    node[midway, above=6pt] {$\mathcal I_{\mathrm{prefix}}$};

\draw[brace, early]
    (8.0,{\yE+0.82}) -- (12.0,{\yE+0.82})
    node[midway, above=6pt] {$\mathcal I_{\mathrm{early}}$};

%%%%%%%%%%%%%%%%%%%%%%%%%%%%%%%%%%%%%%%%%%%%%%%%%%%%%%%%%%%%
% Late instance
%%%%%%%%%%%%%%%%%%%%%%%%%%%%%%%%%%%%%%%%%%%%%%%%%%%%%%%%%%%%

\node[title, anchor=west] at (-1.0,{\yL+1.75})
    {\small The late instance $\inst_{\mathrm L}$:};

% Axis
\draw[axis] (0.7,\yL) -- (13.1,\yL);

% Left labels
\node[arr, red!70, anchor=east] at (0.65,{\yL+0.36}) {arrivals};
\node[timelab, anchor=east] at (0.65,{\yL-0.35}) {time};

% Ticks and time labels
\foreach \x/\lab in {
    1.0/{\scriptsize $0$},
    1.3/{\scriptsize $\delta s$},
    2.0/{\scriptsize $s$},
    3.0/{\scriptsize $2s$},
    7.0/{\scriptsize $ks$},
    8.3/{\scriptsize $(k+1+\delta)s$},
    11.3/{\scriptsize $(k+M+\delta)s$}
}{
    \draw[tick] (\x,{\yL-0.12}) -- (\x,{\yL+0.12});
    \node[timelab, below=3pt] at (\x,\yL) {\lab};
}

% Additional unlabeled burst tick
\draw[tick] (9.3,{\yL-0.12}) -- (9.3,{\yL+0.12});

% Ellipses
\node at (5.0,{\yL+0.22}) {\large $\cdots$};
\node at (10.4,{\yL+0.22}) {\large $\cdots$};

% Arrival labels: prefix
\node[arr, red!70] at (1.0,{\yL+0.43}) {\tiny $q\quad$};
\node[arr, red!70] at (1.3,{\yL+0.43}) {\tiny $C-q$};
\node[arr, red!70] at (2.0,{\yL+0.43}) {\scriptsize $C$};
\node[arr, red!70] at (3.0,{\yL+0.43}) {\scriptsize $C$};
\node[arr, red!70] at (7.0,{\yL+0.43}) {\scriptsize $C$};

% Arrival labels: late continuation
\node[arr, blue!70] at (8.3,{\yL+0.43}) {\scriptsize $C$};
\node[arr, blue!70] at (9.3,{\yL+0.43}) {\scriptsize $C$};
\node[arr, blue!70] at (11.3,{\yL+0.43}) {\scriptsize $C$};

% Braces
\draw[brace, prefixstyle]
    (1.0,{\yL+0.82}) -- (7.0,{\yL+0.82})
    node[midway, above=6pt] {$\mathcal I_{\mathrm{prefix}}$};

\draw[brace, late]
    (8.3,{\yL+0.82}) -- (11.3,{\yL+0.82})
    node[midway, above=6pt] {$\mathcal I_{\mathrm{late}}$};

\end{tikzpicture}%
}
\caption{The two instances in the construction of~\Cref{thm:adv-lookahead-cap-lb}.}
\label{fig:early-late-instances}
\end{figure}

\begin{proofsketch}
Fix $k=\lceil L/s\rceil$ and consider the two instances illustrated in
\Cref{fig:early-late-instances}. The two instances share a common prefix:
$q$ jobs arrive at time $0$, $C-q$ jobs arrive at time $\delta s$, and then
$C$ jobs arrive at each time $s,2s,\ldots,ks$. They differ only after this
prefix. In the early instance $\inst_{\mathrm E}$, $q$ additional jobs arrive
at time $(k+1)s$, followed by $M$ full bursts of $C$ jobs spaced exactly $s$
time units apart. In the late instance $\inst_{\mathrm L}$, instead, the $M$
full bursts occur at times $(k+1+\delta)s,(k+2+\delta)s,\ldots,
(k+M+\delta)s$. Since $L\le ks$, any $L$-lookahead scheduler cannot distinguish
the two continuations before time $\delta s$. Thus it must decide whether to
start a batch at time $0$, containing only the first $q$ jobs, or to wait until
time $\delta s$ and start a full batch of size $C$, without knowing which
continuation will occur. These two decisions are respectively optimal for
$\inst_{\mathrm E}$ and $\inst_{\mathrm L}$. With the correct initial decision,
the schedule can clear the prefix exactly before the first of the $M$ full
bursts, and then process each burst immediately upon arrival. Hence, as
$M\to\infty$, the optimal cost on either instance is $CsM+o(M)$.

If the scheduler makes the wrong initial decision, it is left permanently out
of phase with the repeated full bursts. On $\inst_{\mathrm E}$, if it waits
until $\delta s$, then after the prefix it has $q$ leftover jobs while the full
bursts arrive every $s$ time units. It can either defer these $q$ jobs until the
end, paying an additional $qs$ per burst asymptotically, or process them
immediately, delaying every burst by $\delta s$ and paying an additional
$C\delta s$ per burst. Thus the asymptotic extra cost is
$\min\{qs,C\delta s\}M$. Similarly, on $\inst_{\mathrm L}$, if the scheduler
starts early at time $0$, it is left with $C-q$ leftover jobs, and the
asymptotic extra cost is
$\min\{(C-q)s,C(1-\delta)s\}M$. Therefore the deterministic lower bound is
\[
    1+\min\left\{\frac qC,\delta,1-\frac qC,1-\delta\right\}.
\]
Choosing $q=\lfloor C/2\rfloor$ and $\delta=q/C$ gives
$1+\lfloor C/2\rfloor/C$, which is $3/2$ when $C$ is even. For randomized
schedulers, let $p$ be the probability of starting early. The adversary
chooses the worse of the two continuations, giving asymptotic ratio at least
\[
    1+\max\left\{(1-p)\min\left\{\frac qC,\delta\right\},
    p\min\left\{1-\frac qC,1-\delta\right\}\right\}.
\]
Optimizing over $p$ and again taking $q=\lfloor C/2\rfloor$ and
$\delta=q/C$ yields
\[
    1+\frac{\lfloor C/2\rfloor\lceil C/2\rceil}{C^2},
\]
which is $5/4$ when $C$ is even. Taking $M$ sufficiently large absorbs the
$o(M)$ terms and gives the claimed lower bounds up to an arbitrary
$\epsilon>0$.
\end{proofsketch}

Thus, in the adversarial model, perfect but finite lookahead does not eliminate the online--offline separation. This is somewhat surprising: in the uncapacitated problem, lookahead allows the scheduler to approach the offline optimum as the window grows, whereas finite capacity makes batching decisions irreversible. Once the scheduler commits capacity before seeing the unseen continuation, the resulting mistake cannot always be repaired. In the next section, we show that this obstruction is specific to the adversarial model: when arrivals are generated by a fixed stochastic process, finite lookahead is sufficient to overcome the online--offline separation even with capacity constraints. Moreover, the convergence to the offline optimal rate is exponentially fast in the length of the lookahead window, in contrast to the polynomial bound of~\Cref{thm:adv-lookahead-uncap-lower}.

\section{The Poisson Batching Model}\label{sec:poisson}

We now study our batching problem under stochastic arrivals. In particular, we assume that jobs arrive according to a Poisson process of rate $\lambda>0$. Equivalently, after any time $t$ and conditional on any history up to time $t$, the next arrival occurs at time $t+T'$, where $T'\sim \mathrm{Exp}(\lambda)$; the arrival process is therefore memoryless. For any integer number $n$, we use $\inst(n)$ to denote the random instance comprised by the first $n$ job arrivals of that Poisson process.

Our objective in this setting is the asymptotic average flow-time per job, which we call the \emph{rate} of a scheduling policy. We establish existence of this rate, together with bounds on the offline-optimal rate, in \Cref{sec:stoch-rate}. In \Cref{sec:stoch-online}, we study online schedulers without lookahead and show that, under Poisson arrivals, it suffices to optimize over a finite class of $C$ threshold policies. We then prove that even the best such online policy has a strictly larger rate than the offline optimum. Finally, in \Cref{sec:stoch-lookahead,sec:stoch-lookahead-lb}, we study schedulers with finite lookahead: we give a policy whose rate converges to the offline-optimal rate exponentially fast in the lookahead window length, and prove a matching lower bound up to constants in the exponent.

\subsection{The rate of a scheduling policy and bounds on the optimal rate}
\label{sec:stoch-rate}

Given a (feasible) scheduling protocol $\sched$ that respects the batching capacity $C$ and the processing time $s$, we use $\sched(n)$ to denote the (random) total flow-time incurred by $\sched$ on the (random) instance $\inst(n)$. By definition, note that $S(n)\ge sn$ with probability $1$ and therefore this becomes unbounded as $n\to\infty$. Hence, we define the rate of scheduler $\sched$ as the (asymptotically) average flow-time per job that it achieves on such Poisson generated instances, whenever that limit exists. Formally:
\[\rate(\sched):= \lim_{n\to\infty}\frac{\expect{\sched(n)}}{n}.\]

For a fixed $n$, let $\sched^*(n)$ denote the offline optimum which observes the entire realization of $\inst(n)$ before constructing a schedule minimizing total flow-time. By definition, for every feasible scheduler $\sched$ and every $n$, we have $\sched(n)\ge\sched^*(n)$ with probability $1$. Consequently, whenever the corresponding limits exist,
$$
\rate(\sched)\ge \rate(\sched^*)
:=
\lim_{n\to\infty}
\frac{\expect{\sched^*(n)}}{n}.
$$

We will not need, and do not claim, that this limit exists for every feasible scheduler.  Instead, we establish rate existence for the schedulers considered in this paper. For the offline optimum $\sched^*$, existence follows from the superadditivity of the optimal flow-time (\Cref{lemma:superadditive}) and Fekete's lemma. For online and lookahead schedulers, we restrict attention to policies that are stationary and regenerative: whenever the system becomes empty, the policy returns to a fixed internal state and no information about future arrivals is carried over to the next busy period. For such schedulers, the rate exists by the standard renewal-reward argument \cite{bremaud1999markov}, possibly with value $+\infty$.
%As we show in~\Cref{thm:poisson-stability}, the rate of the offline optimum scheduler always exists and is finite only when $\lambda s < B$.

%\shuchi{Are the terms ``regenerative'' and ``renewal-reward'' standard for this literature?} \sanjay{yes, I think these are common in queueing literature (I added ref to Bremaud's book).}

%A straightforward application of Fekete's lemma, coupled with the super-additivity of the flow-time objective suffice to show that these limits are well-defined for any feasible scheduler $\sched$ \textcolor{blue}{Actually no, just for the offline}. The goal of this section is to understand whether simple online policies, and later policies with finite lookahead, can match the asymptotic rate of the offline optimum. 

%We first identify the stability threshold. The Poisson process injects jobs at average rate $\lambda$, while the server can process at most $B$ jobs every $s$ time units. Hence the maximum service rate is $B/s$. If $\lambda s\ge B$, arrivals occur at least as fast as the system can process them, and the backlog necessarily grows on the scale relevant for flow-time. In contrast, when $\lambda s<B$, the system has positive slack. In this stable regime, natural non-idling policies that start a full batch whenever at least $B$ jobs are pending (a property we call work-conserving and is satisfied by any potentially optimal policy without loss), have finite asymptotic rate. We formalize this threshold as follows.

Having separated the question of existence from that of finiteness, we next identify the stability threshold. The Poisson process injects jobs at average rate $\lambda$, while the server can process at most $C$ jobs every $s$ time units. Hence the maximum service rate is $C/s$. As is intuitive, when $\lambda s\ge C$, the system has no positive slack and the average flow-time per job diverges. Furthermore, when $\lambda s<C$, the system is stable under work-conserving policies (\Cref{def:work-conserving}). This is formalized below for completeness, and further discussed in Appendix~\ref{app:poisson-offline}.

%\sanjay{Change theorem to proposition. We can use Theorem  for key non-trivial results.}

\begin{restatable}{proposition}{poissonstability}
\label{thm:poisson-stability}
Assume jobs arrive according to a Poisson process of rate $\lambda>0$, batches have capacity $C$, and each batch requires processing time $s$. Then, the limit $\rate(\sched^*)$ always exists and furthermore it is finite if and only if $\lambda s <C$.
\end{restatable}

%%%%%%%Sanjay%%%%%% The proof of~\Cref{thm:poisson-stability} is deferred to Appendix \textcolor{blue}{XX}. In fact, the proof gives a sharper dichotomy. When $\lambda s>C$, the queue length diverges under every feasible scheduler. Moreover, for every ``reasonable'' scheduler $\sched$, the expected total flow-time satisfies $\expect{\sched(n)} = \kappa n^2 \pm o(n^2)$ as $n\to\infty$, where the leading constant $\kappa=\kappa(\lambda,C,s)$ is independent of $\sched$. Thus, in the strictly overloaded regime, all ``reasonable'' schedulers are asymptotically equivalent up to lower-order terms. At the critical point $\lambda s=C$, the system has zero drift. The backlog no longer grows linearly, but its fluctuations are still large enough to make the average flow-time per job diverge. In particular, we show that $\expect{\sched(n)}=\Theta(n^{3/2})$
%%%%%%%%Sanajy%%%%%%% for every feasible work-conserving scheduler $\sched$. Optimizing the leading constant in this critical regime is a finer question, which we leave for future work.

For the remainder of the section, we restrict attention to the stable regime $\lambda s<C$. Our goal is to determine whether simpler policy classes, such as online schedulers and finite-lookahead algorithms, can match or approximate the (offline) optimal asymptotic rate, i.e. $\rate(\sched^*)$.

\subsection{Online Schedulers}
\label{sec:stoch-online}

We next consider schedulers without access to future arrivals. Formally, an \emph{online scheduler} observes the arrival process only causally: at any time $t$, its decision may depend on arrivals and scheduling decisions up to time $t$, but not on arrivals after $t$. Thus, whenever the processor becomes idle, the scheduler must decide whether to start a batch from the currently pending jobs or wait for additional arrivals.

The memoryless property of Poisson arrivals greatly simplifies online decisions. Suppose the processor is idle with $q$ pending jobs. If $q=0$, there is no decision to make. If $q\ge C$, the scheduler can start a full batch of $C$ jobs without loss in flow-time by work-conservation (\Cref{def:work-conserving}). The only nontrivial case is therefore $q\in\{1,2,\ldots,C-1\}$, where the scheduler must decide whether to process the $q$ jobs immediately or wait and reconsider. Crucially, conditional on no new arrival occurring, waiting does not change the stochastic state relevant to future arrivals: by memorylessness, the residual time until the next arrival is still exponential with rate $\lambda$. Thus, after any positive waiting time without an arrival, the scheduler faces the same decision as before, except that the pending jobs have accumulated additional waiting cost. Hence there is no advantage to waiting a deterministic amount of time and then starting the same batch if no arrival occurs.

Moreover, the optimal online decision is monotone in the queue length: if it is optimal to start a batch with $q$ pending jobs, then it is also optimal to start with $q+1$ pending jobs. Randomization also does not improve the optimal rate: because the objective is linear in the schedule's actions and the state evolution is Markovian, an optimal online policy can be chosen deterministically. Thus the relevant online policies are exactly threshold policies: for some $k\in\{1,\ldots,C\}$, wait while fewer than $k$ jobs are pending, and start a batch as soon as at least $k$ jobs are pending. We summarize these findings in the following statement, formally proven in Appendix~\ref{app:poisson-online}.

\begin{restatable}{proposition}{poissonthreshold}\label{thm:online-thresholds}
Assume $0<\lambda s<C$. For any $k\in \{1,2,\ldots, C\}$, let  $\sched^\on_k$ denote the work-conserving online scheduler that only starts a batch of $\min\{q,C\}$ jobs whenever the processor is idle and there are $q\ge k$ pending arrivals\footnote{Note that $\sched_1^\on$ corresponds to the greedy scheduler $\greedy$ analyzed in the adversarial setting.}. Then, one of the schedulers $\{\sched^\on_k\}_{k=1}^C$ achieves optimal rate across all online schedulers $\sched$\footnote{Formally, the left hand optimum is taken over online-policies whose rate exists as an ordinary limit.}, i.e.

\[ \min_{\sched:\text{ online}}\bigg(\rate(\sched)\bigg)= \min_{k\in\{1,\ldots, C\}} \bigg(\rate(\sched_k^\on)\bigg).\]
\end{restatable}

The main result of this section is a separation between the best rate achievable by any threshold policy and the optimal offline rate in the stable regime. Proving such a separation in the stochastic setting requires a different approach from standard adversarial lower bounds. In adversarial models, lower bounds are typically obtained by allowing the adversary to tailor the future input sequence to the online algorithm’s early mistakes. In the stochastic setting, by contrast, the arrival process is fixed and cannot be adapted. Nevertheless, for each online threshold scheduler $\sched_k^\on$, we identify an input pattern on which the scheduler is suboptimal and show that this pattern occurs with constant probability under Poisson arrivals in any sufficiently long interval. This yields a gap between the rate of that scheduler and the offline optimum. Since only a finite number of threshold policies can be optimal online, these individual gaps imply a general separation. We include a proof sketch for the following theorem here; the formal proof is deferred to Appendix~\ref{app:online_vs_offline}.

%For the main separation result, however, the identity of the best threshold is immaterial. In the stable regime, every threshold policy has finite rate, but every such rate remains separated from the offline optimum by a positive constant.

\begin{restatable}{theorem}{poissonseparation}\label{thm:online-gap}
Assume $0<\lambda s<C$. Then, for every $k\in\{1,2,\ldots,C\}$,
$$
\rate(\sched_k^\on)
\ge
\rate(\sched^*)+\beta,
$$
where the constant $\beta>0$ depends on $\lambda,s$, and $C$.\footnote{Since $s\le \rate(\sched^*)<\infty$ in the stable regime, we can translate additive rate gaps into multiplicative competitive ratio gaps.}
\end{restatable}

\begin{proof}[Proof sketch]
By the threshold characterization of optimal online schedulers (\Cref{thm:online-thresholds}), it suffices to compare the offline optimum with each threshold scheduler $\sched_k^\on$ for $k\in\{1,\ldots,C\}$. Fix such a threshold $k$. We decompose the execution of $\sched_k^\on$ into regeneration cycles: a cycle begins with an arrival to an empty system and ends at the next time the system becomes empty again. These cycles are i.i.d., and the rate of $\sched_k^\on$ is the expected flow-time per cycle divided by the expected number of jobs per cycle. The subtlety is that the offline optimum need not respect the same regeneration cycles as $\sched_k^\on$. We therefore avoid analyzing the offline optimum cycle-by-cycle directly; instead, on certain good cycles, we locally replace the schedule of $\sched_k^\on$ by a {\em better} schedule that completes no later than the original cycle completion time.
This modified schedule can be concatenated across cycles and is globally feasible. Since the offline optimum is no worse than this modified schedule, any constant expected improvement per cycle yields a gap in rates. 

For each threshold $k$, we identify a good event $\mathcal E_k$ within a single regeneration cycle on which the scheduler mistimes its first batch relative to the nearby arrivals: it either starts before a soon-arriving job can be included or waits even though the initial job could already be processed without delaying the remainder of the cycle (see (\ref{eqn:E_k-sketch}) below for a formal definition). In either case, locally re-timing the first batch yields a feasible replacement schedule that reduces flow time while preserving the cycle completion time. Formally, $\mathcal E_k$ has positive probability $p_k$, and conditioned on $\mathcal E_k$, there exists a feasible replacement schedule for the jobs in the cycle that completes no later than $\sched_k^\on$, so that all subsequent cycles are unaffected, and improves the cycle’s total flow time by at least some $\beta_k>0$. These properties imply

%Thus, for each threshold $k$, it is enough to define an event $\mathcal E_k$ on one regeneration cycle with the following properties. First, $\mathcal E_k$ must have positive probability $p_k$. Second, conditioned on $\mathcal E_k$, there must be a feasible local replacement schedule for the jobs in that cycle. Third, this replacement schedule must complete no later than $\sched_k^\on$ completes the cycle, so that later cycles are unaffected. Finally, it must improve the total flow-time in the cycle by at least some $\beta_k>0$. These properties imply
\[
    \rate(\sched_k^\on)-\rate(\sched^*)
    \ge
    \frac{p_k\beta_k}{\expect{\cycjobs(\sched_k^\on)}}
\]
where $\cycjobs(\sched_k^\on)$ is the number of jobs in one regeneration cycle.
We then show that, for every $k$, one can define the event $\mathcal E_k$ so that
\[
    p_k\beta_k \ge \frac{\lambda s^2e^{-2\lambda s}}{8}\;\;\text{(\Cref{lemma:separation-construction})},
\]
while the expected cycle size of any threshold scheduler is at most
\[
    \expect{\cycjobs(\sched_k^\on)}
    \le
    \frac{C^2}{C-\lambda s}e^{\lambda s}\;\;\text{(\Cref{lemma:threshold-cycle-size-bound})}.
\]

The good event is defined differently for $k=1$ and for $k\ge 2$. For $k=1$, let $\varepsilon=s/4$ and let $\mathcal E_1$ be the event that exactly one additional job arrives in $(0,\varepsilon]$ and no further job arrives in $(\varepsilon,2s]$. On this event, $\sched_1^\on$ processes the first two jobs as two singleton batches, whereas the replacement schedule waits until the second arrival and processes the two jobs together in one batch; this completes before the original cycle ends and improves flow-time by at least $s/2$. 

For $k\ge 2$, let $T_{k-1}$ be the time of the $(k-1)$st arrival after the initial arrival, i.e., the time at which $\sched_k^\on$ first reaches its threshold and starts its first batch. Define
\begin{equation} \label{eqn:E_k-sketch}
    \mathcal E_k := \{T_{k-1}\ge s\}.
\end{equation}
On this event, the replacement schedule processes the initial job immediately as a singleton batch, completing it at time $s$, and then from time $T_{k-1}$ onward mimics $\sched_k^\on$ while omitting that already-completed job. This leaves the cycle completion time unchanged and improves the flow-time by $T_{k-1}$. Combining these improvements with the cycle-size bound gives
\[
    \rate(\sched_k^\on)-\rate(\sched^*)
    \ge
    \frac{\lambda s^2(C-\lambda s)}{8C^2}e^{-3\lambda s},
\]
uniformly over all $k$. Since some threshold scheduler is optimal among online schedulers, the same positive gap separates the optimal online rate from the offline optimum.
\end{proof}

Thus, even the optimal online scheduler is asymptotically suboptimal: causal knowledge of the arrival process is not enough to match the rate of the offline optimum.

\subsection{Schedulers with (Perfect) Lookahead}
\label{sec:stoch-lookahead}

Following the separation between online schedulers and the offline optimum under Poisson arrivals, established in the previous section, we now ask whether finite perfect lookahead can overcome this gap. We answer this question affirmatively: we design an online algorithm whose performance approaches the offline optimum rate exponentially fast in the length of the lookahead window.

A natural first idea is to compute the optimal schedule over the entire lookahead window and then follow it exactly. The difficulty is that such a locally optimal schedule may interact with jobs arriving just beyond the lookahead window, destroying the regeneration structure needed to invoke memorylessness of the Poisson process. The guiding idea behind our algorithm, which we call $\algname$ (\Cref{alg:poisson_lookahead}), is therefore to use lookahead only when it certifies a genuine regeneration point. Specifically, the algorithm searches for a prefix of the predicted arrivals that can be completed optimally before the next arrival following that prefix. Such a prefix is self-contained: scheduling it optimally returns the system to the empty state before any outside job can interfere. If no such prefix exists, the lookahead window does not reveal a safe offline-optimal action, and the algorithm instead falls back on a purely online scheduler (namely the ``greedy'' algorithm $\sched_1^\on$ that starts a batch as soon as a job is available) until the system clears. This design preserves the renewal structure needed for the analysis, while ensuring that, except on rare ``bad'' cycles, the algorithm behaves exactly like the offline optimum on the portion of the cycle it serves.

\begin{algorithm}[h]
\caption{$\algname$}
\label{alg:poisson_lookahead}
\begin{algorithmic}[1]
\Loop
    \State Wait until the first arrival; let its arrival time be $a_0$.  \label{step:cycle-start}
    \State Use the $L$-lookahead window to observe all arrivals in $(a_0,a_0+L]$.
    \State Let $ a_1<a_2<\cdots<a_m$ be the observed arrivals in $(a_0,a_0+L]$, and set $a_{m+1}:=a_0+L$.
    \For{$j=0,1,\ldots,m$}
        \State Let $\inst_j:=\{a_0,a_1,\ldots,a_j\}$ be the prefix instance consisting of the first $j+1$ jobs.
        \State Let $C_j$ be the minimum completion time among offline-optimal schedules for $\inst_j$.
        \If{$C_j<a_{j+1}$}
            \State Execute an offline-optimal schedule for $\inst_j$ and end the current cycle. \label{step:cycle-end-1}
            \State \textbf{Continue} to the next cycle.
        \EndIf
    \EndFor
    \State Run $\sched_1^\on$ until all jobs in the lookahead window have completed and the system is empty.
    \State End the current cycle. \label{step:cycle-end-2}
\EndLoop
\end{algorithmic}
\end{algorithm}

A key technical component in the analysis of $\algname$ is the renewal structure described above. Each cycle starts from an empty system, begins with the next arrival, and ends the next time the system becomes empty. The independence of arrivals and decisions across cycles is preserved despite the use of lookahead: if a certifiable prefix is found, the algorithm serves that prefix and clears the system before the next arrival outside it; if no such prefix is found, the algorithm discards the lookahead information and runs $\sched_1^\on$ until the system clears and all observed arrivals have been processed. Thus, at the end of each cycle, the algorithm carries forward no state and no useful information about future arrivals. By the strong Markov property and the memorylessness of the Poisson process, the arrival process after the cycle is independent of the past and has the same law as the original process. Consequently, the long-run performance of $\algname$ can be analyzed through a single representative cycle using the renewal-reward theorem.

The performance analysis then reduces to distinguishing good cycles from bad ones. On a good cycle, in which a certifiable prefix is found, the algorithm incurs no loss in flow-time relative to the offline optimum for the jobs it serves. On a bad cycle, in which no certifiable prefix exists within the lookahead window, the algorithm may incur additional loss before the system clears. Thus all loss relative to the offline optimum can be charged to the event that no certifiable prefix is found. The main technical task is to show that this event is exponentially unlikely in the lookahead length, and that the loss incurred on bad cycles is sufficiently well controlled. This is the content of the main result of this section.

In the following statement a cycle denotes the time interval starting in Step~\ref{step:cycle-start} in some iteration of \Cref{alg:poisson_lookahead} and ending in Step~\ref{step:cycle-end-1} (for a good cycle) or Step~\ref{step:cycle-end-2} (for a bad cycle).

\begin{lemma}\label{lemma:bad_event_bound}
    For any cycle in \Cref{alg:poisson_lookahead}, let $\bad$ denote the ``bad'' event that no suitable prefix was found by $\algname$. Then, $\prob{\bad}\le \kappa\cdot e^{-\gamma L}$ where $\kappa,\gamma>0$ are constants depending only on $\lambda,C$ and $s$.
\end{lemma}

We prove Lemma~\ref{lemma:bad_event_bound} at the end of this section. Equipped with it, and following our discussion from above, we obtain the following result, formally proven in Appendix~\ref{app:poisson-lookahead-upper}.

\begin{restatable}{theorem}{prefixlookahead}\label{thm:poisson-upper}
    Assume $0<\lambda s <C$ and fix any $L\ge 0$. Then, we have
    \[\rate(\algname)\le \rate(\sched^*) + \kappa'\cdot e^{-\gamma' L}\]
    where $\kappa',\gamma'>0$ are constants depending only on $\lambda,C$ and $s$.
\end{restatable}

Comparing~\Cref{thm:poisson-upper} with ~\Cref{thm:adv-lookahead-uncap-lower,thm:adv-lookahead-cap-lb}, we see that the Poisson model is qualitatively more favorable than the adversarial one. In the capacitated setting, sufficiently large lookahead allows the scheduler to approximate the offline optimal rate arbitrarily well; moreover, the convergence is exponentially fast in the lookahead length, in contrast to the polynomial lower bounds that arise in the adversarial model even in the unbounded batch size setting. 

We conclude this section with a proof sketch of our main technical statement (Lemma~\ref{lemma:bad_event_bound}), bounding the probability that $\algname$ is not able to identify a certifiable prefix  within an arrival window of length $L$. A formal proof can be found in Appendix~\ref{app:poisson-lookahead-upper}; see \Cref{lemma:lookahead-bad-probability}.

\begin{proof}[Proof sketch]
The key idea is to upper-bound the probability that $\algname$ fails to find a safe prefix by identifying a simple event that guarantees success. Shift the first arrival of the cycle to time $0$, and run $\sched_1^\on$ on the same arrival sequence. Let $H_L$ be the event that $\sched_1^\on$ clears the system at some time $T\le L-3s$, and that the next arrival after $T$ lies in $(T+2s,T+3s]$. On this event, the prefix consisting of all jobs up to time $T$ has a feasible schedule of makespan $T$ under $\sched_1^\on$. By 
\Cref{lemma:opt-makespan}, any offline-optimal schedule for this prefix has makespan at most $T+2s$. Since the next arrival occurs after $T+2s$, this prefix can be scheduled optimally and cleared before the next arrival; and since that next arrival occurs by time $T+3s\le L$, the certificate is visible within the lookahead window. Thus $H_L$ implies that $\algname$ finds a safe prefix, and hence
\[
\prob{\bad}\le \prob{\overline{H_L}}.
\]

It remains to argue that $\overline{H_L}$ is exponentially unlikely. The intuition is that every time $\sched_1^\on$ clears the system, the process gets a fresh independent chance to certify $H_L$: namely, the next idle gap must lie in the interval $(2s,3s]$. Since the arrival process is Poisson, this happens with some fixed probability, $p_H=e^{-2\lambda s}-e^{-3\lambda s}>0$, independently after each clearing epoch. Thus, as $\sched_1^\on$ goes through its regenerative cycles, the number of cycles needed before seeing a certifying idle gap is geometrically distributed with constant success probability.

The only remaining issue is that the cycles of $\sched_1^\on$ have random lengths. We prove a cycle-length tail bound (\Cref{lemma:s1-cycle-tail}) showing  that these lengths have an exponential moment. Therefore, the time until the first certifying gap is essentially a geometric sum of random variables with exponential tails, and such a sum also has an exponential tail. Consequently, the probability that no certifying gap appears by time $L-3s$ decays exponentially in $L/s$. This gives
\[
\prob{\overline{H_L}}
\le \exp(-\Omega(L/s)),
\]
and hence the same bound for $\prob{\bad}$.
\end{proof}

\subsection{A Lower Bound on the Performance of Schedulers with Lookahead}
\label{sec:stoch-lookahead-lb}

Finally, it is natural to ask whether some other lookahead algorithm could converge to the offline-optimal rate faster as the lookahead window grows. We show that, up to constant factors in the exponent, this is impossible. The main difficulty in proving such a lower bound is that the offline optimum observes the entire instance, and its decisions may couple arrivals that occur far apart in time. To isolate the local information-theoretic obstruction faced by any finite-lookahead scheduler, we introduce a new cycle decomposition of the arrival process. The cycles are defined by minimal prefixes that can be scheduled optimally and cleared before the next arrival. We show that there exists an offline optimum that respects this decomposition: on each such prefix, it processes the jobs independently of the rest of the instance. This yields a renewal structure and reduces the lower-bound analysis to a single representative cycle.

As in the proof of \Cref{thm:online-gap}, we identify a sequence of arrivals that forces suboptimal behavior for any online algorithm with limited lookahead and occurs in any cycle with nonzero probability. The hard event $\hat{\mathcal{E}}_k$ creates two continuations that are indistinguishable to any scheduler with lookahead $L=ks$ at the moment when it must make its first consequential batching decision. The cycle starts with $u=\lceil \capac/2\rceil$ jobs arriving near time $0$, followed shortly after by the remaining $v=\lfloor \capac/2\rfloor$ jobs needed to form a full batch. The scheduler must therefore decide whether to process the first $u$ jobs early, or to wait and process the first full batch of $\capac$ jobs. We then place a long sequence of full $\capac$-job bursts just before the times $s,2s,\ldots,ks$. This visible prefix keeps the two initial decisions nearly balanced: choosing early creates a persistent $v$-job offset, while choosing to wait avoids this offset but delays the first $u$ jobs. The decisive information is hidden just beyond the lookahead window. In the early continuation, $u$ additional jobs arrive soon enough to complete the $v$-job offset, so the early commitment is better. In the late continuation, the same $u$ jobs arrive too late, so waiting is better. Since the two continuations agree throughout the scheduler's lookahead window, every finite-lookahead scheduler makes the wrong first commitment on one of them and incurs a constant excess flow-time loss. Finally, we show that our constructed event $\hat{\mathcal{E}}_k$ has probability $\exp(-\Theta(k))$ under Poisson arrivals. Translating this one-cycle loss through the renewal-reward decomposition yields the following lower bound, formally proven in Appendix~\ref{app:poisson-lookahead-lower}.

%Using this decomposition, we construct an event on a cycle with probability roughly $\exp(-L)$. Conditional on this event, the scheduler must make a decision whose correct resolution depends on arrivals occurring beyond its lookahead window. Consequently, every $L$-lookahead scheduler makes the wrong decision with constant probability on this event, and incurs a constant additive flow-time loss whenever this happens. 

\begin{theorem}\label{thm:poisson-lb}
Assume $0<\lambda s<\capac$. Let $\capac\ge 2$ and fix any $L\ge 0$. For every $L$-lookahead scheduler $\sched$,
$$
\rate(\sched)
\ge
\rate(\sched^*) + \kappa'' e^{-\gamma'' L},
$$
where $\kappa'',\gamma''>0$ are constants depending only on $\lambda$, $\capac$, and $s$.
\end{theorem}
%\begin{proof}
%    \textcolor{blue}{This is also one I would like to add, but it takes space so let's decide on that later.} \shuchi{Perhaps leave this proof to the appendix given the detailed discussion above.}
%\end{proof}

\section{Lookahead with Imperfect Predictions: Consistency and Robustness}
\label{sec:prediction_error}

So far, we have considered perfect predictions within the lookahead window. We showed that such predictions are useful in both the uncapacitated adversarial setting, where they allow the competitive ratio to approach $1$ as the window grows (\Cref{thm:adv-lookahead-uncap-upper}), and the capacitated Poisson setting, where they allow the rate to converge exponentially fast to the offline optimum (\Cref{thm:poisson-upper}). In practice, however, predicted arrivals will rarely be exact. We therefore ask whether imperfect predictions can still improve upon purely online scheduling.

\paragraph{The Prediction Model.}
Let $\inst=\{r_j\}_{j=1}^n$ denote the true arrival sequence, which may be either adversarial or stochastic, and let $\inst'=\{r'_j\}_{j=1}^n$ denote the predicted arrival sequence. At each time $t$, the algorithm observes all true arrivals in $\inst$ up to time $t$, together with all predicted arrivals in $\inst'$ up to time $t+L$. We assume that the two instances contain the same number of jobs, index both arrival sequences in non-decreasing order, and match jobs by rank. We define the maximum prediction error as
$$
\epsilon
:=
\max_{j\in[n]} |r_j-r'_j|.
$$
We assume that $\epsilon\ge 0$ is known to the algorithm. In the Poisson setting, the same rank-wise error bound is assumed to hold for every finite prefix of the arrival sequence. 

Crucially, no assumptions are made whatsoever about the predicted sequence; in fact, it can be fully adversarial and of arbitrarily low quality. Obtaining high-accuracy predictions is outside the scope of this work. Instead, we show that given any such sequence, we can obtain guarantees that scale with the underlying precision without making any further assumptions.

\paragraph{Our Results.}
We introduce simple error-aware modifications of our two main algorithms, namely $\advname$ and $\algname$. When $\epsilon=0$, they recover the corresponding perfect-lookahead guarantees. As the error grows, their performance degrades smoothly until they revert to prediction-free baselines, retaining robust guarantees under arbitrarily inaccurate predictions.

\begin{restatable}{theorem}{predictionerror}
\label{thm:imperfect-prediction}
Consider the prediction model above. In the uncapacitated adversarial setting, the error-aware implementation of $\advname$ satisfies
$$
\frac{\expect{\mathrm{Error}\advname(\inst)}}{\opt(\inst)}
\le
1+\min\left\{
O\left(
\frac{s}{L}
+
\frac{\epsilon}{s}
\right),
1
\right\}.
$$
Moreover, assuming $\lambda s<\capac$, the error-aware implementation of $\algname$ satisfies
$$
\rate(\mathrm{Error}\algname)
\le
\rate(\sched^*)
+
\min\left\{
O_{\lambda,s,\capac}(1)\cdot e^{-\Omega_{\lambda,s,\capac}(L)}
+
O(\epsilon),
O_{\lambda,s,\capac}(1)
\right\}.
$$
\end{restatable}

We prove \Cref{thm:imperfect-prediction} in Appendix~\ref{app:prediction-error-proofs}. Both modifications execute schedules computed from the predictions in the same manner: whenever the unmodified algorithm uses its lookahead to commit to a schedule for the predicted future, the modified algorithm commits to the corresponding true jobs, matched by rank, and starts each batch as soon as it becomes feasible. This delays every batch by at most $\epsilon$ and increases the total flow time by at most $2\epsilon$ per job.

Below, we describe the error-aware modifications and sketch the proof of their error-dependent consistency guarantees. Robustness follows by comparing these bounds with the corresponding prediction-free baselines (\Cref{thm:adv-online-bounds} and~\Cref{thm:online-thresholds}) and using whichever guarantee is better, ensuring that low-quality predictions cannot hurt the policies we design.

\paragraph{Error-Aware Random Shift.}
Recall that $\advname$ partitions time into periods of length $L+2s$, each consisting of an active interval of length $L$ and an ignored interval of length $2s$. Under prediction error, we remove an $\epsilon$-length guard interval from each endpoint of every active interval and optimize only over the resulting active core, of length $(L-2\epsilon)^+$. Jobs predicted in either guard interval are treated as ignored jobs.

The true counterpart of every job predicted in an active core arrives within the original active interval. Moreover, the core ends $\epsilon$ before the original active interval, compensating for the possible $\epsilon$ delay incurred when transferring its predicted-optimal schedule. Thus, schedules from consecutive periods remain disjoint. The ignored portion of each period increases from $2s$ to $2s+2\epsilon$, and every ignored job incurs flow time $O(s+\epsilon)$. The original analysis therefore gives
$$
\frac{\expect{\mathrm{Error}\advname(\inst)}}{\opt(\inst)}
\le
1+
O\left(
\frac{s}{L}
+
\frac{\epsilon}{s}
\right).
$$

\paragraph{Error-Aware Prefix Lookahead.}
For $\algname$, a certificate for the predicted instance need not be a valid regeneration certificate for the true instance. The modified algorithm therefore searches only within the effective horizon $(L-2\epsilon)^+$. For each predicted prefix, let $T'_j$ be the earliest completion time among its offline-optimal schedules, and let $r'_{j+1}$ be the next visible predicted arrival. The prefix is accepted only if
$$
T'_j+2\epsilon<r'_{j+1}.
$$
The algorithm then transfers a predicted-optimal schedule for this prefix to the corresponding true jobs. The transferred schedule completes by time at most $T'_j+\epsilon$, while the next true arrival satisfies $r_{j+1}\ge r'_{j+1}-\epsilon$. Hence $T'_j+\epsilon
<
r'_{j+1}-\epsilon
\le
r_{j+1}$ so every accepted prefix remains a valid regeneration certificate.

Conversely, true certificates with sufficient additional slack remain detectable from the predictions. Accounting for the shortened horizon, the probability that no such certificate is found remains exponentially small in $L$ whenever the prediction-assisted guarantee is better than the prediction-free baseline. Transferring the resulting schedules contributes an additional $O(\epsilon)$ flow time per job. If no valid certificate is found, $\algname$ runs $\sched^{\on}_1$ on the true arrivals until the system becomes empty.

\section{Evaluation}
\label{sec:eval}

We complement our theoretical results with an experimental evaluation on both synthetic and real-world arrival traces. The synthetic experiments use a Monte Carlo simulator for Poisson arrivals and test the main predictions of the stochastic theory: long-run rates stabilize, purely online policies remain separated from the offline optimum, finite lookahead rapidly closes this gap, and noisy predictions degrade gracefully. 
The real-world experiments evaluate whether similar qualitative behavior appears on production request traces.

We implement three batch scheduling algorithm baselines: the offline optimum, the online optimum, and lookahead. The offline optimum (OPT) is the exact flow-time minimizing schedule computed via dynamic programming. For the online algorithm baseline, we implement a pseudo-online algorithm that runs the $\capac$ schedulers $\sched_k^\on$ for $k\in \{1,\ldots , \capac\}$ and executes the one achieving the minimum flow time; it therefore estimates a lower bound on the optimal online rate from~\Cref{thm:online-thresholds}. For lookahead, we implement a simplified variant of PrefixLookahead that uses the same lookahead information but omits the certification and fallback steps.\footnote{These steps are needed in the Poisson analysis to preserve regeneration. In finite-trace evaluations, however, they may be overly conservative, triggering frequent fallbacks and limiting the gains from lookahead. The simplified scheduler avoids these issues while remaining a valid \(L\)-lookahead policy: it uses only pending jobs and arrivals visible within the next \(L\) time units, respects the same batching constraints, and commits only to the next batch before re-planning.} Whenever the server is idle, the scheduler solves the offline subproblem induced by the currently pending jobs and the predicted arrivals in the next \(L\) time units, and then executes the first batching decision of that schedule.
% \aditya{add this as footnote? We use this simplified lookahead scheduler only for evaluation.}
The results are summarized in Figure~\ref{fig:lookahead-evaluation-main}; we describe the experimental setup and findings in the next two subsections. Finally, we compare with the implementation of \Cref{alg:poisson_lookahead} in Appendix~\ref{sec:lookahead-complete}.

\begin{figure}[h]
    \centering

    \begin{subfigure}[t]{0.24\textwidth}
        \centering
        \includegraphics[width=\linewidth]{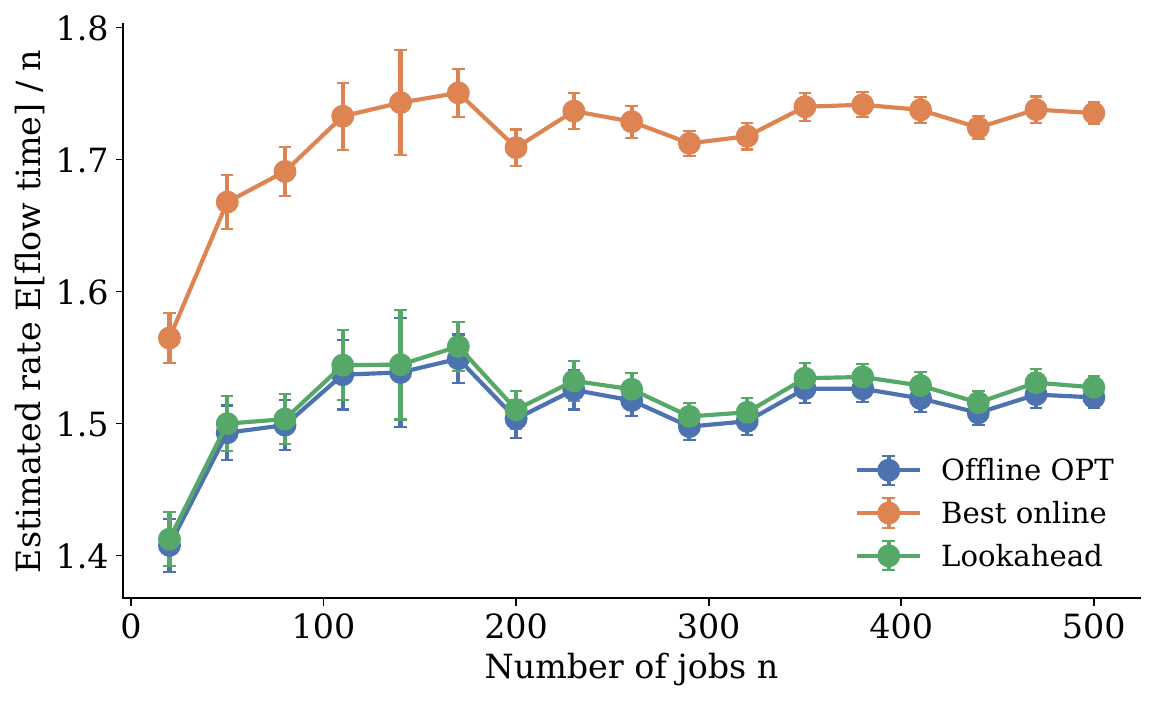}
        \caption{Rate Convergence}
        \label{fig:online-lookahead-rate-main}
    \end{subfigure}
    \hfill
    \begin{subfigure}[t]{0.24\textwidth}
        \centering
        \includegraphics[width=\linewidth]{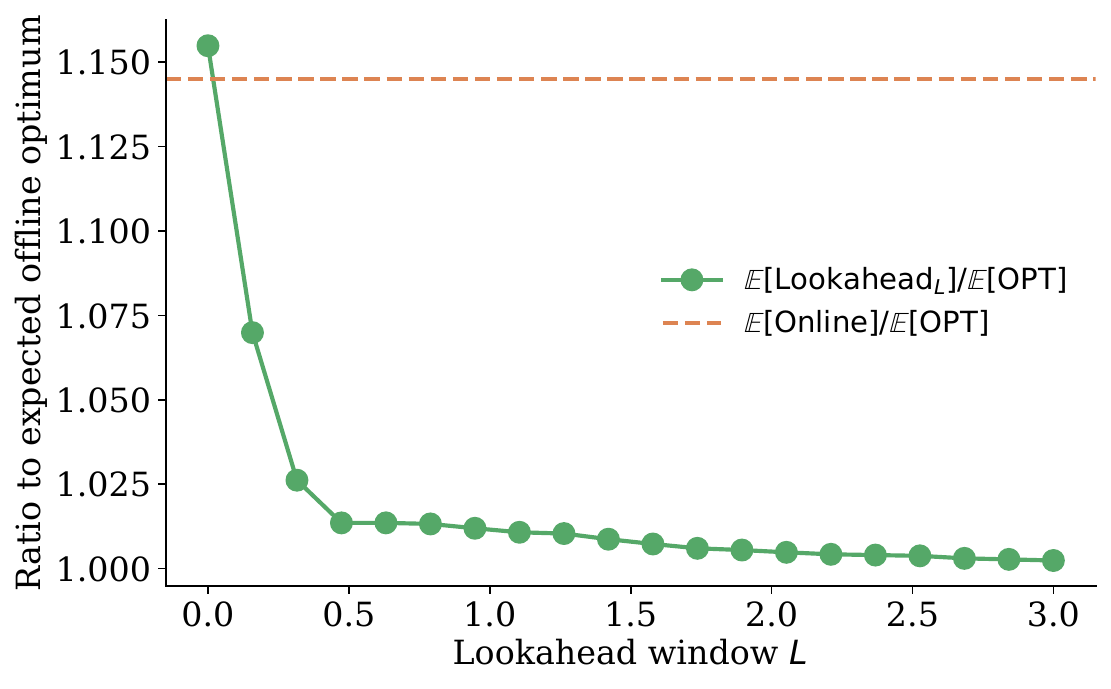}
        \caption{Lookahead Effect}
        \label{fig:lookahead-ratio-vs-L-main}
    \end{subfigure}
    \hfill
    \begin{subfigure}[t]{0.24\textwidth}
        \centering
        \includegraphics[width=\linewidth]{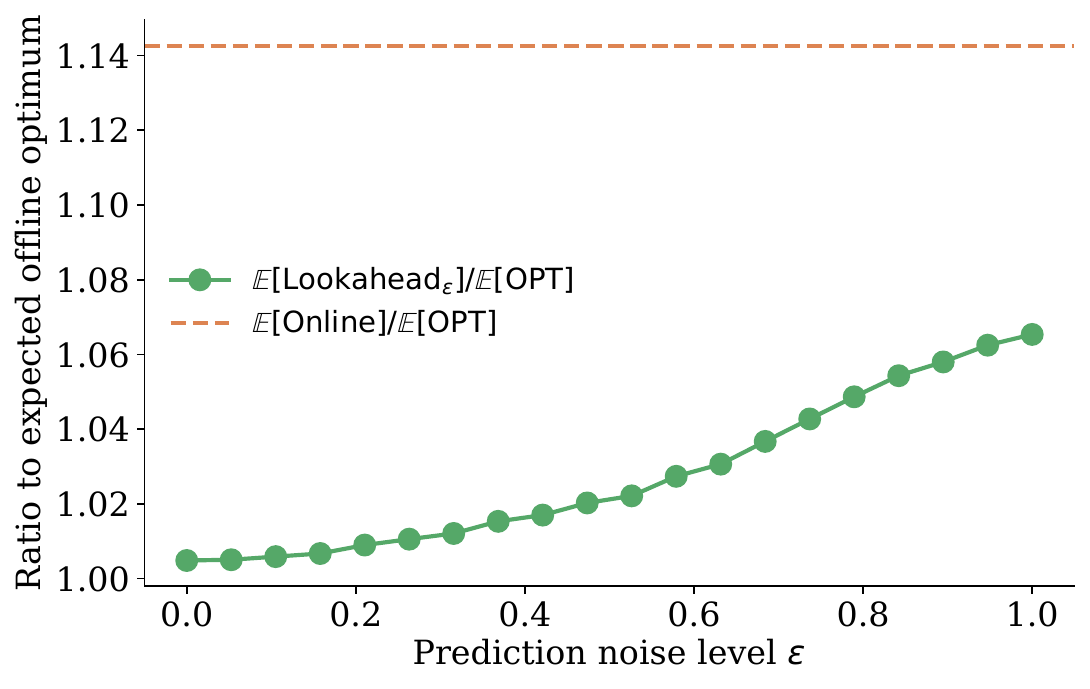}
        \caption{Error Robustness}
        \label{fig:lookahead-ratio-vs-epsilon-main}
    \end{subfigure}
    \hfill
    \begin{subfigure}[t]{0.24\textwidth}
        \centering
        \includegraphics[width=\linewidth]{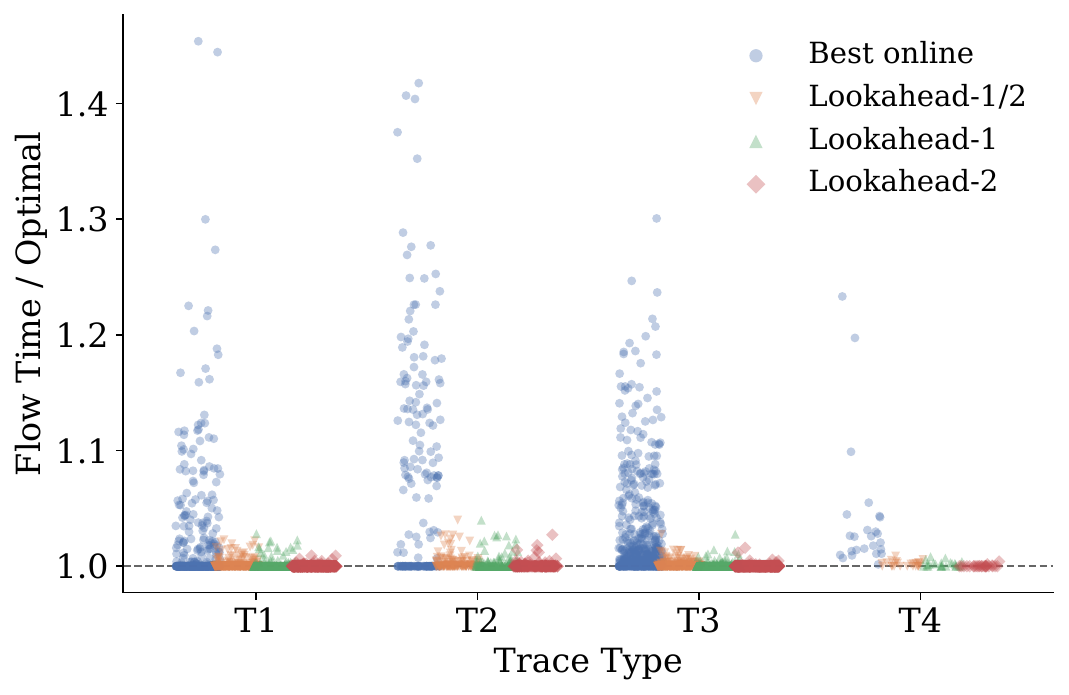}
        \caption{Real Traces}
        \label{fig:real-trace-flow-time-main}
    \end{subfigure}

    \caption{Experimental evaluation of online batching with lookahead. Panels~(a)--(c) use synthetic Poisson arrivals. Panel~(d) depicts evaluations on real-world request traces.}
    \label{fig:lookahead-evaluation-main}
\end{figure}

\subsection{Simulations over Poisson Arrivals}

We first evaluate the algorithms on synthetic traces generated by a Poisson arrival process. For the purposes of illustration, we use rate $\lambda=1.3$, setup time $s=1$ and batch capacity $\capac=2$ so that $\lambda s < \capac$. We use $L=2s=2$ as the default lookahead window. Each data point is averaged over $100$ independently sampled instances.

\paragraph{Asymptotic Rates.} In Figure~\ref{fig:online-lookahead-rate-main}, we compare the performance of our three baselines on the random instance $\inst(n)$ for increasing values of $n$. The three lines depict the estimated average flow-time rates. All algorithms are evaluated on identical instances using common randomness.

We make the following observations. First, all three curves stabilize as $n$ grows, which is consistent with the existence of long-run rates in the stable regime $\lambda s<C$ (\Cref{thm:poisson-stability}). 
The figure also shows a clear separation between the best online policy and the offline optimum, as predicted from~\Cref{thm:online-gap}. In contrast, the finite-lookahead scheduler lies much closer to the offline optimum, showing that short-horizon sample-path information eliminates most of the online loss, as predicted by~\Cref{thm:poisson-upper}.

\paragraph{Effect of the Lookahead Window.} In Figure~\ref{fig:lookahead-ratio-vs-L-main} we isolate the effect of the lookahead length. For different values of the lookahead window $L$, we depict the ratio between the pseudo-online lower bound and the offline optimum (orange line) and the corresponding ratio between the lookahead-algorithm described above and the offline optimum (green line). As~\Cref{thm:poisson-upper} predicts, the last ratio approaches $1$ exponentially fast as $L$ grows.

\paragraph{Robustness under Prediction Error.}

We next test whether the lookahead advantage depends on exact predictions. 
For each true arrival in the lookahead window, we perturb its predicted arrival time by an independent bounded error of scale $\epsilon$, re-sort the perturbed predicted arrivals, and run the same receding-horizon lookahead scheduler on the perturbed window. 
Thus, $\epsilon$ controls the maximum rank-matched arrival-time error in the predicted sequence, matching the prediction-error model of Section~\ref{sec:prediction_error}. Figure~\ref{fig:lookahead-ratio-vs-epsilon-main} shows the resulting slowdown as a function of $\epsilon$. 
The slowdown increases smoothly with the prediction error, rather than exhibiting a sharp failure threshold. 
Moreover, for a broad range of errors, the noisy-lookahead scheduler remains better than the prediction-free online baseline. 
This supports the robustness message of Section~\ref{sec:prediction_error}: imperfect predictions still provide useful information, and their benefit degrades gracefully as the predicted arrival sequence moves away from the true one.

\subsection{Experimenting over Real-World Request Traces}

Finally, we evaluate our algorithms on four real-world request-arrival traces: AzureFunctionsDataset2019 (T1), AzureFunctionsInvocationTrace2021 (T2), AlibabaGPU2020 (T3), and AlibabaGenAI2026 (T4). 
To approximate a simple cluster deployment in which requests from the same function or user are routed to the same server, we partition the Azure traces by function and the Alibaba traces by user, treating each resulting group as a separate arrival sequence. 
We discard sequences containing fewer than 100 arrivals. 
For each instance, we compute the competitive ratio of each algorithm relative to the offline optimum and plot it as a scatter point in Figure~\ref{fig:real-trace-flow-time-main}, grouped by trace type.

For each trace instance, we set the setup time $s$ to half of the median inter-arrival time. 
We set the maximum batch size $\capac$ to approximately three times the expected number of arrivals during one setup interval. 
We also verify that the results are not sensitive to $s$ and $\capac$ choices in Appendix~\ref{sec:parameter-generalization}.
For the lookahead algorithm, we sweep the lookahead window length $L$ with different magnitudes of $s$. 

\paragraph{Results.}

Across all four trace types, we find that lookahead consistently reduces slowdown relative to the best fixed-wait baseline, indicating that real workloads are not completely adversarial, and contain short-term structure that can be exploited by prediction-guided batching.
Increasing the lookahead window further reduces the gap to the offline optimum. 
These results complement our theory: although finite lookahead cannot eliminate worst-case gaps in the capacitated adversarial model, it provides robust and practically meaningful improvements on real request traces.
% \dimitris{Let's briefly comment on why we don't implement the theoretical algorithms.}

\begin{comment}
    
\section{Conclusion}
This paper studies short-horizon predictions as a principled tool for online batching in inference serving. Beginning with the idealized setting with perfect predictions of future arrivals, in the adversarial case, we show that bounded batch capacity limits what any finite-lookahead scheduler can guarantee, yielding a worst-case barrier.  However, this barrier can be overcome with stochasticity in the arrival distribution: under Poisson arrivals in the stable regime, a scheduler with finite perfect lookahead approaches the offline-optimal rate exponentially fast in the window size. Further, this rate is robust under settings where we only have imperfect predictions of the short-term future arrivals. Our trace-based evaluation on realistic workloads shows that modest lookahead reduces flow time beyond fixed-waiting baselines and narrows the gap to the offline optimum. More broadly, the results motivate batching schedulers that combine online decisions, short-window forecasts, and queueing structure. This points toward prediction-aware serving policies with both analyzable guarantees and measurable performance gains.
\end{comment}

\section{Conclusion}

This paper studies short-horizon predictions in the form of lookahead as a principled tool for online batching and scheduling. Perhaps counterintuitively, we find that finite lookahead is most powerful when requests are generated stochastically. The intuition is that, under stochastic arrivals, the probability of encountering the adverse patterns that limit online performance decays exponentially with the length of the pattern. 

Our experiments on real request traces reveal consistent performance gains: even modest lookahead substantially outperforms fixed-wait baselines and closes much of the gap to the offline benchmark. Taken together, our theoretical and empirical results show that the value of lookahead is determined not by the prediction horizon alone, but by how that horizon interacts with the system constraints and the structure of the arrival process.

Extending this theoretical and experimental framework to broader and richer batching models is a natural direction for future work.

\subsection*{Declaration of AI Usage}
The paper was written and submitted on July 10, 2026. The authors used AI-based tools to improve the presentation of the manuscript, including polishing the writing and revising the organization, layout, and presentation order of proofs. The authors developed the underlying algorithmic ideas and constructed and sketched the proofs themselves; AI tools were additionally used to help present these arguments in a clearer and more polished form. The authors take full responsibility for the content of the paper.

%Our results identify when finite lookahead can fundamentally overcome online uncertainty in batching. Under adversarial arrivals, its power is limited by capacity: with unbounded batches, performance approaches the offline optimum, whereas any finite capacity preserves a constant gap for every bounded lookahead window. Under stable Poisson arrivals, by contrast, the obstruction is confined to exponentially unlikely events: regeneration allows a finite-lookahead scheduler to match the offline optimum except on such events, so the online-offline gap vanishes exponentially fast with the lookahead length, at an asymptotically tight rate.

%These guarantees remain meaningful when forecasts are imperfect. Our error-aware algorithms degrade smoothly with the rank-matched arrival-time error while retaining constant prediction-free guarantees under arbitrarily inaccurate predictions. Experiments on synthetic arrivals and real request traces further show that modest lookahead improves substantially over fixed-wait baselines and narrows the gap to the offline benchmark. Taken together, our results show that the power of lookahead is determined not by the prediction horizon alone, but by its interaction with the system constraints and the structure of the arrival process.

\clearpage
\appendix

\startcontents[appendices]
\section*{Appendix}
%\section*{Appendix Contents}
%\printcontents[appendices]{}{1}[2]{}

\section{The Adversarial Batching Model (Omitted Proofs from~\Cref{sec:adversarial})}

This appendix contains the formal proofs for the adversarial batching model studied in~\Cref{sec:adversarial}. We begin in~\Cref{app:superadditive} with the proof of the super-additivity property (\Cref{lemma:superadditive}). In~\Cref{app:makespan}, we then establish general structural facts about the offline optimum and prove the makespan comparison lemma (\Cref{lemma:opt-makespan}). The remaining subsections provide the formal proofs of the main adversarial results: the online bounds of~\Cref{thm:adv-online-bounds} in~\Cref{appendix:adv-online-bounds}, the uncapacitated lookahead upper bound of~\Cref{thm:adv-lookahead-uncap-upper} in~\Cref{appendix:adv-lookahead-upper}, the uncapacitated lookahead lower bound of~\Cref{thm:adv-lookahead-uncap-lower} in~\Cref{appendix:adv-uncap-lookahead-lb}, and the capacitated lookahead lower bound of~\Cref{thm:adv-lookahead-cap-lb} in~\Cref{sec:adv-cap-lb}.

\subsection{Super-Additivity of the Optimal Flow-Time (\Cref{lemma:superadditive})}\label{app:superadditive}

We begin with the proof of~\Cref{lemma:superadditive}, which states that the optimal flow-time objective is super-additive over partitions of the job set. This property will be useful when comparing algorithms to the offline optimum: it allows us to lower-bound the offline cost on a full instance by the sum of the offline costs on disjoint sub-instances.

\superadditivelemma*
\begin{proof}
Consider an optimal offline schedule for $\inst$. We construct from it two feasible schedules, one for $\inst_1$ and one for $\inst_2$. For each $k\in\{1,2\}$, take the schedule for $\inst$ and delete from every batch all jobs that do not belong to $\inst_k$; if a batch becomes empty, discard it. All remaining batches are processed at the same times as in the original schedule.

This gives a feasible schedule for $\inst_k$. Indeed, every job of $\inst_k$ is processed no earlier than its release time, because this was true in the original schedule. Moreover, every remaining batch has size at most $\capac$, since it is a subset of a batch in the original schedule.

Let $T_j^*$ denote the completion time of job $j$ in the optimal schedule for $\inst$. In the restricted schedule for $\inst_k$, every job $j\in \inst_k$ completes at the same time $T_j^*$. Therefore, the cost of this restricted schedule is exactly
$$
\sum_{j\in \inst_k} (T_j^*-r_j).
$$
Since $\opt(\inst_k)$ is the minimum possible flow time for $\inst_k$, we have
$$
\opt(\inst_k)
\le
\sum_{j\in \inst_k} (T_j^*-r_j).
$$
Summing over $k=1,2$, and using that $\inst_1,\inst_2$ partition $\inst$, gives
$$
\begin{aligned}
\opt(\inst_1)+\opt(\inst_2)
&\le
\sum_{j\in \inst_1}(T_j^*-r_j)
+
\sum_{j\in \inst_2}(T_j^*-r_j) =
\sum_{j\in \inst}(T_j^*-r_j)
=
\opt(\inst).
\end{aligned}
$$
\end{proof}

\subsection{Approximate Makespan Optimality of the Offline Optimum (\Cref{lemma:opt-makespan})}\label{app:makespan}

We next prove~\Cref{lemma:opt-makespan}. We first introduce notation for the number of arrivals, the number of completed jobs, and the makespan of a schedule.

\begin{definition}[Number of Arrivals]
Let $\inst$ be any instance. For each time $t\ge 0$, we use
$$
A(t\mid \inst):=\big|\{r_j\in\inst: r_j\le t\}\big|
$$
to denote the number of jobs that have arrived on instance $\inst$ by time $t$.
\end{definition}

\begin{definition}[Number of Completed Jobs]
Let $\inst$ be any instance and let $\sched$ be any schedule. For each time $t\ge 0$, we use
$$
N_\sched(t\mid \inst):=\big|\{r_j\in\inst: T_\sched(j)\le t\}\big|
$$
to denote the number of jobs that have completed their processing under $\sched$ on instance $\inst$ by time $t$.
\end{definition}

\begin{definition}[Makespan]
Let $\inst$ be any instance and let $\sched$ be any schedule. We use
$$
\makespan(\sched\mid\inst):=\max_{r_j\in \inst}\{T_\sched(j)\}
$$
to denote the makespan of $\sched$ on instance $\inst$, i.e. the largest completion time among all jobs.
\end{definition}

We first record that offline-optimal schedules are necessarily work-conserving.

\begin{proposition}\label{claim:opt-work-conserving}
Let $\inst$ be any instance and let $\sched$ be any feasible schedule for $\inst$ that minimizes flow-time, i.e. $\sched$ satisfies $\sched(\inst)=\opt(\inst)$. Then $\sched$ is work-conserving.
\end{proposition}

\begin{proof}
Let $\sched$ be an offline-optimal schedule. First, suppose that $\sched$ starts a batch while $q$ jobs are pending but includes fewer than $\min\{q,\capac\}$ of them. Choose one omitted pending job and add it to this same batch, removing it from the later batch in which it was originally scheduled. If that later batch becomes empty, discard it. This modification preserves feasibility, does not delay any job, and strictly decreases the completion time of the added job. This contradicts optimality.

Second, suppose that the processor is idle at some time $t$ while at least $\capac$ jobs are pending. Let $u\ge t$ be the next time at which $\sched$ starts a batch. By the first part, this batch is full. If some job in this batch has not yet arrived by time $t$, then we may swap it with a job that was already pending at time $t$ and is scheduled in a later batch. This preserves feasibility and does not change the multiset of completion times. Repeating this exchange if necessary, we obtain an optimal schedule in which the batch started at time $u$ consists only of jobs that were already pending at time $t$. Starting this full batch at time $t$ instead of $u$ is feasible, does not delay any later batch, and strictly decreases the completion times of the jobs in this batch. This again contradicts optimality. Therefore $\sched$ is work-conserving.
\end{proof}

We next show that a work-conserving schedule cannot lag too far behind any other feasible schedule in terms of the number of completed jobs.

\begin{lemma}\label{lemma:job_difference}
Let $\inst$ be any instance, let $\sched_1$ be any work-conserving schedule for $\inst$, and let $\sched_2$ be any feasible schedule for $\inst$. Then, for all $t\ge 0$,
$$
N_{\sched_2}(t\mid \inst)
\le
N_{\sched_1}(t+s\mid \inst)+\capac-1 .
$$
\end{lemma}

\begin{proof}
Assume for contradiction that the statement is false, and let $\tau$ be the first time at which a violation occurs. Thus,
\begin{equation}\label{eq_lemma_count_1}
N_{\sched_2}(\tau\mid \inst)
>
N_{\sched_1}(\tau+s\mid \inst)+\capac-1 .
\end{equation}
Since no schedule can complete a job before processing it for time $s$, we have $\tau\ge s$. By the minimality of $\tau$, the inequality is not violated at time $\tau-s$, so
\begin{equation}\label{eq_lemma_count_2}
N_{\sched_2}(\tau-s\mid \inst)
\le
N_{\sched_1}(\tau\mid \inst)+\capac-1 .
\end{equation}

By~\Cref{eq_lemma_count_1}, there are at least $\capac$ jobs that are completed by $\sched_2$ by time $\tau$ but are not completed by $\sched_1$ by time $\tau+s$. These jobs must have been released by time $\tau-s$. Moreover, none of them can be in a batch of $\sched_1$ that starts at or before time $\tau$, since any such batch would complete by time $\tau+s$. Hence all these jobs are pending under $\sched_1$ throughout the interval $(\tau-s,\tau]$.

Since $\sched_1$ is work-conserving, it cannot remain idle while these $\capac$ jobs are pending. Therefore, either $\sched_1$ starts a full batch at time $\tau$, or it is already processing a full batch that started after time $\tau-s$, or it completes a batch at time $\tau$ and then immediately starts a full batch. In all cases, $\sched_1$ completes at least $\capac$ additional jobs during the interval $(\tau,\tau+s]$. Thus,
\begin{equation}\label{eq_lemma_count_3}
N_{\sched_1}(\tau+s\mid \inst)
\ge
N_{\sched_1}(\tau\mid \inst)+\capac .
\end{equation}
On the other hand, every feasible schedule can complete at most $\capac$ jobs in any interval of length $s$, and hence
\begin{equation}\label{eq_lemma_count_4}
N_{\sched_2}(\tau\mid \inst)
\le
N_{\sched_2}(\tau-s\mid \inst)+\capac .
\end{equation}
Combining~\Cref{eq_lemma_count_2}, \Cref{eq_lemma_count_3}, and~\Cref{eq_lemma_count_4}, we get
$$
\begin{aligned}
N_{\sched_2}(\tau\mid \inst)-N_{\sched_1}(\tau+s\mid \inst)
&\le
N_{\sched_2}(\tau-s\mid \inst)+\capac
-
N_{\sched_1}(\tau+s\mid \inst) \\
&\le
N_{\sched_2}(\tau-s\mid \inst)
-
N_{\sched_1}(\tau\mid \inst) \\
&\le
\capac-1,
\end{aligned}
$$
contradicting~\Cref{eq_lemma_count_1}.
\end{proof}

We can now prove the approximate makespan optimality of the offline optimum.

\makespanlemma*
\begin{proof}
Let $\sched^*$ be an offline-optimal schedule for $\inst$, and let $\sched$ be any feasible schedule for $\inst$. Define $M^*:=\makespan(\sched^*\mid\inst)$ and $M:=\makespan(\sched\mid\inst)$. It is enough to show that $M^*\le M+2s$, since $\sched$ was arbitrary. By definition of the makespan, the schedule $\sched$ completes every job by time $M$. Since every batch requires processing time $s$, every job in $\inst$ must have arrived by time $M-s$. Moreover,
$$
N_\sched(M\mid\inst)=|\inst|.
$$
By~\Cref{claim:opt-work-conserving}, the schedule $\sched^*$ is work-conserving. Applying~\Cref{lemma:job_difference} with $\sched_1=\sched^*$, $\sched_2=\sched$, and $t=M$, we obtain
$$
|\inst|
=
N_\sched(M\mid\inst)
\le
N_{\sched^*}(M+s\mid\inst)+\capac-1.
$$
Thus, at time $M+s$, the schedule $\sched^*$ has at most $\capac-1$ unfinished jobs.

If no jobs are unfinished at time $M+s$, then $M^*\le M+s$ and we are done. Otherwise, all unfinished jobs have already arrived, since every job arrived by time $M-s$. If $\sched^*$ is idle at time $M+s$, then starting a batch containing all unfinished jobs at time $M+s$ is feasible and strictly improves the total flow time, contradicting optimality. Hence $\sched^*$ must be processing a batch at time $M+s$. Since there are at most $\capac-1$ unfinished jobs in total, this batch has room to contain every unfinished job. If some unfinished job were not in this batch, adding it to the current batch would preserve feasibility and strictly decrease its completion time without delaying any other job, again contradicting optimality. Therefore all unfinished jobs are contained in the batch being processed at time $M+s$, and they complete by time $M+2s$. Thus $M^*\le M+2s$ for every feasible schedule $\sched$ and the proof follows.
\end{proof}

\subsection{The Competitive Ratio of Online Algorithms (\Cref{thm:adv-online-bounds})}\label{appendix:adv-online-bounds}

We begin by observing that the number of completed jobs and arrivals can be used to express the total flow time as an area. Indeed, at every time $t$, the quantity $A(t\mid \inst)-N_\sched(t\mid \inst)$ is exactly the number of jobs that have arrived but have not yet completed processing under $\sched$. Therefore, integrating this quantity over time counts the total amount of time that jobs spend in the system, which is precisely the total flow time.

\begin{fact}\label{claim:flow-def}
Let $\inst$ be any instance and let $\sched$ be any schedule. Then, the total flow time of $\sched$ on $\inst$ can be expressed as
$$
\sched(\inst)
=
\int_{t=0}^\infty
\big(A(t\mid \inst) - N_\sched(t\mid \inst)\big)\,dt .
$$
\end{fact}

We also record a simple lower bound on the optimal flow time, stating that every job contributes at least $s$ to the total flow time of any feasible schedule, since each batch requires $s$ units of processing time.

\begin{fact}\label{claim:opt-lb}
Let $\inst$ be any instance. Then, $\opt(\inst)\ge s|\inst|$.
\end{fact}

We now prove~\Cref{thm:adv-online-bounds}, which we re-state for the reader's convenience. For simplicity, we break the proof into three separate statements: the upper bound on greedy is proven in~\Cref{lemma:adv_greedy}, the deterministic lower bound in~\Cref{lemma:adv_det_lb}, and the randomized lower bound in~\Cref{lemma:adv_ran_lb}.

\adversarialonline*

\begin{lemma}\label{lemma:adv_greedy}
The greedy online schedule $\greedy$ is $2$-competitive.
\end{lemma}

\begin{proof}
Fix any adversarial instance $\inst$, and let $\sched^*$ be an offline-optimal schedule for $\inst$. We begin by showing that the number of jobs completed by the greedy scheduler at any time dominates the number of jobs completed by the offline optimum, up to a lag of $s$. Namely, we show that
$$
N_\greedy(t+s\mid\inst)\ge N_{\sched^*}(t\mid\inst)
$$
for all $t\ge 0$.

Assume for contradiction that this is not the case, and let $\tau$ be the first time such that
\begin{equation}\label{eq:adv_online_ub_1}
N_\greedy(\tau+s\mid\inst)
<
N_{\sched^*}(\tau\mid\inst).
\end{equation}
Since no schedule can complete any job before time $s$, we have $\tau\ge s$. By the minimality of $\tau$, the inequality is not violated at time $\tau-s$, and hence
\begin{equation}\label{eq:adv_online_ub_2}
N_\greedy(\tau\mid\inst)
\ge
N_{\sched^*}(\tau-s\mid\inst).
\end{equation}
Furthermore, by the capacity constraint, no schedule can complete more than $\capac$ jobs during any interval of length $s$. Therefore,
\begin{equation}\label{eq:adv_online_ub_3}
N_{\sched^*}(\tau\mid\inst)
\le
N_{\sched^*}(\tau-s\mid\inst)+\capac .
\end{equation}
Combining~\Cref{eq:adv_online_ub_1,eq:adv_online_ub_2,eq:adv_online_ub_3}, we obtain
\begin{equation}\label{eq:adv_online_ub_4}
N_\greedy(\tau+s\mid\inst)
<
N_\greedy(\tau\mid\inst)+\capac .
\end{equation}

We now derive a contradiction. By~\Cref{eq:adv_online_ub_1}, there exists a job $j$ that is completed by $\sched^*$ by time $\tau$ but is not completed by $\greedy$ by time $\tau+s$. Since $j$ is completed by $\sched^*$ by time $\tau$, it must have arrived by time $\tau-s$. Moreover, $j$ cannot be contained in any greedy batch that starts at or before time $\tau$, since such a batch would complete by time $\tau+s$. Thus, job $j$ is pending for greedy throughout the interval $[\tau-s,\tau]$.

By definition of the greedy scheduler, this can only happen if greedy is continuously busy during this interval. Moreover, any greedy batch that starts during $[\tau-s,\tau]$ must be full; otherwise, it would include job $j$. Hence greedy completes a full batch of $\capac$ jobs during the interval $[\tau,\tau+s]$. This implies
$$
N_\greedy(\tau+s\mid\inst)
=
N_\greedy(\tau\mid\inst)+\capac,
$$
contradicting~\Cref{eq:adv_online_ub_4}. Therefore,
$$
N_\greedy(t+s\mid\inst)\ge N_{\sched^*}(t\mid\inst)
$$
for all $t\ge 0$.

Using~\Cref{claim:flow-def}, we have
$$
\greedy(\inst)-\opt(\inst)
=
\int_{t=0}^\infty
\big(N_{\sched^*}(t\mid\inst)-N_\greedy(t\mid\inst)\big)\,dt .
$$
By the domination inequality above,
$$
\greedy(\inst)-\opt(\inst)
\le
\int_{t=0}^\infty
\big(N_\greedy(t+s\mid\inst)-N_\greedy(t\mid\inst)\big)\,dt
=
s|\inst|.
$$
The final equality holds because each completed job contributes to the integrand for exactly $s$ units of time. By~\Cref{claim:opt-lb}, $s|\inst|\le \opt(\inst)$. Hence $\greedy(\inst)\le 2\cdot \opt(\inst)$ as desired.
\end{proof}

\begin{lemma}\label{lemma:adv_det_lb}
For any $\epsilon>0$ and any deterministic online scheduler $\sched$, there exists an instance $\inst$ such that $\sched(\inst)\ge (2-\epsilon)\cdot\opt(\inst)$.
\end{lemma}

\begin{proof}
Fix the batch capacity $\capac\ge 2$ and the batch processing time $s>0$. Let $\sched$ be any deterministic online scheduler. We construct an adaptive adversarial instance $\inst$ against $\sched$. The instance starts with a single job at time $r_1=0$. Let $\Delta_1$ be the amount of time $\sched$ waits before starting a batch containing this job, assuming no other jobs arrive. If $\Delta_1\ge s$, then the adversary releases no further jobs. In that case, $\sched(\inst)=\Delta_1+s\ge 2s$, whereas the offline optimum processes the job immediately and pays $\opt(\inst)=s$. Thus the desired lower bound follows. We may therefore assume $\Delta_1<s$.

Now fix an arbitrary number of jobs $n$ and a small parameter $\delta>0$. The adversary continues as follows. Whenever $\sched$ starts processing the currently pending job, the adversary releases the next job after $\delta$ time units. Thus, if $\sched$ starts processing job $i$ at time $t_i$, with $t_1:=\Delta_1$, then job $i+1$ is released at time $r_{i+1}=t_i+\delta$ for every $i=1,\ldots,n-1$. This ensures that job $i+1$ arrives just after $\sched$ has committed to a batch containing only job $i$.

Under this construction, the first job waits $t_1=\Delta_1$ time units before entering service. For each $i\ge 2$, job $i$ arrives at time $t_{i-1}+\delta$, while the previous batch cannot complete before time $t_{i-1}+s$. Hence job $i$ waits at least $s-\delta$ before it can even begin processing. Accounting also for the $s$ units of processing time incurred by every job, we obtain
$$
\sched(\inst)
\ge
ns+\Delta_1+(n-1)(s-\delta).
$$

On the other hand, consider the following feasible offline schedule. The offline scheduler waits until job $2$ arrives, batches jobs $1$ and $2$ together, and then processes each remaining job individually at its release time. We claim that this is feasible. Indeed, since $\sched$ starts a singleton batch containing job $i$ at time $t_i$, it cannot start a batch containing job $i+1$ before the previous batch has completed. Therefore $t_i\ge t_{i-1}+s$ for every $i\ge 2$. Using $r_{i+1}=t_i+\delta$, this implies $r_{i+1}\ge r_i+s$ for every $i\ge 2$. Thus, after the offline scheduler completes the first batch containing jobs $1$ and $2$ at time $r_2+s$, job $3$ has not yet arrived; and more generally, if the offline scheduler processes job $i$ alone at time $r_i$, then it completes at time $r_i+s\le r_{i+1}$. Finally, the maximum batch size used by this schedule is $2\le \capac$, so the schedule is feasible. Its total flow time is
$$
ns+\Delta_1+\delta \ge \opt(\inst).
$$
Consequently,
$$
\frac{\sched(\inst)}{\opt(\inst)}
\ge
\frac{ns+\Delta_1+(n-1)(s-\delta)}
{ns+\Delta_1+\delta}
\ge
\frac{n(2s-\delta)-(s-\delta)}
{ns+s+\delta}.
$$
As $n\to\infty$, this lower bound tends to $2-\delta/s$. Taking $\delta<\epsilon s$ and then taking $n$ sufficiently large completes the proof.
\end{proof}

\begin{lemma}\label{lemma:adv_ran_lb}
For any $\epsilon>0$ and any randomized online scheduler $\sched$, there exists an instance $\inst$ such that $\sched(\inst)\ge (11/9-\epsilon)\cdot \opt(\inst)$.
\end{lemma}

\begin{proof}
Fix the batch capacity $\capac\ge 2$ and the batch processing time $s>0$. We will use parameters $q\in\{1,2,\ldots,\capac-1\}$, $x\in(0,1)$, and $M\ge 1$, to be chosen later. Define the two arrival sequences
$$
\inst_0
:=
\{\text{$q$ jobs at $t=0$}\}
\cup
\{\text{$\capac-q$ jobs at $t=xs$}\}
\cup
\{\text{$q$ jobs at $t=s$}\}
\cup
\bigcup_{i=2}^{M+1}
\{\text{$\capac$ jobs at $t=is$}\},
$$
and
$$
\inst_x
:=
\{\text{$q$ jobs at $t=0$}\}
\cup
\{\text{$\capac-q$ jobs at $t=xs$}\}
\cup
\bigcup_{i=1}^{M}
\{\text{$\capac$ jobs at $t=(i+x)s$}\}.
$$
Pictorially, we depict these instances in~\Cref{fig:randomized-lb-instances}. At a high level, we will show that on $\inst_0$ the optimal schedule starts a batch at $t=0$ containing the first $q$ jobs, whereas on $\inst_x$ it is better to wait until time $xs$ and start a full batch of size $\capac$. However, the two instances are indistinguishable before time $s$, so an online scheduler cannot know which of these two first actions is better.
% Preamble:
% \usepackage{tikz}
% \usetikzlibrary{decorations.pathreplacing}
% \usepackage{float} % needed if you use [H]

\begin{figure}[H]
\centering
\begin{tikzpicture}[
    x=1cm,
    y=1cm,
    tick/.style={black, line width=0.8pt},
    axis/.style={black, line width=0.8pt},
    arr/.style={font=\small},
    timelab/.style={font=\small},
    title/.style={font=\large},
    prefixstyle/.style={draw=red!70, text=red!70},
    zeroinst/.style={draw=green!55!black, text=green!55!black},
    xinst/.style={draw=blue!70, text=blue!70},
    brace/.style={decorate, decoration={brace, amplitude=5pt}, line width=0.8pt}
]

% y-coordinates
\def\yA{0}
\def\yB{-3.0}

%%%%%%%%%%%%%%%%%%%%%%%%%%%%%%%%%%%%%%%%%%%%%%%%%%%%%%%%%%%%
% Top timeline: instance I_0
%%%%%%%%%%%%%%%%%%%%%%%%%%%%%%%%%%%%%%%%%%%%%%%%%%%%%%%%%%%%

\node[title, anchor=west] at (-1.0,{\yA+1.75}) {\small The instance $\inst_0$:};

% Axis
\draw[axis] (0.7,\yA) -- (12.6,\yA);

% Left labels
\node[arr, red!70, anchor=east] at (0.65,{\yA+0.36}) {arrivals};
\node[timelab, anchor=east] at (0.65,{\yA-0.35}) {time};

% Ticks and time labels
\foreach \x/\lab in {
    1.0/{\scriptsize $0$},
    1.8/{\scriptsize $xs$},
    2.8/{\scriptsize $s$},
    4.6/{\scriptsize $2s$},
    12.0/{\scriptsize $(M+1)s$}
}{
    \draw[tick] (\x,{\yA-0.12}) -- (\x,{\yA+0.12});
    \node[timelab, below=3pt] at (\x,\yA) {\lab};
}

% Extra unlabeled ticks to suggest continuation
\draw[tick] (6.4,{\yA-0.12}) -- (6.4,{\yA+0.12});

% Ellipsis
\node at (8.7,{\yA+0.22}) {\large $\cdots$};

% Arrival labels: common prefix
\node[arr, red!70] at (1.0,{\yA+0.43}) {\scriptsize $q$};
\node[arr, red!70] at (1.8,{\yA+0.43}) {\scriptsize $\capac-q$};

% Arrival labels: I_0 continuation
\node[arr, green!55!black] at (2.8,{\yA+0.43}) {\scriptsize $q$};
\node[arr, green!55!black] at (4.6,{\yA+0.43}) {\scriptsize $\capac$};
\node[arr, green!55!black] at (6.4,{\yA+0.43}) {\scriptsize $\capac$};
\node[arr, green!55!black] at (12.0,{\yA+0.43}) {\scriptsize $\capac$};

% Braces
\draw[brace, prefixstyle] (1.0,{\yA+0.82}) -- (1.8,{\yA+0.82})
    node[midway, above=6pt] {Common Prefix};

\draw[brace, zeroinst] (2.8,{\yA+0.82}) -- (12.0,{\yA+0.82})
    node[midway, above=6pt] {Bad continuation if $\sched$ batches late};

%%%%%%%%%%%%%%%%%%%%%%%%%%%%%%%%%%%%%%%%%%%%%%%%%%%%%%%%%%%%
% Bottom timeline: instance I_x
%%%%%%%%%%%%%%%%%%%%%%%%%%%%%%%%%%%%%%%%%%%%%%%%%%%%%%%%%%%%

\node[title, anchor=west] at (-1.0,{\yB+1.75}) {\small The instance $\inst_x$:};

% Axis
\draw[axis] (0.7,\yB) -- (12.6,\yB);

% Left labels
\node[arr, red!70, anchor=east] at (0.65,{\yB+0.36}) {arrivals};
\node[timelab, anchor=east] at (0.65,{\yB-0.35}) {time};

% Ticks and time labels
\foreach \x/\lab in {
    1.0/{\scriptsize $0$},
    1.8/{\scriptsize $xs$},
    2.8/{\scriptsize $s$},
    3.7/{\scriptsize $(1+x)s$},
    5.5/{\scriptsize $(2+x)s$},
    11.4/{\scriptsize $(M+x)s$}
}{
    \draw[tick] (\x,{\yB-0.12}) -- (\x,{\yB+0.12});
    \node[timelab, below=3pt] at (\x,\yB) {\lab};
}

% Ellipsis
\node at (8.7,{\yB+0.22}) {\large $\cdots$};

% Arrival labels: common prefix
\node[arr, red!70] at (1.0,{\yB+0.43}) {\scriptsize $q$};
\node[arr, red!70] at (1.8,{\yB+0.43}) {\scriptsize $\capac-q$};

% Arrival labels: I_x continuation
\node[arr, blue!70] at (3.7,{\yB+0.43}) {\scriptsize $\capac$};
\node[arr, blue!70] at (5.5,{\yB+0.43}) {\scriptsize $\capac$};
\node[arr, blue!70] at (11.4,{\yB+0.43}) {\scriptsize $\capac$};

% Braces
\draw[brace, prefixstyle] (1.0,{\yB+0.82}) -- (1.8,{\yB+0.82})
    node[midway, above=6pt] {Common Prefix};

\draw[brace, xinst] (3.7,{\yB+0.82}) -- (11.4,{\yB+0.82})
    node[midway, above=6pt] {Bad continuation if $\sched$ batches early};

\end{tikzpicture}
\caption{The two adversarial instances used in the proof of \Cref{lemma:adv_ran_lb}.}
\label{fig:randomized-lb-instances}
\end{figure}

We first compute the asymptotic costs of the two canonical schedules. Let $\sched_0$ denote the schedule that starts with a batch of the first $q$ jobs at time $0$, and let $\sched_x$ denote the schedule that waits until time $xs$ and starts with a full batch of size $\capac$.

Consider first the instance $\inst_0$. Under $\sched_0$, the first $q$ jobs are processed at time $0$, the next $\capac$ pending jobs are processed at time $s$, and then each of the $M$ full bursts arriving at times $is$, for $i=2,\ldots,M+1$, is processed immediately upon arrival. Therefore,
$$
\sched_0(\inst_0)
=
\big(q+(M+1)\capac\big)s
+
(\capac-q)(1-x)s .
$$
Under $\sched_x$, the schedule waits until time $xs$ and processes the first $\capac$ jobs as a full batch, completing it at time $(1+x)s$. At that point the $q$ jobs arriving at time $s$ have also arrived. The flow time incurred up to this point is $\capac s+2qxs$. From that point on, the schedule either processes these $q$ leftover jobs immediately and delays all future full bursts by $xs$, or carries them until the end. Choosing the better of the two options gives
$$
\sched_x(\inst_0)
=
\big(q+(M+1)\capac\big)s
+
2qxs
+
\min\big\{M\capac xs,\ q(M+1-x)s\big\}.
$$

Now consider the instance $\inst_x$. Under $\sched_x$, the schedule is aligned with the $M$ future full bursts, and the only waiting beyond processing time is the waiting of the first $q$ jobs until time $xs$. Hence
$$
\sched_x(\inst_x)
=
(M+1)\capac s+qxs .
$$
Under $\sched_0$, after the first batch at time $0$, the schedule has $\capac-q$ leftover jobs at time $s$. It can either process these leftover jobs immediately and delay all future bursts by $(1-x)s$, or carry them until the end. Thus
$$
\sched_0(\inst_x)
=
(M+1)\capac s
+
(\capac-q)(1-x)s
+
\min\big\{M\capac(1-x)s,\ (\capac-q)(M+x)s\big\}.
$$

Set
$$
q=\lfloor \capac/2\rfloor
\qquad\text{and}\qquad
x=\frac{q}{\capac}.
$$
As $M\to\infty$, both $\sched_0(\inst_0)$ and $\sched_x(\inst_x)$ are asymptotic to $\capac sM$. Moreover,
$$
\frac{\sched_x(\inst_0)}{\sched_0(\inst_0)}
\to
1+x
\qquad\text{and}\qquad
\frac{\sched_0(\inst_x)}{\sched_x(\inst_x)}
\to
2-x.
$$

We now invoke Yao's minmax principle. Consider an oblivious distribution that selects $\inst_0$ with probability $p$ and $\inst_x$ with probability $1-p$. It suffices to lower-bound the expected competitive ratio of every deterministic online scheduler against this distribution.

Before time $s$, the two instances are indistinguishable. Thus any deterministic online scheduler either starts a batch before time $xs$, or waits at least until time $xs$. In the first case, its asymptotic cost is no smaller than that of $\sched_0$ on both instances. In the second case, its asymptotic cost is no smaller than that of $\sched_x$ on both instances. Therefore, using the asymptotic cost computations above, for every $\epsilon>0$ we can choose $M$ large enough so that every deterministic online scheduler has expected ratio at least
$$
\min\big\{
p+(1-p)(2-x),\ p(1+x)+(1-p)
\big\}
-\epsilon .
$$
Choosing $p=1-x$ equalizes the two terms and gives
$$
1+x-x^2-\epsilon .
$$
Since $x=\lfloor \capac/2\rfloor/\capac$, this equals
$$
1+\frac{\lfloor \capac/2\rfloor}{\capac}
-
\bigg(\frac{\lfloor \capac/2\rfloor}{\capac}\bigg)^2
-\epsilon.
$$
For even $\capac$, this is $5/4-\epsilon$. For odd $\capac$, it is
$$
\frac{5}{4}-\frac{1}{4\capac^2}-\epsilon.
$$
The minimum over all $\capac\ge 2$ is attained at $\capac=3$ and equals $11/9-\epsilon$. Hence, for every randomized online scheduler, there exists an instance on which its expected competitive ratio is at least $11/9-\epsilon$.
\end{proof}

\subsection{The Analysis of RandomShift (\Cref{thm:adv-lookahead-uncap-upper})}\label{appendix:adv-lookahead-upper}

In this section, we give the full analysis of algorithm $\advname$ in the uncapacitated $\capac=\infty$ setting and prove \Cref{thm:adv-lookahead-uncap-upper}. 

\randomshift*
\begin{proof}
Fix any instance $\inst$, any batch processing time $s>0$ and any lookahead $L>0$. Furthermore, let $\inst_\act$ be the sub-instance of active jobs, as defined by $\advname$, and $\inst_\ign$ be the sub-sequence of ignored jobs. Note that these sequences are random, as they are defined based on the realization of the random offset $\tau$ sampled at the beginning of $\advname$. Since the random offset $\tau$ induces a uniformly random shift of the period $L+2s$, each arrival time falls in an ignored interval with probability $2s/(L+2s)$, and in an active interval with probability $L/(L+2s)$. Therefore,
\begin{equation}\label{eq:random_shift_split}
  \expect{\lvert \inst_\ign\rvert}
=
\frac{2s}{L+2s}\lvert \inst\rvert
\qquad \text{and} \qquad
\expect{\lvert \inst_\act\rvert}
=
\frac{L}{L+2s}\lvert \inst\rvert .  
\end{equation}

Now fix any realization of the threshold $\tau$ and therefore the sub-instances $\inst_\act$ and $\inst_\ign$. By Lemma~\ref{lemma:superadditive} we have that 
\begin{equation}\label{eq:random_shift_opt}\opt(\inst) \ge \opt(\inst_\act) + \opt(\inst_\ign).
\end{equation}
We also separately charge the flow time suffered by $\advname$ on jobs from $\inst_\act$ and $\inst_\ign$ respectively, to write 
\begin{equation}\label{eq:random_shift_alg}\advname(\inst) = F_\act + F_\ign.
\end{equation}
Here, $F_\act$ (resp. $F_\ign$) denotes the total flow time paid by $\advname$ due to the jobs in the sub-sequence $\inst_\act$ (resp. $\inst_\ign$).

We first observe that, by construction,
\begin{equation}\label{eq:random_shift_alg_1}
F_\act=\opt(\inst_\act).
\end{equation}
This is true as each active interval comprising $\inst_\act$ is processed via an offline optimal schedule on its jobs, and these schedules never overlap with each other: indeed, any offline optimal schedule for a sub-instance of maximum arrival time $r$ will complete by $t=r+2s$ as in the uncapacitated setting ($\capac=\infty$), it is always optimal to include all pending jobs when starting a batch, each job has arrived at $t=r$ and the last batch that started before $t=r$ will finish by $t=r+s$. Therefore, $\advname$ is indeed able to optimally process each active interval before the next one starts, due to allowing sufficient spacing ($2s$ time units in particular) between consecutive active intervals.

Combining \eqref{eq:random_shift_opt}, \eqref{eq:random_shift_alg} and \eqref{eq:random_shift_alg_1}, we conclude that 
\begin{equation}\label{eq:random_shift_reduced}
   \advname(\inst) \le \opt(\inst) + F_\ign - \opt(\inst_\ign). 
\end{equation}

Furthermore, $\advname$ ensures that no ignored job remains in the system for more than $4s$ units of time. Indeed, even if a job arrives at the start of the ignored interval (which has length $2s$), by construction of $\advname$ it will start processing at most $s$ units of time after the start of the next active interval (and then take another $s$ units of time to process). At the same time, we know from~\Cref{claim:opt-lb} that $\opt(\inst_\ign)\ge s\cdot \lvert\inst_\ign\rvert$. Therefore, we have that 

\begin{equation}\label{eq:random_shift_ignored}
    F_\ign - \opt(\inst_\ign) \le 4s\cdot \lvert\inst_\ign\rvert - s\cdot \lvert\inst_\ign\rvert = 3s\cdot \lvert\inst_\ign\rvert .
\end{equation}

Combining \eqref{eq:random_shift_reduced} with \eqref{eq:random_shift_ignored} we finally have that for every realization of the offset $\tau$, the total flow time of $\advname$ satisfies
\[\advname(\inst) \le \opt(\inst) + 3s\cdot \lvert\inst_\ign\rvert \]
with probability $1$, and by taking expectations and using \eqref{eq:random_shift_split} we finally have
\[\advname(\inst) \le \opt(\inst) + \frac{6s^2}{L+2s}\cdot \lvert\inst\rvert .\]
The proof is completed by using~\Cref{claim:opt-lb} once more, stating that $\opt(\inst)\ge s\cdot \lvert\inst\rvert.$
\end{proof}

\subsection{Lower Bound for Uncapacitated Lookahead Algorithms (\Cref{thm:adv-lookahead-uncap-lower})}\label{appendix:adv-uncap-lookahead-lb}

\advuncaplb*
\begin{proof}
    
Up to constants, we can assume without loss that $L=ks$, where $k\ge 5$ is an integer. Fix such a $k$, and set
$$
\epsilon:=\frac{s}{4k}
\qquad\text{and}\qquad
\tau:=\frac{3s}{4k}.
$$
We construct two instances. The first instance is
$$
\inst_1:=\{0,\tau,s,2s,\ldots,(k-1)s\}.
$$
The second instance $\inst_2$ is identical to $\inst_1$, except that it contains $k$ additional jobs arriving at time $ks+\epsilon$. The two instances are indistinguishable to an $L=ks$-lookahead scheduler at every time $t<\epsilon$, since the only arrivals on which they differ occur at time $ks+\epsilon$. Pictorially, we depict these instances in~\Cref{fig:uncap-lookahead-lower-bound}. 
% Preamble:
% \usepackage{tikz}
% \usetikzlibrary{decorations.pathreplacing}
% \usepackage{float} % needed if you use [H]

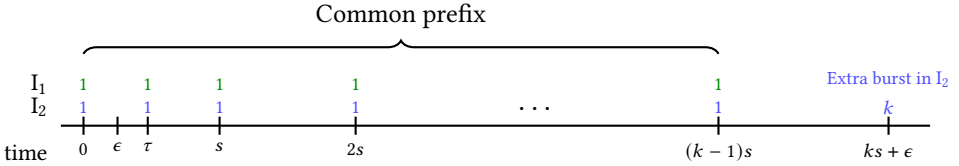
\begin{figure}[H]
\centering
\begin{tikzpicture}[
    x=1cm,
    y=1cm,
    tick/.style={black, line width=0.8pt},
    axis/.style={black, line width=0.8pt},
    arr/.style={font=\small},
    timelab/.style={font=\small},
    instone/.style={text=green!55!black},
    insttwo/.style={text=blue!70},
    brace/.style={decorate, decoration={brace, amplitude=5pt}, line width=0.8pt}
]

% baseline
\def\yA{0}

% Main equally-spaced time points
\def\xzero{1.0}
\def\xs{2.8}
\def\xtwos{4.6}
\def\xkmone{9.4}

% Intermediate points between 0 and s
\def\xeps{1.45}
\def\xtau{1.85}

% Slightly after the (hidden) ks point
\def\xkseps{11.65}

% Axis
\draw[axis] (0.7,\yA) -- (12.2,\yA);

% Left labels
\node[arr, anchor=east] at (0.65,{\yA+0.55}) {$\inst_1$};
\node[arr, anchor=east] at (0.65,{\yA+0.25}) {$\inst_2$};
\node[timelab, anchor=east] at (0.65,{\yA-0.35}) {time};

% Ticks and time labels
\foreach \x/\lab in {
    \xzero/{\scriptsize $0$},
    \xeps/{\scriptsize $\epsilon$},
    \xtau/{\scriptsize $\tau$},
    \xs/{\scriptsize $s$},
    \xtwos/{\scriptsize $2s$},
    \xkmone/{\scriptsize $(k-1)s$},
    \xkseps/{\scriptsize $ks+\epsilon$}
}{
    \draw[tick] (\x,{\yA-0.12}) -- (\x,{\yA+0.12});
    \node[timelab, below=3pt] at (\x,\yA) {\lab};
}

% Ellipsis
\node at (7.0,{\yA+0.18}) {\large $\cdots$};

% I_1 arrivals
\node[arr, instone] at (\xzero,{\yA+0.55}) {\scriptsize $1$};
\node[arr, instone] at (\xtau,{\yA+0.55}) {\scriptsize $1$};
\node[arr, instone] at (\xs,{\yA+0.55}) {\scriptsize $1$};
\node[arr, instone] at (\xtwos,{\yA+0.55}) {\scriptsize $1$};
\node[arr, instone] at (\xkmone,{\yA+0.55}) {\scriptsize $1$};

% I_2 arrivals
\node[arr, insttwo] at (\xzero,{\yA+0.25}) {\scriptsize $1$};
\node[arr, insttwo] at (\xtau,{\yA+0.25}) {\scriptsize $1$};
\node[arr, insttwo] at (\xs,{\yA+0.25}) {\scriptsize $1$};
\node[arr, insttwo] at (\xtwos,{\yA+0.25}) {\scriptsize $1$};
\node[arr, insttwo] at (\xkmone,{\yA+0.25}) {\scriptsize $1$};
\node[arr, insttwo] at (\xkseps,{\yA+0.25}) {\scriptsize $k$};

% Burst label
\node[arr, insttwo] at (\xkseps,{\yA+0.62}) {\scriptsize Extra burst in $\inst_2$};

% Braces
\draw[brace] (\xzero,{\yA+0.95}) -- (\xkmone,{\yA+0.95})
    node[midway, above=6pt] {Common prefix};

\end{tikzpicture}
\caption{The two adversarial instances used in the uncapacitated lookahead lower bound.}
\label{fig:uncap-lookahead-lower-bound}
\end{figure}

We first record simple upper bounds on the offline optimum. For $\inst_1$, consider the schedule that starts batches at
$$
\tau,\ s+\tau,\ 2s+\tau,\ldots,(k-1)s+\tau.
$$
The first batch contains the jobs released at times $0$ and $\tau$, and each subsequent batch contains the next regularly spaced job. Thus, every job contributes $s$ units of processing time, and the total waiting time is $k\tau$. Therefore,
\begin{equation}\label{eq:uncap_lb_opt1}
\opt(\inst_1)\le (k+1)s+k\tau.
\end{equation}
For $\inst_2$, consider the schedule that starts batches at
$$
0,\ s,\ 2s,\ldots,(k-1)s,\ ks+\epsilon.
$$
The only waiting time in this schedule is $s-\tau$, incurred by the job released at time $\tau$. Hence
\begin{equation}\label{eq:uncap_lb_opt2}
\opt(\inst_2)\le (2k+1)s+(s-\tau)=(2k+2)s-\tau.\end{equation}

Now fix an arbitrary deterministic $L$-lookahead scheduler $\sched$. Since $\inst_1$ and $\inst_2$ are indistinguishable before time $\epsilon$, the scheduler behaves identically on the two instances until that time. Also, because $\capac=\infty$, we may assume without loss that whenever $\sched$ starts a batch, it includes all currently pending jobs; omitting a pending job can only delay that job and cannot improve the schedule. We distinguish two cases.

First, suppose that $\sched$ starts a batch at some time $a<\epsilon$. We give the scheduler instance $\inst_1$. Since $a<\epsilon<\tau$, the first batch can contain only the job released at time $0$. The processor is then busy until time $a+s$. Consequently, the job released at time $\tau$ cannot start before time $a+s$. Therefore, the jobs released at times $0$ and $\tau$ alone incur total waiting at least
$$
a+(a+s-\tau)\ge s-\tau.
$$
Since $\inst_1$ contains $k+1$ jobs, the total processing contribution is $(k+1)s$. Hence
\begin{equation}\label{eq:uncap_lb_alg1}
    \sched(\inst_1)
    \ge
    (k+1)s+(s-\tau)
    =
    (k+2)s-\tau.
\end{equation}
Combining~\eqref{eq:uncap_lb_opt1} and~\eqref{eq:uncap_lb_alg1} and substituting $\tau=3s/(4k)$, we have
\begin{equation}\label{eq:uncap_lb_lb1}
    \frac{\sched(\inst_1)}{\opt(\inst_1)}
    \ge
    \frac{(k+2)s-\tau}{(k+1)s+k\tau}
    =
    1+\frac{k-3}{k(4k+7)}
\end{equation}
if $\sched$ starts its first batch at $t=\alpha < \epsilon$.

Second, suppose that $\sched$ does not start any batch before time $\epsilon$. We give the scheduler instance $\inst_2$. We show that in this case the total waiting time is at least $s+\tau$. Let $a\ge\epsilon$ be the starting time of the first batch. We distinguish between the following cases:

\begin{enumerate}
    \item If $a<\tau$, then the first batch contains only the job released at time $0$. The processor is busy until time $a+s$, so the job released at time $\tau$ cannot start before time $a+s$. Now consider the regularly spaced jobs released at times $s,2s,\ldots,(k-1)s$. If any one of these jobs is not started before the next arrival time in the sequence
    $$
    s,2s,\ldots,(k-1)s,ks+\epsilon,
    $$
    then that job alone waits at least $s$. Together with the waiting of the jobs released at times $0$ and $\tau$, this gives total waiting at least
    $$
    a+(a+s-\tau)+s
    \ge
    2s-\tau+2\epsilon
    \ge
    s+\tau,
    $$
    where the last inequality uses $\tau=3\epsilon$ and $k\ge 1$.
    
    Otherwise, every regularly spaced job is started before the next arrival time in the above sequence. Since the first batch completes at time $a+s$, a simple induction implies that each job released at time $is$, for $i=1,\ldots,k-1$, waits at least $a$, and the final $k$ jobs released at time $ks+\epsilon$ wait at least $a-\epsilon$ each. Therefore, the total waiting time is at least
    $$
    a+(a+s-\tau)+(k-1)a+k(a-\epsilon)
    =
    s-\tau+(2k+1)a-k\epsilon.
    $$
    Using $a\ge\epsilon$ and $\tau=3\epsilon$, this is at least
    $$
    s-\tau+(k+1)\epsilon
    =
    s+(k-2)\epsilon
    \ge
    s+\tau,
    $$
    where the final inequality uses $k\ge 5$. Therefore, if $\alpha\in [\epsilon,\tau)$, the total waiting time of $\sched$ on $\inst_2$ is at least $s+\tau$.

    \item Next suppose that $a\in[\tau,s)$. Then the first batch contains the jobs released at times $0$ and $\tau$. If some regularly spaced job released at time $is$, for $i\in{1,\ldots,k-1}$, is not started before the next arrival time in the sequence
    $$
    s,2s,\ldots,(k-1)s,ks+\epsilon,
    $$
    then that job waits at least $s$. The two jobs in the first batch incur total waiting $a+(a-\tau)\ge\tau$, so the total waiting is at least $s+\tau$.
    
    Otherwise, every regularly spaced job is started before the next arrival time in the above sequence. Since the first batch completes at time $a+s$, an induction implies that each job released at time $is$, for $i=1,\ldots,k-1$, waits at least $a$, and the final $k$ jobs released at time $ks+\epsilon$ wait at least $a-\epsilon$ each. Hence the total waiting time is at least
    $$
    a+(a-\tau)+(k-1)a+k(a-\epsilon)
    =
    (2k+1)a-\tau-k\epsilon.
    $$
    Using $a\ge\tau$ and $\tau=3\epsilon$, this is at least
    $$
    2k\tau-k\epsilon
    =
    5k\epsilon
    =
    \frac{5}{4}s
    \ge
    s+\tau,
    $$
    where the final inequality uses $k\ge 3$. Therefore, if $\alpha\in [\tau,s)$, the total waiting time of $\sched$ on $\inst_2$ is at least $s+\tau$.

    \item Finally, suppose that $a\ge s$. Then the jobs released at times $0$ and $\tau$ alone incur waiting at least
    $$
    s+(s-\tau)=2s-\tau.
    $$
    Since $\tau\le s/2$, this is at least $s+\tau$. Therefore, if $\alpha\ge s$, the total waiting time of $\sched$ on $\inst_2$ is at least $s+\tau$.
\end{enumerate}

Thus, whenever $\sched$ does not start a batch before time $\epsilon$, its total waiting time on $\inst_2$ is at least $s+\tau$. Since $\inst_2$ contains $2k+1$ jobs, we have

\begin{equation}\label{eq:uncap_lb_alg2}
    \sched(\inst_2)
    \ge
    (2k+1)s+(s+\tau)
    =
    (2k+2)s+\tau.
\end{equation}

Combining~\eqref{eq:uncap_lb_opt2} and~\eqref{eq:uncap_lb_alg2} and substituting $\tau=3s/(4k)$, we have
\begin{equation}\label{eq:uncap_lb_lb2}
    \frac{\sched(\inst_2)}{\opt(\inst_2)}
    \ge
    \frac{(2k+2)s+\tau}{(2k+2)s-\tau}
    =
    1+\frac{6}{8k^2+8k-3}.
\end{equation}

Finally, combining~\eqref{eq:uncap_lb_lb1} and~\eqref{eq:uncap_lb_lb2}, for every deterministic $L=ks$-lookahead scheduler, there exists an instance $\inst\in\{\inst_1,\inst_2\}$ such that
$$
\frac{\sched(\inst)}{\opt(\inst)}
\ge
1+
\min\left\{
\frac{k-3}{k(4k+7)},
\frac{6}{8k^2+8k-3}
\right\}.
$$
For every $k\ge 5$, both terms inside the minimum are at least $1/(4k^2)$. Therefore,
$$
\frac{\sched(\inst)}{\opt(\inst)}
\ge
1+\frac{1}{4k^2}.
$$
Since $k=L/s$, this gives
$$
\frac{\sched(\inst)}{\opt(\inst)}
\ge
1+\Omega\left(\frac{s^2}{L^2}\right)
$$
for $L=ks$. For an arbitrary lookahead value $L\ge s$, choose an integer
$$
k:=\max\left\{5,\left\lceil \frac{L}{s}\right\rceil\right\}.
$$
Then $ks\ge L$, so the two instances above remain indistinguishable before time $\epsilon$ even to an $L$-lookahead scheduler. Moreover, $k=O(L/s)$ for $L\ge s$. Applying the lower bound above with this value of $k$ gives
$$
\frac{\mathbb E[\sched(\inst)]}{\opt(\inst)}
\ge
1+\Omega\left(\frac{1}{k^2}\right)
=
1+\Omega\left(\frac{s^2}{L^2}\right),
$$
as claimed.

Finally, if $\sched$ is a randomized schedule, then we can instead offer a uniformly random selection between $\inst_1$ and $\inst_2$; for every realization of $\sched$'s randomness, and in particular the first batching time $a$, with probability $1/2$ it will draw the wrong instance, and therefore with the same arguments we can derive
$$
\frac{\sched(\inst)}{\opt(\inst)}
\ge
1+\frac{1}{8k^2}.
$$
\end{proof}

\subsection{Lower Bound for Capacitated Lookahead Algorithms (Theorem~\ref{thm:adv-lookahead-cap-lb})}
\label{sec:adv-cap-lb}

\advcaplb*
\begin{proof}
    
Fix any batch capacity $\capac\ge 2$, any batch processing time $s>0$, and any finite lookahead $L>0$. Let
\[
    k:=\lceil L/s\rceil,
\]
so that $ks\ge L$. Finally, fix a real parameter $\delta\in (0,1)$ and integer parameters $q\in \{1,\ldots,C-1\}$ and $M\ge 1$, to be specified later. We define three arrival sequences:
\begin{enumerate}
    \item $\inst_{\mathrm{prefix}}:=
    \{\text{$q$ arrivals at $t=0$}\}
    \cup
    \{\text{$\capac-q$ arrivals at $t=\delta s$}\}\cup\bigcup_{i=1}^k \{\text{$\capac$ arrivals at $t=is$}\}$,
    \item $\inst_{\mathrm{early}}:=
    \{\text{$q$ arrivals at $t=(k+1)s$}\}
    \cup
    \bigcup_{i=2}^{M+1} \{\text{$\capac$ arrivals at $t=(k+i)s$}\}$, and,
    \item $\inst_{\mathrm{late}}:=
    \bigcup_{i=1}^{M} \{\text{$\capac$ arrivals at $t=(k+i+\delta)s$}\}$.
\end{enumerate}
These define two instances
\[
    \inst_{\mathrm{E}}:=\inst_{\mathrm{prefix}}\cup \inst_{\mathrm{early}}
    \qquad\text{and}\qquad
    \inst_{\mathrm{L}}:=\inst_{\mathrm{prefix}}\cup \inst_{\mathrm{late}}.
\]
The two instances are indistinguishable to an $L$-lookahead scheduler up to time $t=\delta s$: indeed, since $L\le ks$, at time $\delta s$ the scheduler can see only up to time $\delta s+L\le (k+\delta)s<(k+1)s$, while the two continuations first differ at time $(k+1)s$. Thus the scheduler must decide whether to start a batch before time $\delta s$ without knowing whether the continuation will be $\inst_{\mathrm{early}}$ or $\inst_{\mathrm{late}}$. The instances are designed so that $\inst_{\mathrm{E}}$ rewards batching the first $q$ jobs early, while $\inst_{\mathrm{L}}$ rewards waiting until time $\delta s$ and starting with a full batch of size $\capac$.

First consider the optimal schedules for the two instances. As $M\to\infty$, the flow time incurred before the first of the $M$ repeated bursts is lower-order, and the asymptotic cost is determined by the steady-state behavior on those bursts. Since each burst contains $\capac$ jobs and consecutive bursts are spaced exactly $s$ time units apart, any scheduler incurs at least $s$ flow time per burst asymptotically. This lower bound is achievable in both instances if the correct initial decision is made.

In instance $\inst_{\mathrm{E}}$, start a batch at time $0$ containing the first $q$ jobs. Then, for each $i\in[k]$, start a batch at time $is$ containing the $\capac$ jobs that arrive at time $is$. Finally, at time $(k+1)s$, start a batch containing the $\capac-q$ jobs that arrived at time $\delta s$ together with the $q$ jobs that arrive at time $(k+1)s$. Thus the system is empty by time $(k+2)s$, when the first burst of $\capac$ jobs arrives. Thereafter, each burst can be processed immediately upon arrival.

In instance $\inst_{\mathrm{L}}$, start a batch at time $\delta s$ containing the first $\capac$ arrivals. Then, for each $i\in[k]$, start a batch at time $(i+\delta)s$ containing the $\capac$ jobs that arrived at time $is$. Thus the system is empty by time $(k+1+\delta)s$, when the first burst of $\capac$ jobs arrives. Thereafter, each burst can again be processed immediately upon arrival.

Therefore,
\begin{equation}\label{eq:adv-lb-optimal-asymptotic}
    \lim_{M\to\infty}\frac{\opt(\inst_{\mathrm{E}})}{M}
    =
    \lim_{M\to\infty}\frac{\opt(\inst_{\mathrm{L}})}{M}
    =
    Cs.
\end{equation}

We now analyze the best possible performance subject to making the wrong initial decision.

First suppose that in instance $\inst_{\mathrm{E}}$ the scheduler does not start a batch before time $\delta s$. Let $\opt(\inst_{\mathrm{E}}\mid \mathrm{L})$ denote the minimum flow time achievable subject to this constraint. Since it is always optimal to start a batch whenever at least $\capac$ jobs are pending and the processor is idle, such a schedule has $q$ pending jobs and an idle processor at time $(k+1+\delta)s$. At this point, there are two asymptotically relevant choices:
\begin{enumerate}
    \item Remain idle until time $(k+2)s$, when the first burst of $\capac$ jobs arrives; process all $M$ bursts immediately as they arrive; and then process the $q$ pending jobs after the final burst. This incurs total flow time
    \[
        M\cdot (Cs+qs)+o(M).
    \]
    \item Immediately process the $q$ pending jobs. This shifts the processing of every subsequent burst by $\delta s$, causing an additional delay of $\delta s$ for each of the $\capac$ jobs in each burst. This incurs total flow time
    \[
        M\cdot (Cs+C\delta s)+o(M).
    \]
\end{enumerate}
Any intermediate choice gives a convex combination of these two asymptotic penalties and is therefore no better. Hence
\begin{equation}\label{eq:adv-lb-late-asymptotic}
    \lim_{M\to\infty}\frac{\opt(\inst_{\mathrm{E}}\mid \mathrm{L})}{M}
    =
    Cs+\min\{qs,C\delta s\}.
\end{equation}

Next suppose that in instance $\inst_{\mathrm{L}}$ the scheduler starts a batch before time $\delta s$, containing only the first $q$ arrivals. Starting this batch at time $0$ is the best possible such early decision, so it suffices to lower bound the cost of schedules subject to this stronger assumption. Let $\opt(\inst_{\mathrm{L}}\mid \mathrm{E})$ denote the minimum flow time achievable subject to starting the first batch at time $0$.

Again, the scheduler should start a batch whenever at least $\capac$ jobs are pending and the processor is idle. Such a schedule has $\capac-q$ pending jobs and an idle processor at time $(k+1)s$. At this point, there are two asymptotically relevant choices:
\begin{enumerate}
    \item Remain idle until time $(k+1+\delta)s$, when the first burst of $\capac$ jobs arrives; process all $M$ bursts immediately as they arrive; and then process the $\capac-q$ pending jobs after the final burst. This incurs total flow time
    \[
        M\cdot (Cs+(C-q)s)+o(M).
    \]
    \item Immediately process the $\capac-q$ pending jobs. This shifts the processing of every subsequent burst by $(1-\delta)s$, causing an additional delay of $(1-\delta)s$ for each of the $\capac$ jobs in each burst. This incurs total flow time
    \[
        M\cdot (Cs+C(1-\delta)s)+o(M).
    \]
\end{enumerate}
Thus
\begin{equation}\label{eq:adv-lb-early-asymptotic}
    \lim_{M\to\infty}\frac{\opt(\inst_{\mathrm{L}}\mid \mathrm{E})}{M}
    =
    Cs+\min\{(C-q)s,C(1-\delta)s\}.
\end{equation}

We are now ready to derive the lower bounds. The adversary declares that the true instance is either $\inst_{\mathrm{E}}$ or $\inst_{\mathrm{L}}$. Since the two instances are indistinguishable until time $\delta s$, the scheduler must decide whether to start a batch before time $\delta s$ without knowing which continuation will occur.

Define
\[
    a:=\min\left\{\frac{q}{C},\delta\right\},
    \qquad
    b:=\min\left\{1-\frac{q}{C},1-\delta\right\}.
\]
By \eqref{eq:adv-lb-optimal-asymptotic} and \eqref{eq:adv-lb-late-asymptotic}, if the scheduler waits until time $\delta s$ on instance $\inst_{\mathrm{E}}$, then its asymptotic competitive ratio is at least
\[
    1+a.
\]
Similarly, by \eqref{eq:adv-lb-optimal-asymptotic} and \eqref{eq:adv-lb-early-asymptotic}, if the scheduler starts a batch before time $\delta s$ on instance $\inst_{\mathrm{L}}$, then its asymptotic competitive ratio is at least
\[
    1+b.
\]

For a deterministic scheduler, the adversary observes which of these two actions the scheduler takes before time $\delta s$. If the scheduler waits, the adversary chooses $\inst_{\mathrm{E}}$; if the scheduler starts early, the adversary chooses $\inst_{\mathrm{L}}$. Therefore the asymptotic competitive ratio is at least
\[
    1+\min\{a,b\}
    =
    1+\min\left\{
        \frac{q}{C},
        \delta,
        1-\frac{q}{C},
        1-\delta
    \right\}.
\]
Now choose
\[
    q=\lfloor C/2\rfloor
    \qquad\text{and}\qquad
    \delta=\frac{q}{C}.
\]
Then
\[
    a=\frac{q}{C}
    =
    \frac{\lfloor C/2\rfloor}{C},
    \qquad
    b=1-\frac{q}{C}
    =
    \frac{\lceil C/2\rceil}{C},
\]
and hence
\[
    \min\{a,b\}
    =
    \frac{\lfloor C/2\rfloor}{C}.
\]
Thus the deterministic lower bound tends to
\[
    1+\frac{\lfloor C/2\rfloor}{C}
\]
as $M\to\infty$. Choosing $M$ sufficiently large gives the claimed
\[
    1+\frac{\lfloor C/2\rfloor}{C}-\epsilon
\]
lower bound.

Now consider a randomized scheduler. Let $p$ denote the probability that the scheduler starts a batch before time $\delta s$; this probability is well-defined because the two instances are indistinguishable up to time $\delta s$. If the adversary chooses $\inst_{\mathrm{E}}$, then with probability $1-p$ the scheduler makes the wrong initial decision, and the asymptotic competitive ratio is at least
\[
    1+(1-p)a.
\]
If the adversary chooses $\inst_{\mathrm{L}}$, then with probability $p$ the scheduler makes the wrong initial decision, and the asymptotic competitive ratio is at least
\[
    1+p b.
\]
Therefore the adversary can force asymptotic competitive ratio at least
\[
    1+\max\{(1-p)a,p b\}.
\]
For any $p\in[0,1]$,
\[
    \max\{(1-p)a,p b\}
    \ge
    \frac{ab}{a+b},
\]
with equality when $p=a/(a+b)$. Using again
\[
    q=\lfloor C/2\rfloor
    \qquad\text{and}\qquad
    \delta=\frac{q}{C},
\]
we have
\[
    a=\frac{\lfloor C/2\rfloor}{C},
    \qquad
    b=\frac{\lceil C/2\rceil}{C},
\]
and therefore
\[
    \frac{ab}{a+b}
    =
    \frac{
        \frac{\lfloor C/2\rfloor}{C}
        \cdot
        \frac{\lceil C/2\rceil}{C}
    }{
        \frac{\lfloor C/2\rfloor+\lceil C/2\rceil}{C}
    }
    =
    \frac{\lfloor C/2\rfloor\lceil C/2\rceil}{C^2}.
\]
Thus the randomized lower bound tends to
\[
    1+\frac{\lfloor C/2\rfloor\lceil C/2\rceil}{C^2}
\]
as $M\to\infty$. Choosing $M$ sufficiently large gives the claimed
\[
    1+\frac{\lfloor C/2\rfloor\lceil C/2\rceil}{C^2}-\epsilon
\]
lower bound.

\end{proof}

\section{The Stochastic Batching Model (Omitted Proofs from~\Cref{sec:poisson})}
\label{app:poisson}

In this appendix, we formally prove our various statements for the stochastic batching model, as presented in~\Cref{sec:poisson}. First, in~\Cref{app:poisson-offline} we establish the existence of the offline-optimal flow-time asymptotic rate under Poisson arrivals and establish the finiteness of that rate if and only if $\lambda s < \capac$ (\Cref{thm:poisson-stability}); we also provide a brief discussion of the unstable regime $\lambda s \ge \capac$. Then, in~\Cref{app:poisson-online}, we formally establish a series of reductions demonstrating that the optimal rate across all online schedules in the stable regime $\lambda s < \capac$ is finite, and achieved by one of $\capac$ threshold-based policies (\Cref{thm:online-thresholds}). While doing so, we also discuss a cycle renewal decomposition for these algorithms, and how one can optimize over them both theoretically and empirically. We then continue by establishing the strict separation between online and offline schedules for the Poisson arrival setting in~\Cref{app:online_vs_offline}, proving~\Cref{thm:online-gap}. Next we take a detour and provide an analysis for the greedy online scheduler $\sched_1^\on$ in~\Cref{app:poisson_greedy}, establishing the key ingredients for the proof of~\Cref{lemma:bad_event_bound}. Finally, we conclude this chapter with the discussion of (perfect) lookahead algorithms, analyzing $\algname$ and proving~\Cref{thm:poisson-upper} in~\Cref{app:poisson-lookahead-upper} and presenting the matching (up to constants) lower bound of~\Cref{thm:poisson-lb} in~\Cref{app:poisson-lookahead-lower}

\subsection{The Offline Optimal Schedule for the Poisson Setting (\Cref{thm:poisson-stability})}
\label{app:poisson-offline}
Recall that $\sched^*(n)$ denotes the (random) optimal total flow time on the first $n$ (random) arrivals of a rate-$\lambda$ Poisson process. We first prove that the offline-optimal rate, defined as the ordinary limit $\rate(\sched^*):=\lim_{n\to \infty}(\expect{\sched^*(n)}/n)$, always exists, possibly equal to $+\infty$ (\Cref{app:subsec_limit_existance}). We then characterize when this limit is finite and identify $\lambda s < \capac$ as the sufficient and necessary condition for this to happen (\Cref{app:subsec_limit_finitie}).

\subsubsection{Existence of the Offline Rate}\label{app:subsec_limit_existance}

The existence of the offline rate follows from super-additivity of the flow-time objective and Fekete's lemma. In particular, fix any integers $n,m\ge 1$, and let
$$
\inst=\{r_1,\ldots,r_{n+m}\}
$$
be a realization of $\inst(n+m)$ with jobs labeled by arrival order $r_j\le r_{j+1}$. Split this instance into
$$
\inst_1:=\{r_1,\ldots,r_n\},
\qquad
\inst_2':=\{r_{n+1},\ldots,r_{n+m}\}.
$$
By super-additivity of the optimal flow-time objective (\Cref{lemma:superadditive}), we have that

\begin{equation}\label{eq:limit_exists}
    \opt(\inst)
    \ge
    \opt(\inst_1)+\opt(\inst_2').
\end{equation}

Furthermore, the optimal flow time is invariant under shifting all release times by the same amount. Therefore, if
$$
\inst_2:=\{r_{n+1}-r_n,\ldots,r_{n+m}-r_n\},
$$
then $\opt(\inst_2')=\opt(\inst_2)$. By the strong Markov property of the Poisson process at time $r_n$, the shifted instance $\inst_2$ is independent of $\inst_1$ and has the same distribution as $\inst(m)$. Furthermore, $\inst_1$ has the same distribution as $\inst(n)$ for the same reason. Therefore, taking expectations in~\eqref{eq:limit_exists} gives
$$
\expect{\sched^*(n+m)}
\ge
\expect{\sched^*(n)}
+
\expect{\sched^*(m)}.
$$
Thus the sequence $a_n:=\expect{\sched^*(n)}$ is super-additive. Fekete's lemma implies
$$
\rate(\sched^*):=\lim_{n\to\infty}\frac{a_n}{n}
=
\sup_{n\ge 1}\frac{a_n}{n},
$$
where the common value may be $+\infty$. Therefore, $\rate(\sched^*)$ is always a well-defined limit.

\subsubsection{Upper Bounds on the Offline Optimum}

We next upper bound the growth of the offline optimum by comparing it to a simple feasible block scheduler.

\begin{lemma}\label{claim:block-scheduler-upper-bound}
For the offline optimum, the following upper bounds hold:
$$
\expect{\sched^*(n)}
\le
\begin{cases}
O_{\lambda,s,\capac}(n),
& \text{if } \lambda s<\capac, \\
O_{\lambda,s,\capac}(n^{3/2}),
& \text{if } \lambda s=\capac, \\
\displaystyle
\frac{\lambda s-\capac}{2\lambda\capac}n^2+o_{\lambda,s,\capac}(n^2),
& \text{if } \lambda s>\capac.
\end{cases}
$$
\end{lemma}

\begin{proof}
Consider the following feasible block scheduler. Partition the jobs into consecutive blocks of size $\capac$: the first block contains jobs $1,\ldots,\capac$, the second contains jobs $\capac+1,\ldots,2\capac$, and so on. The scheduler starts each full block as soon as all jobs in that block have arrived and the processor is idle. If fewer than $\capac$ jobs remain at the end, it processes them as a final partial block. Since $\sched^*$ is offline-optimal, the expected flow time of this block scheduler upper bounds $\expect{\sched^*(n)}$.

For each full block $k$, let
$$
R_k:=r_{k\capac}
$$
be the arrival time of the last job in that block, and set $R_0:=0$. Define
$$
X_k:=R_k-R_{k-1}.
$$
The random variables $(X_k)_{k\ge 1}$ are i.i.d. Erlang random variables with parameters $(\capac,\lambda)$, so
$$
\expect{X_k}=\frac{\capac}{\lambda}.
$$
Let $S_k$ be the time at which the $k$-th full block starts processing, and define
$$
W_k:=S_k-R_k
$$
to be the waiting time of the $k$-th full block after it becomes available. By construction, the start times satisfy
$$
S_{k+1}=\max\{R_{k+1},S_k+s\}.
$$
Subtracting $R_{k+1}$ gives the Lindley \cite{L52} recursion
$$
W_{k+1}
=
\big(W_k+s-X_{k+1}\big)^+.
$$
Thus $(W_k)_{k\ge 1}$ is a reflected random walk with increments
$$
Y_k:=s-X_k.
$$
These increments are i.i.d., have finite exponential moments, and satisfy
$$
\expect{Y_k}
=
s-\frac{\capac}{\lambda}
=
\frac{\lambda s-\capac}{\lambda}.
$$

We use standard bounds for this reflected random walk. If $\lambda s<\capac$, then the increments have negative drift, and
$$
\sup_{k\ge 1}\expect{W_k}<\infty.
$$
If $\lambda s=\capac$, then the increments have mean zero and finite variance, so standard maximal inequalities for mean-zero random walks imply
$$
\expect{W_k}=O_{\lambda,s,\capac}(\sqrt{k}).
$$
Finally, if $\lambda s>\capac$, then the increments have positive drift
$$
\delta:=\frac{\lambda s-\capac}{\lambda}>0,
$$
and the reflected random walk grows linearly:
$$
\expect{W_k}
=
\delta k+o_{\lambda,s,\capac}(k).
$$

It remains to translate these waiting-time bounds into flow-time bounds. For a full block $k$, once all jobs in the block have arrived, each of the $\capac$ jobs waits an additional $W_k$ time units before processing begins, and then spends $s$ time units in service. The expected within-block waiting before the final job of the block arrives depends only on $\lambda$ and $\capac$. Hence the expected contribution of block $k$ to total flow time is at most
$$
\capac\expect{W_k}+O_{\lambda,s,\capac}(1).
$$
Let
$$
M:=\left\lfloor\frac{n}{\capac}\right\rfloor
$$
be the number of full blocks. The final partial block contains fewer than $\capac$ jobs; its contribution is bounded by the same estimate for one additional block and is lower-order in all three regimes. Therefore,
$$
\expect{\sched^*(n)}
\le
\capac\sum_{k=1}^{M}\expect{W_k}
+
O_{\lambda,s,\capac}(n).
$$

\begin{enumerate}
    \item If $\lambda s<\capac$, then $\sup_k\expect{W_k}<\infty$, and so
    $$
    \expect{\sched^*(n)}
    =
    O_{\lambda,s,\capac}(n).
    $$
    
    \item If $\lambda s=\capac$, then $\expect{W_k}=O_{\lambda,s,\capac}(\sqrt{k})$, and hence
    $$
    \expect{\sched^*(n)}
    \le
    \capac\sum_{k=1}^{M}O_{\lambda,s,\capac}(\sqrt{k})
    +
    O_{\lambda,s,\capac}(n)
    =
    O_{\lambda,s,\capac}(n^{3/2}).
    $$

    \item If $\lambda s>\capac$, since
    $$
    \expect{W_k}=\delta k+o_{\lambda,s,\capac}(k),
    $$
    we have
    $$
    \sum_{k=1}^{M}\expect{W_k}
    =
    \frac{\delta M^2}{2}+o_{\lambda,s,\capac}(M^2).
    $$
    Using $M=n/\capac+O(1)$ gives
    \begin{align*}
    \expect{\sched^*(n)}
    &\le
    \capac\left(
    \frac{\delta M^2}{2}
    +
    o_{\lambda,s,\capac}(M^2)
    \right)
    +
    O_{\lambda,s,\capac}(n) \\
    &=
    \frac{\delta}{2\capac}n^2
    +
    o_{\lambda,s,\capac}(n^2) 
    =
    \frac{\lambda s-\capac}{2\lambda\capac}n^2
    +
    o_{\lambda,s,\capac}(n^2).
\end{align*}
\end{enumerate}

This completes the proof.
\end{proof}

\subsubsection{Lower Bounds on the Offline Optimum}

We now prove matching lower bounds. The stable-regime lower bound is immediate from the fact that every job spends at least $s$ time units in the system. The critical and overloaded lower bounds use a standard workload argument.

\begin{lemma}\label{claim:offline-growth-lower-bounds}
For the offline optimum, the following lower bounds hold:
$$
\expect{\sched^*(n)}
\ge
\begin{cases}
sn,
& \text{if } \lambda s<\capac, \\
\Omega_{\lambda,s,\capac}(n^{3/2}),
& \text{if } \lambda s=\capac,\\ 
\displaystyle
\frac{\lambda s-\capac}{2\lambda\capac}n^2-o_{\lambda,s,\capac}(n^2),
& \text{if } \lambda s>\capac.
\end{cases}
$$
\end{lemma}

\begin{proof}
The stable-regime lower bound is immediate: by~\Cref{claim:opt-lb}, we have that $\sched^*(n)\ge sn$ with probability $1$. We now prove the critical and overloaded lower bounds. Let $N(t)$ denote the number of arrivals of the rate-$\lambda$ Poisson process by time $t$. We first record the deterministic workload lower bound used in both arguments. For the finite instance $\inst(n)$, the number of released jobs by time $t$ is $\min\{N(t),n\}$. Any feasible scheduler can complete at most $\capac$ jobs every $s$ units of time, and hence at most $(\capac/s)t$ jobs by time $t$. Since total flow time is the area under the number of released but unfinished jobs (\Cref{claim:flow-def}), with probability $1$, we have that 

\begin{equation}\label{eq:rate_lb}
\sched^*(n)
\ge
\int_0^\infty
\bigg(
\min\{N(t),n\}
-
\frac{\capac}{s}t
\bigg)^+
dt.
\end{equation}

We first consider the overloaded regime $\lambda s>\capac$. Dividing \eqref{eq:rate_lb} by $n^2$ and performing the change in variables $t=nx$ gives
$$
\frac{\sched^*(n)}{n^2}
\ge
\int_0^\infty
\bigg(
\min\bigg\{\frac{N(nx)}{n},1\bigg\}
-
\frac{\capac}{s}x
\bigg)^+
dx.
$$
Since the integrand is nonnegative, taking expectations and applying Tonelli's theorem gives
$$
\frac{\expect{\sched^*(n)}}{n^2}
\ge
\int_0^\infty
\expect{\bigg(
\min\bigg\{\frac{N(nx)}{n},1\bigg\}
-
\frac{\capac}{s}x
\bigg)^+}
dx.
$$
By Fatou's lemma,
$$
\liminf_{n\to\infty}
\frac{\expect{\sched^*(n)}}{n^2}
\ge
\int_0^\infty
\liminf_{n\to\infty}
\expect{
\bigg(
\min\bigg\{\frac{N(nx)}{n},1\bigg\}
-
\frac{\capac}{s}x
\bigg)^+
}
dx.
$$
For every fixed $x\ge 0$, the strong law for the Poisson process gives
$$
\frac{N(nx)}{n}\to \lambda x
\qquad\text{almost surely}.
$$
Furthermore, the integrand is upper bounded by $1$, so bounded convergence gives
$$
\lim_{n\to\infty}
\expect{
\bigg(
\min\bigg\{\frac{N(nx)}{n},1\bigg\}
-
\frac{\capac}{s}x
\bigg)^+
}
=
\bigg(
\min\{\lambda x,1\}
-
\frac{\capac}{s}x
\bigg)^+.
$$
Therefore,
$$
\liminf_{n\to\infty}
\frac{\expect{\sched^*(n)}}{n^2}
\ge
\int_0^\infty
\bigg(
\min\{\lambda x,1\}
-
\frac{\capac}{s}x
\bigg)^+
dx.
$$
When $\lambda s>\capac$, the integrand is positive exactly on $x\in(0,s/\capac)$, and a direct calculation yields
$$
\int_0^\infty
\bigg(
\min\{\lambda x,1\}
-
\frac{\capac}{s}x
\bigg)^+
dx
=
\frac{\lambda s-\capac}{2\lambda\capac}.
$$
and therefore
\[
\liminf_{n\to\infty}
\frac{\expect{\sched^*(n)}}{n^2} \ge \frac{\lambda s-\capac}{2\lambda\capac},\]
proving the desired lower bound for $\lambda s > \capac$.

It remains to prove the lower bound in the critical regime. Assume $\lambda s=\capac$. Then
$$
N(t)\sim \mathrm{Poisson}(\lambda t)
=
\mathrm{Poisson}\left(\frac{\capac}{s}t\right)
=
\mathrm{Poisson}\left(\mu\right).
$$
for $\mu:= (\capac t)/s$. For every $t\ge 0$, we have
\begin{equation}\label{eq:rlb1}
\bigg(
\min\{N(t),n\}
-
\frac{\capac}{s}t
\bigg)^+
\ge
(N(t)-\mu)^+\mathbf{1}(N(t)\le n) \
\ge
(N(t)-\mu)^+
-
N(t)\mathbf{1}(N(t)>n).
\end{equation}
If $\mu\le n/2$, then the Chernoff bound for Poisson random variables gives
$$
\prob{N(t)>n}
\le
\prob{N(t)>2\mu}
\le
\exp(-\mu/3).
$$
Therefore, by Cauchy--Schwarz,
\begin{equation}\label{eq:rlb2}
\expect{N(t)\mathbf{1}(N(t)>n)}
\le
\sqrt{\expect{N(t)^2}\prob{N(t)>n}}
\le
(\mu+1)\exp(-\mu/6),
\end{equation}
where we used $\expect{N(t)^2}=\mu^2+\mu$.

We next lower bound $\expect{(N(t)-\mu)^+}$. Let $k:=\lfloor \mu\rfloor$. Using the Poisson identity
$$
\expect{N(t)\mathbf{1}(N(t)\ge r})
=
\mu\prob{N(t)\ge r-1},
$$
with $r=k+1$, we obtain
\begin{align*}
\expect{(N(t)-\mu)^+}
&=
\sum_{\ell=k+1}^{\infty}(\ell-\mu)\prob{N(t)=\ell} \\
&=
\expect{N(t)\mathbf{1}(N(t)\ge k+1)}
-
\mu\prob{N(t)\ge k+1} \\
&=
\mu\prob{N(t)=k}.
\end{align*}
For $\mu\ge 1$, Stirling's bound therefore implies
\begin{equation}\label{eq:rlb3}
\expect{(N(t)-\mu)^+} = \mu\prob{N(t)=k}
\ge
\frac{1}{e^2}\sqrt{\mu}.
\end{equation}

Therefore, for every $t$ satisfying
$$
\frac{s}{\capac}\le t\le \frac{ns}{2\capac},
$$
or equivalently $1\le \mu\le n/2$, by~\eqref{eq:rlb1},~\eqref{eq:rlb2} and~\eqref{eq:rlb3} we have
$$
\expect{
\bigg(
\min\{N(t),n\}
-
\frac{\capac}{s}t
\bigg)^+
}
\ge
\frac{1}{e^2}\sqrt{\frac{\capac t}{s}}
-
\left(\frac{\capac t}{s}+1\right)
\exp\left(-\frac{\capac t}{6s}\right).
$$
Taking expectations in \eqref{eq:rate_lb} and applying Tonelli's theorem gives
$$
\expect{\sched^*(n)}
\ge
\int_{s/\capac}^{ns/(2\capac)}
\left[
\frac{1}{e^2}\sqrt{\frac{\capac t}{s}}
-
\left(\frac{\capac t}{s}+1\right)
\exp\left(-\frac{\capac t}{6s}\right)
\right]dt.
$$
The first term integrates to $\Omega_{\lambda,s,\capac}(n^{3/2})$, while the second term has integral bounded by a constant depending only on $s$ and $\capac$. Hence, whenever $\lambda s = \capac$ we have
$$
\expect{\sched^*(n)}
=
\Omega_{\lambda,s,\capac}(n^{3/2}).
$$
This completes the proof of all three lower bounds.
\end{proof}

\subsubsection{Putting Everything Together}\label{app:subsec_limit_finitie}

Finally, we combine the preceding estimates. From the analysis in~\Cref{app:subsec_limit_existance}, the offline rate exists as an ordinary limit in $[0,+\infty]$. In the stable regime $\lambda s<\capac$, \Cref{claim:block-scheduler-upper-bound} gives $\expect{\sched^*(n)}=O_{\lambda,s,\capac}(n)$, while \Cref{claim:offline-growth-lower-bounds} gives $\expect{\sched^*(n)}\ge sn$. Therefore
$$
s\le \rate(\sched^*)<+\infty.
$$

In the critical regime $\lambda s=\capac$, the upper and lower bounds give
$
\expect{\sched^*(n)}
=
\Theta_{\lambda,s,\capac}(n^{3/2})
$ and hence $\rate(\sched^*)=+\infty$. In the overloaded regime $\lambda s>\capac$, the upper and lower bounds match to leading order:
$$
\expect{\sched^*(n)}
=
\frac{\lambda s-\capac}{2\lambda\capac}n^2
+
o_{\lambda,s,\capac}(n^2),
$$
and therefore $\rate(\sched^*)=+\infty$.

Thus the offline rate always exists and is finite if and only if $\lambda s<\capac$, proving \Cref{thm:poisson-stability}. In the overloaded regime $\lambda s> \capac$, the asymptotics are essentially determined by the deterministic workload imbalance: arrivals inject work faster than full batches can remove it, and the simple block scheduler already matches the offline optimum to leading order. In the critical regime $\lambda s = \capac$, the correct order is $\Theta(n^{3/2})$; optimizing the constant in front of this term is a finer question that we do not pursue here. For the rest of the paper, we therefore focus on the stable regime $\lambda s<\capac$, where the offline optimum has a finite rate and it becomes meaningful to compare online and finite-lookahead schedulers against this benchmark.

\subsection{Online Schedules for the Poisson Setting (\Cref{thm:online-thresholds})}
\label{app:poisson-online}
In this section, we analyze online algorithms for the Poisson batching model in the stable regime $\lambda s<\capac$. Our goal is to prove that, for the purpose of optimizing the asymptotic flow-time rate, it suffices to consider the finite class of threshold policies appearing in~\Cref{thm:online-thresholds}.

We begin by formalizing the class of online schedulers. An online scheduler is a nonanticipatory protocol: at any time $t$, its decisions may depend on the arrival history up to time $t$, its own past decisions, the current set of pending jobs, and its internal randomness, but not on arrivals after time $t$. In particular, whenever the processor is idle and there are pending jobs, the scheduler may either start a batch immediately or wait for additional arrivals before starting the next batch.

The proof proceeds through a sequence of reductions, each simplifying the class of online policies without increasing the optimal achievable rate. First, as in the adversarial setting, we show that it is without loss of generality to restrict attention to work-conserving policies (\Cref{def:work-conserving}). Second, we formulate the rate-minimization problem as an average-cost semi-Markov decision process. This reduces the search over arbitrary history-dependent online policies to stationary timer policies: algorithms specified by deterministic timers $\tau_q\in[0,\infty]$, one for each state $q\in\{1,\ldots,\capac-1\}$, where $q$ is the number of pending jobs when the processor becomes idle. Third, using the regenerative structure of stationary timer policies and a ratio-to-difference transformation for the cycle objective, we show that finite positive timers are unnecessary: each timer may be chosen from $\{0,\infty\}$. Finally, we prove that an optimal binary timer policy can be chosen monotone in the queue length. Hence the set of states in which the scheduler batches immediately is an upper set, and the policy is a threshold policy. Pictorially, the reductions are summarized in~\Cref{fig:online-reduction-chain}.

% Preamble:
% \usepackage{tikz}
% \usepackage{float} % needed only if you use [H]

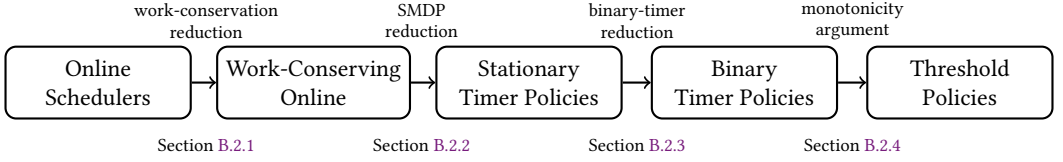
\begin{figure}[H]
\centering
\begin{tikzpicture}[
    x=1cm,
    y=1cm,
    box/.style={
        draw,
        rounded corners,
        line width=0.8pt,
        align=center,
        minimum width=2.45cm,
        minimum height=0.95cm,
        font=\small
    },
    arr/.style={
        ->,
        line width=0.8pt
    },
    toplab/.style={
        font=\scriptsize,
        align=center
    },
    botlab/.style={
        font=\scriptsize,
        align=center
    }
]

% Box centers
\def\xA{1.25}
\def\xB{4.10}
\def\xC{6.95}
\def\xD{9.80}
\def\xE{12.65}
\def\ybox{0}

% Boxes
\node[box] (online) at (\xA,\ybox) {Online\\Schedulers};
\node[box] (wc) at (\xB,\ybox) {Work-Conserving\\Online};
\node[box] (timer) at (\xC,\ybox) {Stationary\\Timer Policies};
\node[box] (binary) at (\xD,\ybox) {Binary\\Timer Policies};
\node[box] (threshold) at (\xE,\ybox) {Threshold\\Policies};

% Arrows
\draw[arr] (online.east) -- (wc.west);
\draw[arr] (wc.east) -- (timer.west);
\draw[arr] (timer.east) -- (binary.west);
\draw[arr] (binary.east) -- (threshold.west);

% Concept labels above arrows
\node[toplab] at (2.68,0.82) {work-conservation\\reduction};
\node[toplab] at (5.53,0.82) {SMDP\\reduction};
\node[toplab] at (8.38,0.82) {binary-timer\\reduction};
\node[toplab] at (11.23,0.82) {monotonicity\\argument};

% Space for section / lemma references below arrows
\node[botlab] at (2.68,-0.82) {\Cref{subsec:red_wc}};
\node[botlab] at (5.53,-0.82) {\Cref{subsec:red_stat}};
\node[botlab] at (8.38,-0.82) {\Cref{subsec:red_binary}};
\node[botlab] at (11.23,-0.82) {\Cref{subsec:red_threshold}};

\end{tikzpicture}
\caption{Reduction from general online schedulers to threshold policies in the stable regime $\lambda s<\capac$.}
\label{fig:online-reduction-chain}
\end{figure}

Along the way, we establish a renewal-reward cycle decomposition for stationary timer policies. This decomposition shows that their ordinary asymptotic rates exist and are finite in the stable regime $\lambda s<\capac$, and expresses each such rate as the ratio of expected flow time to expected number of jobs in a regenerative cycle (\Cref{claim:stationary-timer-cycle-rate}). Finally, in~\Cref{subsec:online_optim}, we explain how to evaluate and optimize over the resulting finite family of threshold policies, both theoretically and empirically.

\subsubsection{From Online to Work-Conserving Online Policies}
\label{subsec:red_wc}

We first show that online schedulers can be assumed to be work-conserving. Recall that from \Cref{def:work-conserving}, work-conservation requires two properties: the scheduler should never leave capacity unused in a batch when additional jobs are already pending, and it should never idle when a full batch is available. Naturally, both requirements are without loss of generality.

\begin{proposition}
\label{lemma:online-work-conserving-wlog}
For every online scheduler $\sched$, there exists a work-conserving online scheduler $\sched'$ such that, for every deterministic arrival sequence $\inst$, we have
$
\sched'(\inst)\le \sched(\inst)
$.

\end{proposition}

\begin{proof}
Fix an online scheduler $\sched$. We modify it in two local ways. First, whenever $\sched$ starts a batch while $q$ jobs are pending, the modified scheduler includes $\min\{q,\capac\}$ pending jobs in that batch. Since the processing time of a batch is always $s$, adding pending jobs to an already-started batch does not delay any completion time and can only complete additional jobs earlier.

Second, whenever the processor is idle and at least $\capac$ jobs are pending, the modified scheduler starts a full batch immediately. This is also harmless: delaying such a full batch only keeps at least $\capac$ available jobs waiting, while starting it immediately uses an idle processor and completes those jobs earlier. The subsequent behavior can be coupled with the original execution by preserving the remaining batch start times and assigning completion slots to the oldest unfinished jobs.

Applying these two modifications throughout the execution produces an online scheduler $\sched'$ that is work-conserving. The construction is nonanticipatory, since it uses only the current pending set and the decision that $\sched$ would make at the current time. On every deterministic arrival sequence, no job completion time is increased, and therefore
$
\sched'(\inst)\le \sched(\inst)
$.
\end{proof}

Taking expectations over the Poisson arrivals and over the internal randomness of $\sched$ (if any) we obtain $\expect{\sched'(n)}\le \expect{\sched(n)}$ for every $n$. Taking limsups, it therefore suffices to consider work-conserving online algorithms from now on when minimizing the rate. We will also use the standard convention that whenever a scheduler starts a batch, it processes the oldest pending jobs. This is without loss of generality: for any fixed finite instance, total flow time depends only on the multiset of completion times, not on the assignment of completion times to job labels. Thus we may relabel completion slots so that jobs complete in arrival order.

\subsubsection{From Work-Conserving Online to Stationary Timer Policies}
\label{subsec:red_stat}

We now reduce from arbitrary work-conserving online schedulers to stationary policies. By work-conservation, the only nontrivial decisions occur when the processor is idle and the number of pending jobs satisfies
$$
q\in\{1,\ldots,\capac-1\}.
$$
If $q=0$, the scheduler has no available job and must wait for the next arrival. If $q\ge \capac$, a full batch is available and the scheduler must start a batch immediately. Thus, whenever the processor is idle with $q\in\{1,\ldots,\capac-1\}$ pending jobs, the scheduler's next decision can be represented by a (potentially random) timer
$$
\tau\in[0,\infty].
$$
The interpretation is that the scheduler waits up to time $\tau$ for another arrival. If an arrival occurs before the timer expires, the timer is canceled and the state becomes $q+1$. If no arrival occurs before time $\tau$, the scheduler starts a batch with the $q$ pending jobs. Importantly, this timer representation does not impose stationarity. A general work-conserving online scheduler may choose the timer at each decision epoch as an arbitrary function of the full past history and its internal randomness. The point is only that the current state and chosen timer determine the law of the next transition.

Indeed, let $X\sim\mathrm{Exp}(\lambda)$ be the time to the next arrival. If $X<\tau$, the system waits for time $X$, accumulates $qX$ units of flow time, completes no jobs, and moves to state $q+1$. If $X\ge \tau$, the system waits for time $\tau$, starts a batch, and completes it after an additional time $s$. The $q$ initial jobs contribute $q(\tau+s)$ units of flow time. During the service interval, new jobs may arrive, but they are not processed by the current batch. A job arriving $u$ time units after service begins contributes $s-u$ units of flow before the next decision epoch. Since arrivals occur at rate $\lambda$, the expected total contribution of arrivals during service is
$$
\lambda\int_0^s (s-u)du
=
\frac{\lambda s^2}{2}.
$$
If $A\sim\mathrm{Poisson}(\lambda s)$ is the number of arrivals during service, the next state is $A$. Finally, for states $q\ge \capac$, the action is forced: the scheduler immediately starts a full batch. If $A\sim\mathrm{Poisson}(\lambda s)$ is the number of arrivals during the service time, then the next state is $q-\capac+A$ and the expected flow accumulated during this transition is
$$
qs+\frac{\lambda s^2}{2}.
$$

Therefore, after restricting to work-conserving schedulers, the online scheduling problem is a countable-state semi-Markov decision process. The state is the number of pending jobs at idle decision epochs, the actions are timers in $[0,\infty]$ for states $\{1,\ldots,\capac-1\}$, and the objective is the long-run flow time per job. A deterministic stationary policy for this SMDP is exactly a fixed choice of one timer for each nontrivial state.

\begin{definition}[Stationary Timer Policies]
\label{def:stationary-timer-policy}
A stationary timer policy is specified by a vector
$$
\tau=(\tau_1,\ldots,\tau_{\capac-1})\in[0,\infty]^{\capac-1}.
$$
Whenever the processor is idle with $q\in\{1,\ldots,\capac-1\}$ pending jobs, the policy sets a timer of length $\tau_q$. If a new job arrives before the timer expires, the timer is canceled and the policy updates its decision at the new state. If the timer expires first, the policy starts a batch with the $q$ pending jobs. As usual, if $q=0$ the policy waits for the next arrival, and if $q\ge\capac$ it starts a full batch immediately.
\end{definition}

We now appeal to the standard stationary-policy reduction for average-cost semi-Markov decision processes. The objective of the SMDP is the long-run expected flow time per completed job: each transition accumulates the queue-area cost incurred until the next decision epoch, and the normalization is by the number of jobs completed. Thus the objective coincides with $\rate(\sched)$. For denumerable-state average-cost SMDPs with compact action spaces and unbounded costs, stationary optimality holds under standard recurrence hypotheses; see, for example,~\cite{FedergruenHordijkTijms1979,sennott1989average}. In our setting, these recurrence hypotheses are guaranteed by the stability condition $\lambda s<\capac$. We therefore obtain the following reduction.

\begin{proposition}
\label{lemma:online-smdp-reduction}
Assume $\lambda s<\capac$. Then\footnote{Formally, we only take the infimum over schedulers for which the rate exists; as we show next, this is indeed the case for all stationary timer policies.}
$$
\inf_{\sched:\emph{ work-conserving online}}
\rate(\sched)
=
\inf_{\tau\in[0,\infty]^{\capac-1}}
\rate(\sched_\tau),
$$
where $\sched_\tau$ denotes the stationary timer policy associated with $\tau$.
\end{proposition}

\begin{proof}
The preceding construction identifies every work-conserving online scheduler with a history-dependent, possibly randomized policy for a semi-Markov decision process whose state is the number of pending jobs at idle decision epochs. The Poisson assumption gives the required Markov property: conditional on the current state, future arrivals are independent of the past, and the distribution of the next state, elapsed time, accumulated flow time, and number of completed jobs depends only on the chosen timer.

The one-step expected flow-time costs grow at most linearly in the queue
length. Moreover, uniformly over timer actions, the number of arrivals
during a service interval has a Poisson distribution and therefore has an exponentially decaying tail. Consequently, the transition kernel has uniformly bounded moments of every fixed order, and all one-step expected costs are finite.

Lastly, the average-cost objective of this SMDP is the long-run expected accumulated flow time normalized by the long-run expected number of completed jobs, which is exactly the definition of $\rate(\sched)$. By the stationary-policy optimality theorem for average-cost semi-Markov decision processes~\cite{FedergruenHordijkTijms1979,sennott1989average}, the infimum over history-dependent randomized policies equals the infimum over deterministic stationary policies. In this model, deterministic stationary policies are precisely the stationary timer policies of~\Cref{def:stationary-timer-policy}. This proves the claim.
\end{proof}

We close this subsection by recording the regenerative structure of stationary timer policies. This will be used repeatedly below. The key point is that, unlike a general history-dependent online scheduler, a stationary timer policy has no internal state beyond the current queue length. Therefore, every time the system becomes empty, the future evolution after the next arrival is independent of the past and has the same law as a fresh execution.

\begin{definition}[Cycle of a Stationary Timer Policy]
\label{def:stationary-timer-cycle}
Fix a stationary timer policy $\sched_\tau$. Consider an execution that starts from an empty system, and condition on the first arrival occurring at time $0$. Let
$$
T_\tau
:=
\inf\{t>0:\text{ at time }t\text{ the processor is idle and there are no pending jobs}\}.
$$
The cycle induced by $\sched_\tau$ is the finite arrival sequence consisting of all jobs that arrive in the interval $[0,T_\tau)$. We denote by $\cycjobs(\sched_\tau)$ the number of jobs in this cycle, and by $\cycflow(\sched_\tau)$ the total flow time incurred by $\sched_\tau$ on these jobs.
\end{definition}

As is standard in renewal-reward analysis for regenerative stochastic processes, the empty-system times of a stationary timer policy serve as regeneration points. This allows us to express the long-run flow-time rate as the ratio between the expected flow time accumulated in one cycle and the expected number of jobs in that cycle.

\begin{proposition}
\label{claim:stationary-timer-cycle-rate}
Assume $\lambda s<\capac$. For every stationary timer policy $\sched_\tau$, the ordinary asymptotic rate exists, is finite, and satisfies
$$
\rate(\sched_\tau)
=
\lim_{n\to\infty}
\frac{\expect{\sched_\tau(n)}}{n}
=
\frac{
\expect{\cycflow(\sched_\tau)}
}{
\expect{\cycjobs(\sched_\tau)}
}.
$$
\end{proposition}

\begin{proof}
Let $(F_i,N_i)_{i\ge 1}$ be the sequence of cycle flow times and cycle sizes generated by $\sched_\tau$ on an infinite rate-$\lambda$ Poisson arrival sequence. Since the arrival process is Poisson and the policy is stationary, the strong Markov property at empty times implies that the pairs $(F_i,N_i)$ are independent and identically distributed, with
$$
(F_i,N_i)
\sim
\bigl(\cycflow(\sched_\tau),\cycjobs(\sched_\tau)\bigr).
$$

We first verify that one cycle has finite first moments. Consider the embedded chain at batch-completion times, and let $X_m$ be the number of unfinished jobs immediately after the $m$-th batch completes. If $X_m=x\ge \capac$, work-conservation forces an immediate full batch, so
$$
X_{m+1}=x-\capac+A_m,
\qquad
A_m\sim\mathrm{Poisson}(\lambda s).
$$
Writing $\mu:=\lambda s<\capac$, the chain has negative drift when $x\ge \capac$:
$$
\expect{X_{m+1}-X_m\mid X_m=x}
=
-(\capac-\mu)<0.
$$
To obtain the moment bound needed for renewal reward, we use the quadratic Lyapunov function $W(x)=x^2$. For $x\ge \capac$,
$$
\begin{aligned}
\expect{W(X_{m+1})-W(X_m)\mid X_m=x}
=
\expect{(x-\capac+A_m)^2-x^2} =
-2(\capac-\mu)x+\expect{(A_m-\capac)^2}.
\end{aligned}
$$
On the finite set $x<\capac$, the expected time until the next batch completion is finite for every fixed timer vector $\tau$: if the relevant timer is finite, the batch starts after the timer expires unless an arrival occurs first; if the timer is infinite, the scheduler waits for enough Poisson arrivals to form a full batch, which has finite mean. Therefore, after increasing the constant to handle this finite set, there exist constants $c_1,c_2>0$, depending only on $\lambda,s,\capac$ and $\tau$, such that
$$
\expect{W(X_{m+1})-W(X_m)\mid X_m=x}
\le
c_1-c_2x
\qquad
\text{for all }x\ge 0.
$$
The Foster--Lyapunov criterion then implies that the embedded chain has finite expected return time to the empty state and, moreover,
$$
\expect{\sum_{m<\sigma_0}(1+X_m)}<\infty,
$$
where $\sigma_0$ is the first return time of the embedded chain to state $0$.

It remains to relate this to flow time. Let $R_m$ be the flow time accumulated between the $m$-th and $(m+1)$-st batch completions. If $X_m=x\ge\capac$, the next batch starts immediately. During the service interval, the $x$ initial unfinished jobs contribute $xs$, while arrivals during service contribute expected flow
$$
\lambda\int_0^s(s-u)du
=
\frac{\lambda s^2}{2}.
$$
Thus
$$
\expect{R_m\mid X_m=x}
=
xs+\frac{\lambda s^2}{2}
\le
C_\tau(1+x).
$$
If instead $x<\capac$, then before the next batch starts the queue size is always smaller than $\capac$; otherwise work-conservation would force an immediate batch. As argued above, the expected idle time before the next batch completion is finite for fixed $\tau$, and during this idle period fewer than $\capac$ jobs are pending. The subsequent service interval contributes at most
$$
\capac s+\frac{\lambda s^2}{2}
$$
in expectation. Hence, after increasing $C_\tau$ if necessary,
$$
\expect{R_m\mid X_m=x}
\le
C_\tau(1+x)
\qquad
\text{for all }x\ge 0.
$$
Combining this bound with
$$
\expect{\sum_{m<\sigma_0}(1+X_m)}<\infty
$$
shows that
$$
\expect{\cycflow(\sched_\tau)}<\infty.
$$
Similarly, since the expected duration of the cycle is finite and arrivals have rate $\lambda$, we also have
$$
0<
\expect{\cycjobs(\sched_\tau)}
<\infty.
$$

We may therefore apply the renewal-reward theorem with renewal length $N_i$ and reward $F_i$. Let
$$
K(n):=\max\left\{k:\sum_{i=1}^k N_i\le n\right\}
$$
be the number of complete cycles contained among the first $n$ arrivals. The complete-cycle contribution satisfies
$$
\frac{1}{n}
\expect{\sum_{i=1}^{K(n)}F_i}
\longrightarrow
\frac{\expect{F_1}}{\expect{N_1}}.
$$
The remaining jobs lie in a single incomplete cycle, whose expected contribution is $o(n)$ by the same integrability bounds. Therefore
$$
\lim_{n\to\infty}
\frac{\expect{\sched_\tau(n)}}{n}
=
\frac{\expect{F_1}}{\expect{N_1}}
=
\frac{
\expect{\cycflow(\sched_\tau)}
}{
\expect{\cycjobs(\sched_\tau)}
}.
$$
This proves both existence and finiteness of the stationary timer rate.
\end{proof}

\subsubsection{From Stationary Timer Policies to Binary Timer Policies}
\label{subsec:red_binary}

We next show that finite positive timers are unnecessary. Intuitively, once the system is idle with $q$ pending jobs, the scheduler faces a simple choice: either batch immediately, or wait for at least one more arrival. Any finite deterministic waiting time interpolates between these two options, but because Poisson arrivals are memoryless and flow accumulates linearly over time, this interpolation is never strictly better than one of the endpoints.

\begin{proposition}
\label{claim:online-binary-timers}
Assume $\lambda s<\capac$. Then
$$
\inf_{\tau\in[0,\infty]^{\capac-1}}
\rate(\sched_\tau)
=
\inf_{\tau\in\{0,\infty\}^{\capac-1}}
\rate(\sched_\tau).
$$
\end{proposition}

\begin{proof}
Let
$$
\rho^*
:=
\inf_{\tau\in[0,\infty]^{\capac-1}}
\rate(\sched_\tau)
$$
be the optimal stationary-timer rate. By the regenerative representation of~\Cref{claim:stationary-timer-cycle-rate}, each stationary timer policy has rate
$$
\frac{
\expect{\cycflow(\sched_\tau)}
}{
\expect{\cycjobs(\sched_\tau)}
}.
$$
We use the standard ratio-to-difference transformation: at the optimal rate $\rho^*$, minimizing the rate is equivalent to minimizing the adjusted one-cycle objective
$$
\expect{\cycflow(\sched_\tau)-\rho^*\cycjobs(\sched_\tau)}.
$$

Let $V(q)$ be the optimal expected adjusted cost until the system next becomes empty, starting from an idle decision epoch with $q$ pending jobs, and normalize $V(0)=0$. Fix a state $q\in\{1,\ldots,\capac-1\}$. If the policy starts a batch immediately, then the $q$ pending jobs complete after time $s$, arrivals during service contribute expected flow $\lambda s^2/2$, and the next state is $A\sim\mathrm{Poisson}(\lambda s)$. Thus the adjusted cost of immediate batching is
$$
A(q)
:=
qs+\frac{\lambda s^2}{2}
-
\rho^* q
+
\expect{V(A)}.
$$

Now consider setting a timer $\tau\in[0,\infty]$. Let $X\sim\mathrm{Exp}(\lambda)$ be the time to the next arrival. For finite $\tau$, the expected adjusted cost of this action is
$$
H(q,\tau)
=
\int_0^\tau
\lambda e^{-\lambda x}
\bigl(qx+V(q+1)\bigr)dx
+
e^{-\lambda\tau}
\bigl(q\tau+A(q)\bigr).
$$
The first term corresponds to the event $X<\tau$: an arrival occurs first, no job is completed, and the state becomes $q+1$. The second term corresponds to the event $X\ge\tau$: the timer expires, the system waits for time $\tau$, and then batches immediately. Differentiating gives
$$
\frac{\partial}{\partial \tau}H(q,\tau)
=
e^{-\lambda\tau}
\left(
q+\lambda\bigl(V(q+1)-A(q)\bigr)
\right).
$$
Crucially, the sign of this derivative is independent of $\tau$. Hence $H(q,\tau)$ is monotone on $[0,\infty)$, and its minimum over $\tau\in[0,\infty]$ is attained at one of the two endpoints: $\tau=0$ or $\tau=\infty$. If the derivative is identically zero, then every timer is optimal, including both endpoints. Applying this argument at every state $q\in\{1,\ldots,\capac-1\}$ shows that there is an optimal stationary timer policy whose timers all lie in $\{0,\infty\}$. Therefore restricting to binary timers does not increase the optimal rate.
\end{proof}

\subsubsection{From Binary Timer Policies to Threshold Policies}
\label{subsec:red_threshold}

We finally show that an optimal binary timer policy can be chosen to be monotone. Thus, once it is optimal to batch immediately with $q$ pending jobs, it is also optimal to batch immediately with any larger number of pending jobs. This recovers the finite class of threshold policies.

\begin{proposition}
\label{claim:online-monotone-binary}
Assume $\lambda s<\capac$. Among optimal binary timer policies, there exists one whose batching set is monotone. Equivalently, there exists an optimal binary timer policy satisfying
$$
\tau_q\ge \tau_{q'}
\qquad
\text{for all }1\le q\le q'\le \capac-1.
$$
\end{proposition}

\begin{proof}
We optimize over binary timer policies. Let $\rho^*$ be the minimum rate among them. Like in the proof of~\Cref{claim:online-binary-timers}, we again use the adjusted objective
$$
\expect{\cycflow(\sched_\tau)-\rho^*\cycjobs(\sched_\tau)}.
$$
Let $V(q)$ be the optimal expected adjusted cost until the system next becomes empty, starting from an idle decision epoch with $q$ pending jobs. We normalize $V(0)=0$.

Fix $q\in\{1,\ldots,\capac-1\}$. Under a binary timer policy, there are only two possible actions: batch immediately or wait until the next arrival. If the scheduler batches immediately, the $q$ pending jobs complete after time $s$, arrivals during service contribute expected flow $\lambda s^2/2$, and the next state is $A\sim\mathrm{Poisson}(\lambda s)$. Thus the adjusted cost of batching is
$$
G_{\batch}(q)
:=
qs+\frac{\lambda s^2}{2}
-
\rho^* q
+
\expect{V(A)}.
$$
If the scheduler waits until the next arrival, then for $X\sim\mathrm{Exp}(\lambda)$ the $q$ pending jobs accumulate expected flow $q\expect{X}=q/\lambda$, no job is completed, and the next state is $q+1$. Thus the adjusted cost of waiting is
$$
G_{\wait}(q)
:=
\frac{q}{\lambda}+V(q+1).
$$

Define the advantage of batching over waiting by
$$
\Delta(q):=G_{\wait}(q)-G_{\batch}(q).
$$
Batching is optimal at state $q$ whenever $\Delta(q)\ge 0$, while waiting is optimal whenever $\Delta(q)\le 0$. We show that $\Delta(q)$ is strictly increasing in $q$. Indeed, for every $q\in\{1,\ldots,\capac-2\}$,
$$
\begin{aligned}
\Delta(q+1)-\Delta(q)
&=
\frac{1}{\lambda}
+
\bigl(V(q+2)-V(q+1)\bigr)
-
s+\rho^*.
\end{aligned}
$$
It remains to lower bound the marginal difference of the value function. We claim that for every $r\ge 0$,
$$
V(r+1)-V(r)\ge s-\rho^*.
$$
To see this, fix any feasible continuation policy starting from an idle state with $r+1$ pending jobs. Remove one of the initially pending jobs, keep the same future arrivals, and preserve the same batch start times and completion slots for all remaining jobs. This gives a feasible continuation from the state with $r$ pending jobs. The removed job would have completed no earlier than time $s$, since the processor is idle at the initial decision epoch and every batch takes time $s$. Hence its adjusted contribution is at least $s-\rho^*$. Therefore the adjusted cost from $r+1$ initial jobs is at least the adjusted cost of the induced continuation from $r$ initial jobs plus $s-\rho^*$. Taking the infimum over continuations gives $V(r+1)\ge V(r)+s-\rho^*$, as desired. Applying this with $r=q+1$ yields $V(q+2)-V(q+1)\ge s-\rho^*$. Substituting into the expression for $\Delta(q+1)-\Delta(q)$ gives
$$
\Delta(q+1)-\Delta(q)
\ge
\frac{1}{\lambda}>0.
$$

Thus, $\Delta(q)$ is strictly increasing on $\{1,\ldots,\capac-1\}$. It follows that the set of states in which batching is optimal is an upper set. Indeed, if batching is optimal at some state $q$, then $\Delta(q)\ge 0$, and therefore $\Delta(q')>0$ for every $q'>q$. Hence batching is also optimal at every larger state. Therefore, there exists a threshold $k\in\{1,\ldots,\capac\}$ such that an optimal policy waits for $q<k$ and batches immediately for $q\ge k$. Equivalently, $\tau_q=\infty$ for $q<k$ and $\tau_q=0$ for $q\ge k$. This proves the claim.
\end{proof}

We have therefore recovered the finite family of threshold schedulers. For each $k\in\{1,\ldots,\capac\}$, let $\sched_k^\on$ denote the online scheduler that waits whenever fewer than $k$ jobs are pending and starts a batch immediately once at least $k$ jobs are pending. Combining our chain of reductions (\Cref{lemma:online-work-conserving-wlog,lemma:online-smdp-reduction,claim:online-binary-timers,claim:online-monotone-binary}) with the cycle decomposition for stationary timer policies (\Cref{claim:stationary-timer-cycle-rate}), we obtain the following:

\begin{corollary}
\label{cor:online-threshold-optimality}
Assume $\lambda s<\capac$. Then
$$
\inf_{\sched:\emph{ online}}\rate(\sched)
=
\min_{k\in\{1,\ldots,\capac\}}
\rate(\sched_k^\on) = \min_{k\in\{1,\ldots,\capac\}}\frac{
\expect{\cycflow(\sched_k^\on)}
}{
\expect{\cycjobs(\sched_k^\on)}
} < \infty.
$$
\end{corollary}
We have therefore proven~\Cref{thm:online-thresholds}.

\subsubsection{Optimizing over Online Threshold Policies}
\label{subsec:online_optim}

By~\Cref{cor:online-threshold-optimality}, optimizing over all online schedulers in the stable regime $\lambda s<\capac$ reduces to comparing the $\capac$ threshold policies $\sched_1^\on,\sched_2^\on,\ldots,\sched_\capac^\on$, for which, the renewal representation gives
$$
\rate(\sched_k^\on)
=
\frac{
\expect{\cycflow(\sched_k^\on)}
}{
\expect{\cycjobs(\sched_k^\on)}
}.
$$

These ratios can be evaluated in two standard ways. The simplest is regenerative simulation: for each $k$, repeatedly simulate independent cycles of $\sched_k^\on$, starting with one arrival into an empty system and stopping when the system first becomes empty again. Averaging the resulting cycle flow times and cycle sizes gives an estimate of $\rate(\sched_k^\on)$.

Alternatively, the same quantities can be computed from the first-step equations of the embedded queue-length chain. For fixed $k$, let $N_k(q)$ and $F_k(q)$ denote respectively the expected number of jobs completed and expected flow time accumulated before the system next becomes empty, starting from an idle decision epoch with $q$ pending jobs. Then
$$
\rate(\sched_k^\on)=\frac{F_k(1)}{N_k(1)}.
$$
The equations are obtained by conditioning on the first transition. If $q<k$, the policy waits for the next arrival, moves to state $q+1$, completes no jobs, and accumulates expected flow $q/\lambda$. If $q\ge k$, the policy immediately processes $\min\{q,\capac\}$ jobs; after the service time $s$, the next state is $A+(q-\capac)^+$ for $A\sim\mathrm{Poisson}(\lambda s)$ and the flow accumulated during this transition is
$qs+\lambda s^2/2$. Together with the boundary condition $N_k(0)=F_k(0)=0$, these transition identities define an infinite linear system for $\{N_k(q),F_k(q)\}_{q\ge 0}$. In numerical computations, one can truncate this system at a large state cutoff and control the error using the negative drift in the stable regime together with standard Poisson tail bounds.

In our experiments, we also compare against a stronger pseudo-online benchmark. This benchmark runs all threshold policies
$$
\sched_1^\on,\ldots,\sched_\capac^\on
$$
on the realized instance and reports the minimum flow time achieved among them. It is not itself an online algorithm, since it chooses the best threshold in hindsight after seeing the instance. Nevertheless, it is a useful optimistic baseline for understanding the performance of online algorithms without having to optimize over the class of threshold-based policies.

\subsection{Separating Online and Offline Optimal Schedules (\Cref{thm:online-gap})}
\label{app:online_vs_offline}
We now prove a strict separation between online and offline schedulers in the Poisson arrival setting, establishing~\Cref{thm:online-gap}. By~\Cref{thm:online-thresholds} (or \Cref{cor:online-threshold-optimality}), it suffices to compare the offline optimum against the threshold schedulers $\sched_k^\on$, for $k\in\{1,\ldots,\capac\}$. The main subtlety is that the offline optimum need not respect the regeneration cycles of a threshold scheduler. We therefore do not decompose $\sched^*$ cycle-by-cycle. Instead, for each threshold scheduler $\sched_k^\on$, we construct a feasible comparison schedule that agrees with $\sched_k^\on$ except on certain good cycles, where it uses a strictly better local schedule. The replacement schedule on each good cycle completes no later than $\sched_k^\on$ completes that cycle, so these local replacements can be concatenated across cycles. Since $\sched^*$ is offline optimal, it is no worse than this comparison schedule. We first formalize this reduction below.

\begin{lemma}
\label{lemma:cycle-improvement-rate-gap}
Assume $\lambda s<\capac$ and fix $k\in\{1,\ldots,\capac\}$. Suppose that, on each regenerative cycle of $\sched_k^\on$ (as defined in~\Cref{def:stationary-timer-cycle}), there is a local replacement rule with the following properties:
\begin{enumerate}
\item the replacement schedule is feasible for the jobs in that cycle;
\item the replacement schedule completes no later than $\sched_k^\on$ completes the cycle;
\item if $\Delta_k$ denotes the flow-time improvement in that cycle, then $\expect{\Delta_k}\ge \Gamma_k$ for some $\Gamma_k>0$.
\end{enumerate}
Then
$$
\rate(\sched_k^\on)-\rate(\sched^*)
\ge
\frac{\Gamma_k}{\expect{\cycjobs(\sched_k^\on)}}.
$$
\end{lemma}

\begin{proof}
Decompose the infinite execution of $\sched_k^\on$ into its regenerative cycles. Let $(F_i,N_i)_{i\ge 1}$ denote the cycle flow times and cycle sizes, so that $F_i$ is the flow time incurred by $\sched_k^\on$ during cycle $i$, and $N_i$ is the number of jobs in that cycle. We construct a feasible comparison scheduler $\sched'_k$ by applying the local replacement rule independently on each cycle, and then idling until the original cycle boundary if the replacement schedule finishes early. Because each replacement completes no later than the corresponding cycle of $\sched_k^\on$, the modification does not affect feasibility in later cycles. Thus $\sched'_k$  is feasible on the entire arrival sequence.

The cycle sizes of $\sched'_k$ are the same as those of $\sched_k^\on$, while the expected flow per cycle decreases by at least $\Gamma_k$. Furthermore, when applying the replacement schedule we never use information from future cycles, so the process continues to regenerate. Hence, by the renewal-reward theorem,
$$
\rate(\sched'_k)
\le
\frac{
\expect{\cycflow(\sched_k^\on)}-\Gamma_k
}{
\expect{\cycjobs(\sched_k^\on)}
}
=
\rate(\sched_k^\on)
-
\frac{\Gamma_k}{\expect{\cycjobs(\sched_k^\on)}}.
$$
Since $\sched^*$ is the offline optimum, $\rate(\sched^*)\le \rate(\sched'_k)$. Combining the two inequalities and re-ordering proves the lemma.
\end{proof}

We now define, for each threshold $k$, a positive-probability local configuration on which the cycle schedule of $\sched_k^\on$ can be improved.

\begin{lemma}
\label{lemma:separation-construction}
Assume $\lambda s<\capac$. For every $k\in\{1,\ldots,\capac\}$, there is a local replacement rule for one regenerative cycle of $\sched_k^\on$ whose expected flow-time improvement satisfies
$$
\Gamma_k
\ge
\frac{\lambda s^2 e^{-2\lambda s}}{8}.
$$
\end{lemma}

\begin{proof}
All times are measured relative to the first arrival of the cycle, which we place at time $0$. First consider $k=1$. Let $\varepsilon:=s/4$ and let $\mathcal E_1$ be the event that exactly one additional job arrives in $(0,\varepsilon]$, and no other job arrives in $(\varepsilon,2s]$. Since arrivals are Poisson, we have $\prob{\mathcal E_1}
=
\lambda\varepsilon e^{-2\lambda s}$. On $\mathcal E_1$, let $x\in(0,\varepsilon]$ be the time of the second arrival. By definition, the scheduler $\sched_1^\on$ starts a singleton batch at time $0$, completing the first job at time $s$, and then starts a second singleton batch at time $s$, completing the second job at time $2s$. Since there are no further arrivals before $2s$, the cycle ends at time $2s$. The flow time of $\sched_1^\on$ on these two jobs is therefore $s+(2s-x)=3s-x$.

The replacement schedule waits until time $x$ and processes both jobs in one batch. This batch completes at time $x+s$, and its flow time is $(x+s)+s=x+2s$. Thus the improvement is
$$
(3s-x)-(x+2s)
=
s-2x
\ge
s-2\varepsilon
=
\frac{s}{2}.
$$
Moreover, $x+s\le \varepsilon+s<2s$, so the replacement completes before the original cycle ends. Therefore the expected improvement for $k=1$ is at least
$$
\Gamma_1
\ge
\lambda\varepsilon e^{-2\lambda s}\cdot \frac{s}{2}
=
\frac{\lambda s^2e^{-2\lambda s}}{8}.
$$

Now consider $k\ge 2$. Let $T_{k-1}$ be the time of the $(k-1)$st arrival after the initial arrival at time $0$. Equivalently, $T_{k-1}$ is the time at which the $k$th job in the cycle arrives, and hence the time at which $\sched_k^\on$ starts its first batch. We have $T_{k-1}\sim \mathrm{Erlang}(k-1,\lambda)$. Let $\mathcal E_k:=\{T_{k-1}\ge s\}$. On this event, $\sched_k^\on$ does not start any batch before time $s$. The replacement schedule processes the initial job immediately as a singleton batch, completing it at time $s$. On $\mathcal E_k$, we have $T_{k-1}\ge s$, so the processor is available again by the time $\sched_k^\on$ starts its first batch. From time $T_{k-1}$ onward, the replacement schedule mimics $\sched_k^\on$, with the only difference that the initial job is omitted from the first batch because it has already been completed. This is feasible, leaves the same set of pending jobs after that batch, and therefore preserves the cycle completion time.

The only job whose completion time changes is the initial job. Under $\sched_k^\on$, it completes at time $T_{k-1}+s$; under the replacement schedule, it completes at time $s$. Thus the improvement on $\mathcal E_k$ is exactly $T_{k-1}$. Therefore
$$
\Gamma_k
\ge
\expect{T_{k-1}\mathbf 1(T_{k-1}\ge s)}.
$$
Using the Erlang tail identity,
$$
\expect{T_{k-1}\mathbf 1(T_{k-1}\ge s)}
=
\frac{k-1}{\lambda}
e^{-\lambda s}
\sum_{j=0}^{k-1}
\frac{(\lambda s)^j}{j!} \ge \frac{e^{-\lambda s}}{\lambda}.
$$
since $k\ge 2$. Finally, using $8e^x>x^2$ for all $x>0$, with $x=\lambda s$, we get
$$
\Gamma_k\ge \frac{e^{-\lambda s}}{\lambda}
\ge
\frac{\lambda s^2e^{-2\lambda s}}{8}.
$$
for every $k\ge 2$, completing the proof.
\end{proof}

The final step towards the proof of~\Cref{thm:online-gap} is to upper bound the expected number of jobs in a threshold cycle.

\begin{lemma}
\label{lemma:threshold-cycle-size-bound}
Assume $\lambda s<\capac$. For every $k\in\{1,\ldots,\capac\}$, we have
$$
\expect{\cycjobs(\sched_k^\on)}
\le
\frac{\capac^2}{\capac-\lambda s}e^{\lambda s}.
$$
\end{lemma}

\begin{proof}
Let $\mu:=\lambda s$. We decompose a cycle of $\sched_k^\on$ into low batches and high-backlog excursions. Call a batch \emph{low} if it starts with at most $\capac$ pending jobs. Since $\sched_k^\on$ starts a batch only when at least $k$ jobs are pending, every low batch completes between $k$ and $\capac$ jobs. A low batch processes all jobs pending at its start. Therefore, after a low batch completes, the system is empty if and only if no arrivals occurred during its service interval. If $A\sim\mathrm{Poisson}(\mu)$ is the number of arrivals during the batch processing, this happens with probability
$$
\prob{A=0}=e^{-\mu}.
$$
By independence of arrivals in disjoint service intervals, the number of low batches in one cycle is geometric with success probability $e^{-\mu}$, and hence has expectation $e^\mu$.

We next bound the expected number of jobs completed between two consecutive low batches. If a low batch receives at most $\capac$ arrivals during service, then after it completes the queue size is at most $\capac$, so the next batch is again low. If instead it receives $A=a>\capac$ arrivals, the process enters a high-backlog excursion starting with queue size $a$. During such an excursion, every batch is full. Let $Z_\ell$ be the queue size immediately after the $\ell$th full batch in the excursion, with $Z_0=a$, and let
$$
T:=\inf\{\ell\ge 0: Z_\ell\le \capac\}.
$$
While $Z_\ell>\capac$, the next full batch removes $\capac$ jobs and receives an independent $\mathrm{Poisson}(\mu)$ number of arrivals. Hence
$$
\expect{Z_{\ell+1}-Z_\ell\mid Z_\ell>\capac}
=
\mu-\capac.
$$
Since $\mu<\capac$, this drift is negative. For every $m\ge 1$, applying the drift identity to the stopped process gives
$$
\expect{Z_{T\wedge m}}
\le
a-(\capac-\mu)\expect{T\wedge m}.
$$
Since $Z_{T\wedge m}\ge 0$, we therefore have
$$
\expect{T\wedge m}
\le
\frac{a}{\capac-\mu}.
$$
Thus, letting $m\to\infty$ and using monotone convergence yields
$$
\expect{T\mid Z_0=a}
\le
\frac{a}{\capac-\mu}.
$$
Each full batch in the excursion completes exactly $\capac$ jobs. Therefore, conditional on $A=a>\capac$, the expected number of jobs completed during the ensuing high-backlog excursion is at most
$$
\capac\cdot \frac{a}{\capac-\mu}.
$$
Averaging over $A$, the expected number of jobs completed in the high-backlog excursion generated by one low batch is at most
$$
\frac{\capac}{\capac-\mu}
\expect{A\mathbf 1(A>\capac)}.
$$

Thus the expected number of jobs completed in one phase, consisting of one low batch and the high-backlog excursion immediately following it, is at most
$$
\capac
+
\frac{\capac}{\capac-\mu}
\expect{A\mathbf 1(A>\capac)}.
$$
Since the expected number of low batches in a cycle is $e^\mu$, we obtain
$$
\expect{\cycjobs(\sched_k^\on)}
\le
e^\mu
\left(
\capac
+
\frac{\capac}{\capac-\mu}
\expect{A\mathbf 1{A>\capac}}
\right).
$$
Using $\expect{A\mathbf 1(A>\capac)}\le \expect{A}=\mu$, this becomes
$$
\expect{\cycjobs(\sched_k^\on)}
\le
e^\mu\capac
\left(
1+\frac{\mu}{\capac-\mu}
\right)
=
e^\mu\frac{\capac^2}{\capac-\mu}.
$$
Substituting $\mu=\lambda s$ proves the lemma.
\end{proof}

Combining the optimality of threshold policies (\Cref{thm:online-thresholds}) with the reduction technique of \Cref{lemma:cycle-improvement-rate-gap} and the bounds of \Cref{lemma:separation-construction} and \Cref{lemma:threshold-cycle-size-bound}, we obtain the proof of \Cref{thm:online-gap} with separation constant
$$
\beta
\ge
\frac{\lambda s^2(\capac-\lambda s)}{8\capac^2}e^{-3\lambda s}.
$$

\subsection{The Analysis of Online Greedy (Towards \Cref{lemma:bad_event_bound})}
\label{app:poisson_greedy}
We next record quantitative bounds for the regenerative cycles of the greedy threshold policy $\sched_1^\on$. Recall that $\sched_1^\on$ starts a batch immediately whenever the processor is idle and at least one job is pending. These bounds will be used in the analysis of PrefixLookahead: when the algorithm fails to find a useful offline prefix, it falls back to $\sched_1^\on$ until the system clears. Thus we need quantitative control over the duration of one greedy fallback cycle; the resulting moment bounds will later be used to control the size and flow contribution of such fallback periods.

As in~\Cref{def:stationary-timer-cycle}, let $\cycjobs(\sched_1^\on)$ and $\cycflow(\sched_1^\on)$ denote the number of jobs and total flow time in one regenerative cycle of $\sched_1^\on$. We also let $\cycdur(\sched_1^\on)$ has an exponentially decaying tail, starting with one job at time $0$ and followed by a fresh rate-$\lambda$ Poisson arrival process. We will show that $\cycdur(\sched_1^\on)$ decays exponentially fast, with rate that is controlled by the following parameter.

\begin{definition}[Greedy tail parameter]
\label{def:greedy-tail-parameter}
Let $\mu:=\lambda s$. In the stable regime $\mu<\capac$, define
$$
c_{\gr}
:=
\capac\log\left(\frac{\capac}{\mu}\right)-(\capac-\mu)>0,
$$
and
$$
\gamma_{\gr}
:=
\min\left\{
\frac{c_{\gr}}{2},
\log\left(1+\frac{e^{-\mu}}{4}\right),
\log\left(
1+
\frac{e^{c_{\gr}-\mu}}{4}
\left(1-e^{-c_{\gr}/2}\right)
\right)
\right\}.
$$
\end{definition}

Although $\gamma_{\gr}$ is the exponential rate used in our proofs, its definition is somewhat hard to parse. We will only need that it is strictly positive and depends only on $\lambda, s,\capac$, but it might be useful to keep in mind the following simpler lower bound, sacrificing precision for simplicity:
$$
\gamma_{\gr}
\ge
\hat{\gamma}_{\gr}
:=
\frac{e^{-\lambda s}}{64}
\min\left\{
1,\frac{(\capac-\lambda s)^2}{\capac}
\right\}.
$$
The above expression follows by lower-bounding the three terms in the minimum that defines $\gamma_\gr$ using
\[\capac\log(\frac{\capac}{\mu})-(\capac-\mu)\ge \frac{(\capac - \mu)^2}{2\capac}\;\; ,\;\;
1-e^{-z}\ge \frac12\min\{z,1\}\;\; \text{and}\;\;
\log(1+y)\ge \frac{y}{2} \;\forall 0\le y\le1 .\]

We now proceed with our main result for this section. The following lemma gives the exponential tail estimate for the greedy cycle duration, together with the moment bounds that we will use later.

\begin{lemma}
\label{lemma:s1-cycle-tail}
Assume $\lambda s<\capac$. Then, for every $x>0$,
$$
\prob{\cycdur(\sched_1^\on)>x}
\le
\frac52
\exp\left(-\gamma_{\gr}\frac{x}{s}\right).
$$
Consequently, for every $\eta\in(0,\gamma_{\gr}/s)$,
$$
\expect{
\exp\left(
\eta\cycdur(\sched_1^\on)
\right)
}
\le
1+
\frac{5\eta}{2(\gamma_{\gr}/s-\eta)}.
$$
and for every integer $r\ge1$,
$$
\expect{\cycdur(\sched_1^\on)^r}
\le
\frac{5r!}{2}
\frac{s^r}{\gamma_{\gr}^r}.
$$
\end{lemma}

\begin{proof}
Let $D:=\cycdur(\sched_1^\on)$ be the greedy cycle duration. During each service interval of length $s$, the number of arrivals is distributed as
$$
X\sim\mathrm{Poisson}(\mu),
\qquad
\mu=\lambda s.
$$

We dominate the greedy cycle by a simpler pessimistic process. Whenever the system is nonempty but has fewer than $\capac$ pending jobs, the process runs one batch of duration $s$ and samples a fresh $X\sim\mathrm{Poisson}(\mu)$. If $X=0$, it halts. If $1\le X<\capac$, it returns to the same type of state. If $X\ge\capac$, it enters a high-backlog phase in which it repeatedly runs full batches until the queue size drops below $\capac$, and then returns to the same type of state. Let $T_{\mathrm{sys}}$ be the duration of this pessimistic process in real time. By construction, the clearing time $T_\mathrm{sys}$ first order stochastically dominates $D$, i.e. we have $D\preceq T_{\mathrm{sys}}$. It is therefore enough to prove an exponential tail for $T_{\mathrm{sys}}$.

Let
$$
p:=\prob{X=0}=e^{-\mu}.
$$
Each visit to the root state halts with probability $p$. Conditional on a failed visit ($X>0$), the duration of that visit is
$$
T_{\mathrm{loop}}=s(1+\widetilde Z),
$$
where
$$
\widetilde Z=
\begin{cases}
0, & 1\le X<\capac,\\
Z(X), & X\ge\capac,
\end{cases}
\qquad
X\sim\mathrm{Poisson}(\mu)\mid X>0.
$$
Here $Z(x)$ is the number of full batches needed for the queue to drop below $\capac$, starting with $x\ge\capac$ pending jobs.

We first control $Z(x)$. During the full-batch phase, let $A_1,A_2,\ldots$ be i.i.d. $\mathrm{Poisson}(\mu)$ random variables, where $A_i$ is the number of arrivals during the $i$-th full batch. If the phase has not ended after $m$ full batches, then
$$
x-m\capac+\sum_{i=1}^m A_i\ge \capac.
$$
Thus, for any $\theta>0$,
$$
\begin{aligned}
\prob{Z(x)>m}
\le
\prob{
\sum_{i=1}^m A_i
\ge m\capac-x+\capac
} \le
\exp\left(
\theta(x-\capac)
-
m\left(\theta\capac-\mu(e^\theta-1)\right)
\right)
\end{aligned}
$$
where we used Markov's inequality. Choose
$$
\theta:=\log\left(\frac{\capac}{\mu}\right).
$$
Then
$$
\theta\capac-\mu(e^\theta-1)
=
\capac\log\left(\frac{\capac}{\mu}\right)-(\capac-\mu)
=
c_{\gr}.
$$
Therefore, for $x\ge\capac$,
$$
\prob{Z(x)>m}
\le
e^{\theta(x-\capac)}e^{-c_{\gr}m}.
$$
Since $Z(x)$ is nonnegative and integer-valued, for every $0<a<c_{\gr}$,
$$
\begin{aligned}
\expect{e^{aZ(x)}}
=
1+(e^a-1)\sum_{m=0}^{\infty}e^{am}\prob{Z(x)>m} 
\le
1+
\frac{e^a-1}{1-e^{a-c_{\gr}}}
e^{\theta(x-\capac)}.
\end{aligned}
$$

We now set $a:=\gamma_{\gr}$. By the definition of $\gamma_{\gr}$, we have $a\le c_{\gr}/2$, and hence $1-e^{a-c_{\gr}}\ge 1-e^{-c_{\gr}/2}$. Again by the definition of $\gamma_{\gr}$,
$$
e^a-1
\le
\frac{e^{c_{\gr}-\mu}}{4}
\left(1-e^{-c_{\gr}/2}\right)
=
\frac{p e^{c_{\gr}}}{4}
\left(1-e^{-c_{\gr}/2}\right).
$$
Consequently,
$$
\frac{e^a-1}{1-e^{a-c_{\gr}}}
\le
\frac{p e^{c_{\gr}}}{4},
$$
and therefore
$$
\expect{e^{aZ(x)}}
\le
1+\frac{p e^{c_{\gr}}}{4}e^{\theta(x-\capac)}.
$$

Averaging over the failed visit gives
$$
\begin{aligned}
\expect{e^{a\widetilde Z}}
&\le
\frac{
1-p
+
\frac{p e^{c_{\gr}}}{4}
\sum_{x=\capac}^{\infty}
e^{-\mu}\frac{\mu^x}{x!}
e^{\theta(x-\capac)}
}{1-p}.
\end{aligned}
$$
The remaining sum is at most
$$
\begin{aligned}
\sum_{x=0}^{\infty}
e^{-\mu}\frac{\mu^x}{x!}
e^{\theta(x-\capac)}
=
e^{-\theta\capac}\expect{e^{\theta X}} 
=
\exp\left(-\theta\capac+\mu(e^\theta-1)\right)
=
e^{-c_{\gr}}.
\end{aligned}
$$
Thus
$$
\expect{e^{a\widetilde Z}}
\le
\frac{1-p+p/4}{1-p}.
$$
Since $T_{\mathrm{loop}}=s(1+\widetilde Z)$ and $a=\gamma_{\gr}$,
$$
\expect{
\exp\left(\frac{\gamma_{\gr}}{s}T_{\mathrm{loop}}\right)
}
=
e^{a}\expect{e^{a\widetilde Z}}
\le
e^{a}\frac{1-p+p/4}{1-p}.
$$
Moreover, $a=\gamma_{\gr}\le\log(1+p/4)$, so $e^a\le 1+p/4$. Hence
$$
(1-p)
\expect{
\exp\left(\frac{\gamma_{\gr}}{s}T_{\mathrm{loop}}\right)
}
\le
\left(1+\frac p4\right)
\left(1-\frac{3p}{4}\right)
\le
1-\frac p2.
$$

The number of failed visits before the first successful one is geometric. Therefore,
$$
\begin{aligned}
\expect{
\exp\left(\frac{\gamma_{\gr}}{s}T_{\mathrm{sys}}\right)
}
=
p e^{\gamma_{\gr}}
\sum_{\ell=0}^{\infty}
\left[
(1-p)
\expect{
\exp\left(\frac{\gamma_{\gr}}{s}T_{\mathrm{loop}}\right)
}
\right]^\ell 
\le
\frac{p e^{\gamma_{\gr}}}{p/2}
=
2e^{\gamma_{\gr}}
\le
\frac52.
\end{aligned}
$$
Markov's inequality gives, for every $x>0$,
$$
\prob{T_{\mathrm{sys}}>x}
\le
\frac52
\exp\left(-\gamma_{\gr}\frac{x}{s}\right).
$$
Since $D\preceq T_{\mathrm{sys}}$, this proves
$$
\prob{\cycdur(\sched_1^\on)>x}
\le
\frac52
\exp\left(-\gamma_{\gr}\frac{x}{s}\right).
$$

It remains to derive the moment bounds. For every $\eta\in(0,\gamma_{\gr}/s)$, tail integration gives
$$
\begin{aligned}
\expect{e^{\eta D}}
=
1+\eta\int_0^\infty e^{\eta x}\prob{D>x}dx 
\le
1+\frac{5\eta}{2}
\int_0^\infty
\exp\left(-\left(\frac{\gamma_{\gr}}{s}-\eta\right)x\right)dx 
=
1+
\frac{5\eta}{2(\gamma_{\gr}/s-\eta)}.
\end{aligned}
$$
Similarly, for every integer $r\ge1$,
$$
\begin{aligned}
\expect{D^r}
=
r\int_0^\infty x^{r-1}\prob{D>x}dx 
\le
\frac{5r}{2}
\int_0^\infty x^{r-1}
\exp\left(-\gamma_{\gr}\frac{x}{s}\right)dx 
=
\frac{5r!}{2}
\frac{s^r}{\gamma_{\gr}^r}.
\end{aligned}
$$
This completes the proof.
\end{proof}

\subsection{The Analysis of PrefixLookahead (\Cref{thm:poisson-upper})}
\label{app:poisson-lookahead-upper}
We now prove~\Cref{thm:poisson-upper}. The algorithm $\algname$ is stated in~\Cref{alg:poisson_lookahead}; here we give the formal renewal and tail estimates behind the proof sketch in~\Cref{sec:stoch-lookahead}. Throughout this section, assume $\lambda s<\capac$. We use the cycle decomposition induced by $\algname$. Each cycle starts when the system is empty and the next job arrives; we shift this first arrival to time $0$. The cycle ends either when $\algname$ executes a certified safe prefix and clears the system, or when no safe prefix is found in the lookahead window and $\algname$ runs $\sched_1^\on$ until the system clears. Let $\inst_\cyc$ denote the random arrival sequence in one such cycle. Write
$$
N(\inst_\cyc):=|\inst_\cyc|,
\qquad
F(\inst_\cyc):=\text{flow time incurred by }\algname\text{ on the cycle}.
$$
We also write $\bad(\inst_\cyc)$ for the event that $\algname$ finds no safe prefix inside the lookahead window and therefore enters fallback mode. Finally, $\sched_1^\on(\inst_\cyc)$ denotes the flow time incurred by running the greedy threshold scheduler on the same cycle instance.

The first step is to reduce the rate comparison to two quantities: the probability of a bad cycle and the second moment of the greedy fallback cost.

\begin{lemma}
\label{lemma:lookahead-cycle-reduction}
Assume $\lambda s<\capac$. The cycles of $\algname$ are i.i.d., the ordinary rate of $\algname$ exists, and
$$
\rate(\algname)
=
\frac{
\expect{F(\inst_\cyc)}
}{
\expect{N(\inst_\cyc)}
}.
$$
Moreover,
$$
\rate(\algname)-\rate(\sched^*)
\le
\sqrt{
\prob{\bad(\inst_\cyc)}
\cdot
\expect{\bigl(\sched_1^\on(\inst_\cyc)\bigr)^2}
}.
$$
\end{lemma}

\begin{proof}
At the end of every cycle, the processor is idle and there are no pending jobs. The time until the next arrival is independent of the past, and after shifting that next arrival to time $0$, the future arrivals form a fresh rate-$\lambda$ Poisson process as $\algname$ hasn't used any future information by construction. Hence, by the strong Markov property, the cycles of $\algname$ are identically and independently distributed. Furthermore, the cycle has finite first moments. On a good cycle, it contains only jobs in the finite lookahead window. On a bad cycle, the fallback execution is dominated by the enlarged greedy execution analyzed in~\Cref{lemma:lookahead-fallback-second-moment} below, which has finite second moment. Therefore
$$
\expect{N(\inst_\cyc)}<\infty,
\qquad
\expect{F(\inst_\cyc)}<\infty.
$$
The renewal-reward theorem gives
$$
\rate(\algname)
=
\frac{
\expect{F(\inst_\cyc)}
}{
\expect{N(\inst_\cyc)}
}.
$$

We next lower-bound the offline rate using the same cycles. Run $\algname$ on an infinite Poisson arrival sequence, and let $\inst_\cyc^1,\inst_\cyc^2,\ldots$ be the resulting cycle instances. Set
$$
N_i:=|\inst_\cyc^i|,
\qquad
O_i:=\opt(\inst_\cyc^i).
$$
The pairs $(N_i,O_i)_{i\ge1}$ are i.i.d., and $\expect{O_1}<\infty$ because $O_1\le F(\inst_\cyc^1)$. Let $\inst(n)$ be the first $n$ arrivals of the infinite sequence, and let
$$
K(n):=
\max\left\{
k:
\sum_{i=1}^k N_i\le n
\right\}.
$$
By superadditivity of the offline optimal flow time (\Cref{lemma:superadditive}), cutting the first $n$ jobs into consecutive subinstances can only decrease the total optimum. Therefore, for every realization,
$$
\opt(\inst(n))
\ge
\sum_{i=1}^{K(n)}\opt(\inst_\cyc^i)
=
\sum_{i=1}^{K(n)}O_i.
$$
Taking expectations, dividing by $n$, and applying the renewal-reward theorem gives
$$
\rate(\sched^*)
\ge
\frac{
\expect{\opt(\inst_\cyc)}
}{
\expect{N(\inst_\cyc)}
},
$$
where we used the existence of the offline rate (see~\Cref{app:subsec_limit_existance}).

It remains to compare $F(\inst_\cyc)$ with $\opt(\inst_\cyc)$ inside one cycle. If the cycle is good, $\algname$ executes an offline-optimal schedule for a prefix that completes before the next arrival outside the prefix. Hence this prefix is the entire cycle and
$$
F(\inst_\cyc)=\opt(\inst_\cyc).
$$
If the cycle is bad, then $\algname$ runs $\sched_1^\on$ until the system clears, and therefore
$$
F(\inst_\cyc)-\opt(\inst_\cyc)
\le
\sched_1^\on(\inst_\cyc).
$$
Combining the two cases,
$$
F(\inst_\cyc)-\opt(\inst_\cyc)
\le
\mathbf 1(\bad(\inst_\cyc))\cdot\sched_1^\on(\inst_\cyc)
$$
and thus
$$
\begin{aligned}
\rate(\algname)-\rate(\sched^*)
&\le
\frac{
\expect{F(\inst_\cyc)-\opt(\inst_\cyc)}
}{
\expect{N(\inst_\cyc)}
} 
\le
\expect{
\mathbf 1(\bad(\inst_\cyc))\cdot\sched_1^\on(\inst_\cyc)
},
\end{aligned}
$$
where we used $N(\inst_\cyc)\ge1$. Cauchy--Schwarz now gives
$$
\expect{
\mathbf 1{\bad(\inst_\cyc)}\sched_1^\on(\inst_\cyc)
}
\le
\sqrt{
\prob{\bad(\inst_\cyc)}
\cdot
\expect{\bigl(\sched_1^\on(\inst_\cyc)\bigr)^2}
}
$$
and the proof follows.
\end{proof}

We now define the constants used in the two technical estimates. Let $\mu:=\lambda s$ and let $\gamma_{\gr}$ be the greedy tail parameter from~\Cref{def:greedy-tail-parameter}. Define
$$
p_H
:=
e^{-2\mu}-e^{-3\mu},
$$
and
$$
\eta_H
:=
\min\left\{
\frac{2p_H\gamma_{\gr}}{27},
\frac{p_H\mu}{16},
\frac{p_H}{20}
\right\}.
$$
Note that both parameters are strictly positive constants depending only on $\lambda,s,\capac$.

We prove the two estimates needed in~\Cref{lemma:lookahead-cycle-reduction}: bad cycles are exponentially unlikely, and the greedy fallback cost has polynomial second moment in $L$.

\begin{lemma}[Formal version of~\Cref{lemma:bad_event_bound}]
\label{lemma:lookahead-bad-probability}
For every $L\ge3s$,
$$
\prob{\bad(\inst_\cyc)}
\le
4\exp\left(-\eta_H\frac{L}{s}\right).
$$
\end{lemma}

\begin{proof}
Shift the first arrival of the cycle to time $0$. For any time $t$, let $r^+(t)$ be the first arrival time strictly after $t$. Run $\sched_1^\on$ from the same initial state, on the same Poisson arrival sequence. Define $H_L$ to be the event that there exists a clearing time $T\le L-3s$ of this greedy execution such that
$$
r^+(T)\in(T+2s,T+3s].
$$
We first show that
$$
H_L\subseteq \overline{\bad}(\inst_\cyc).
$$
Suppose $H_L$ occurs, and fix such a clearing time $T$. Let $\inst_T$ be the prefix consisting of all jobs that arrive up to time $T$. Since $\sched_1^\on$ clears the system at time $T$, it gives a feasible schedule for $\inst_T$ with makespan $T$. By \Cref{lemma:opt-makespan}, there exists an offline-optimal schedule for $\inst_T$ whose makespan is at most $T+2s$. Since $r^+(T)>T+2s$, this offline-optimal schedule completes before the next arrival after the prefix. Moreover, $r^+(T)\le T+3s\le L$, so the certifying next arrival is visible inside the lookahead window. Therefore $\algname$ finds a safe prefix and does not enter fallback mode. Hence
$$
\prob{\bad(\inst_\cyc)}
\le
\prob{\overline{H_L}}.
$$

It remains to bound $\prob{\overline{H_L}}$. We work in time units scaled by $s$. Let
$$
\widehat L:=\frac{L}{s}.
$$
Let $\widehat R$ be the duration of one regenerative cycle of $\sched_1^\on$, divided by $s$. By~\Cref{lemma:s1-cycle-tail}, applied to the scaled duration, for every $\eta\in(0,\gamma_{\gr})$,
$$
\expect{e^{\eta\widehat R}}
\le
1+\frac{5\eta}{2(\gamma_{\gr}-\eta)}.
$$
After each greedy clearing time, the scaled idle time until the next arrival is an independent exponential random variable
$$
\widehat Z\sim\mathrm{Exp}(\mu).
$$
A clearing time certifies $H_L$ exactly when the following idle gap lies in $(2,3]$. This has probability
$$
p_H
=
\prob{\widehat Z\in(2,3]}
=
e^{-2\mu}-e^{-3\mu}.
$$

Fix $\eta\le\eta_H$. Since $\eta\le 2p_H\gamma_{\gr}/27$ and $p_H\le1$,
$$
\expect{e^{\eta\widehat R}}
\le
1+\frac{p_H}{5}.
$$

Also, since $\eta\le p_H\mu/16$, we have
$$
\begin{aligned}
\expect{e^{\eta\widehat Z}\mathbf 1(\widehat Z\notin(2,3])}
\le
\expect{e^{\eta\widehat Z}}
-
\prob{\widehat Z\in(2,3]} =
\frac{\mu}{\mu-\eta}-p_H \le
1+\frac{2\eta}{\mu}-p_H \le
1-\frac{7p_H}{8}.
\end{aligned}
$$
Therefore
$$
\expect{e^{\eta\widehat R}} \cdot
\expect{e^{\eta\widehat Z}\mathbf 1(\widehat Z\notin(2,3])}
\le
\left(1+\frac{p_H}{5}\right)
\left(1-\frac{7p_H}{8}\right)
\le
1-\frac{27p_H}{40}.
$$

Let $A$ be the scaled time of the first greedy clearing epoch whose following idle gap lies in $(2,3]$. If $A\le \widehat L-3$, then $H_L$ occurs. Thus
$$
\prob{\overline{H_L}}
\le
\prob{A>\widehat L-3}.
$$
We bound $A$ by an exponential moment. Summing over the index of the first certifying gap,
$$
\begin{aligned}
\expect{e^{\eta A}}
&\le
\sum_{j=1}^{\infty}
p_H
\expect{e^{\eta\widehat R}}^j
\expect{e^{\eta\widehat Z}\mathbf 1{\widehat Z\notin(2,3]}}^{j-1} \\
&=
\frac{
p_H\expect{e^{\eta\widehat R}}
}{
1-
\expect{e^{\eta\widehat R}},
\expect{e^{\eta\widehat Z}\mathbf 1{\widehat Z\notin(2,3]}}
} \\
&\le
\frac{
p_H(1+p_H/5)
}{
27p_H/40
}
\le
\frac{40}{17}.
\end{aligned}
$$
Taking $\eta=\eta_H$ and applying Markov's inequality gives
$$
\begin{aligned}
\prob{A>\widehat L-3}
&\le
\expect{e^{\eta_H A}}e^{-\eta_H(\widehat L-3)} \
&\le
\frac{40}{17}
\exp\left(-\eta_H\left(\frac{L}{s}-3\right)\right).
\end{aligned}
$$
Since $\eta_H\le p_H/20\le1/20$,
$$
\frac{40}{17}e^{3\eta_H}\le4.
$$
Hence
$$
\prob{\bad(\inst_\cyc)}
\le
\prob{\overline{H_L}}
\le
4\exp\left(-\eta_H\frac{L}{s}\right).
$$
\end{proof}

\begin{lemma}
\label{lemma:lookahead-fallback-second-moment}
For every $L>0$,
$$
\expect{\bigl(\sched_1^\on(\inst_\cyc)\bigr)^2}
\le
\frac{240\capac^2}{\gamma_{\gr}^4}(L+s)^2.
$$
\end{lemma}

\begin{proof}
Couple the cycle $\inst_\cyc$ with an infinite rate-$\lambda$ Poisson arrival sequence starting with one job at time $0$, and run $\sched_1^\on$ on this infinite sequence. Let $\widetilde{\inst}_L$ be the instance consisting of all jobs completed by $\sched_1^\on$ up to the first time at or after $L$ when the system is empty. This enlarged instance contains every job that can affect the fallback execution of $\algname$, and possibly more. Therefore
$$
\sched_1^\on(\inst_\cyc)
\le
\sched_1^\on(\widetilde{\inst}_L)
\qquad\text{with probability }1.
$$

Decompose the execution of $\sched_1^\on$ on $\widetilde{\inst}_L$ into complete greedy regenerative cycles, and let
$$
F_1,F_2,\ldots
$$
be their cycle flow times. These are i.i.d. copies of $\cycflow(\sched_1^\on)$. Let $M_L$ be the number of greedy cycles contained in $\widetilde{\inst}_L$. Since every nonempty greedy cycle lasts at least one service time $s$, and $\widetilde{\inst}_L$ stops at the first clearing time at or after $L$,
$$
M_L\le \left\lceil\frac{L}{s}\right\rceil+1.
$$
Thus
$$
\sched_1^\on(\widetilde{\inst}_L)
=
\sum_{i=1}^{M_L}F_i
\le
\sum_{i=1}^{\lceil L/s\rceil+1}F_i.
$$
By Cauchy--Schwarz,
$$
\left(
\sum_{i=1}^{\lceil L/s\rceil+1}F_i
\right)^2
\le
\left(\left\lceil\frac{L}{s}\right\rceil+1\right)
\sum_{i=1}^{\lceil L/s\rceil+1}F_i^2.
$$
Taking expectations,
$$
\expect{\bigl(\sched_1^\on(\inst_\cyc)\bigr)^2}
\le
\left(\left\lceil\frac{L}{s}\right\rceil+1\right)^2
\expect{\cycflow(\sched_1^\on)^2}.
$$
Since
$$
\left\lceil\frac{L}{s}\right\rceil+1
\le
\frac{2(L+s)}{s},
$$
we get
$$
\expect{\bigl(\sched_1^\on(\inst_\cyc)\bigr)^2}
\le
\frac{4(L+s)^2}{s^2}
\expect{\cycflow(\sched_1^\on)^2}.
$$

It remains to bound the second moment of one greedy cycle. Let $D:=\cycdur(\sched_1^\on)$. During one greedy cycle, the processor starts a new batch immediately whenever the system is nonempty. Hence the number of batches in the cycle is exactly $D/s$. Each batch contains at most $\capac$ jobs, so
$$
\cycjobs(\sched_1^\on)
\le
\frac{\capac}{s}D.
$$
Every job completed in the cycle has flow time at most $D$, and therefore
$$
\cycflow(\sched_1^\on)
\le
\cycjobs(\sched_1^\on)D
\le
\frac{\capac}{s}D^2.
$$
Squaring and taking expectations gives
$$
\expect{\cycflow(\sched_1^\on)^2}
\le
\frac{\capac^2}{s^2}
\expect{D^4}.
$$
By~\Cref{lemma:s1-cycle-tail}, applied with moment order $4$,
$$
\expect{D^4}
\le
\frac{60s^4}{\gamma_{\gr}^4}.
$$
Thus
$$
\expect{\cycflow(\sched_1^\on)^2}
\le
\frac{60\capac^2s^2}{\gamma_{\gr}^4}.
$$
Substituting this into the previous bound yields
$$
\expect{\bigl(\sched_1^\on(\inst_\cyc)\bigr)^2}
\le
\frac{240\capac^2}{\gamma_{\gr}^4}(L+s)^2.
$$
\end{proof}

We can now complete the upper bound and provide the proof of~\Cref{thm:poisson-upper}.

\begin{corollary}
\label{cor:lookahead-upper-bound}
For every $L\ge3s$,
$$
\rate(\algname)-\rate(\sched^*)
\le
\frac{4\sqrt{60}\capac}{\gamma_{\gr}^2}
(L+s)
\exp\left(-\frac{\eta_H}{2s}L\right).
$$
Consequently, there exist constants $\kappa,\gamma>0$, depending only on $\lambda,s,\capac$, such that for every $L>0$,
$$
\rate(\algname)
\le
\rate(\sched^*)+\kappa e^{-\gamma L}.
$$
\end{corollary}

\begin{proof}
By~\Cref{lemma:lookahead-cycle-reduction}, we have that
$$
\rate(\algname)-\rate(\sched^*)
\le
\sqrt{
\prob{\bad(\inst_\cyc)}
\cdot
\expect{\bigl(\sched_1^\on(\inst_\cyc)\bigr)^2}
}.
$$
For $L\ge3s$, substituting~\Cref{lemma:lookahead-bad-probability} and~\Cref{lemma:lookahead-fallback-second-moment} gives
$$
\begin{aligned}
\rate(\algname)-\rate(\sched^*)
&\le
\sqrt{
4e^{-\eta_H L/s}
\cdot
\frac{240\capac^2}{\gamma_{\gr}^4}(L+s)^2
} \
&=
\frac{4\sqrt{60}\capac}{\gamma_{\gr}^2}
(L+s)
\exp\left(-\frac{\eta_H}{2s}L\right).
\end{aligned}
$$

It remains to absorb the polynomial prefactor into the exponential constant. Let
$$
A:=\frac{4\sqrt{60}\capac}{\gamma_{\gr}^2},
\qquad
b:=\frac{\eta_H}{2s},
\qquad
\gamma:=\frac b2=\frac{\eta_H}{4s}.
$$
For $L\ge3s$,
$$
A(L+s)e^{-bL}
=
A\bigl((L+s)e^{-\gamma L}\bigr)e^{-\gamma L}
\le
A\left(s+\frac{1}{e\gamma}\right)e^{-\gamma L}.
$$

This completes the proof for $L\ge 3s$. For $0<L<3s$, we use the same cycle reduction with the trivial bound $\prob{\bad(\inst_\cyc)}\le1$ and~\Cref{lemma:lookahead-fallback-second-moment}. This gives
$$
\rate(\algname)-\rate(\sched^*)
\le
\frac{\sqrt{240}\capac}{\gamma_{\gr}^2}(L+s)
\le
\frac{4\sqrt{240}\capac s}{\gamma_{\gr}^2}.
$$
Since $e^{-\gamma L}\ge e^{-3s\gamma}$ on this range, the bound also has the form $\kappa e^{-\gamma L}$ after increasing $\kappa$ to
$$
\kappa
:=
\max\left\{
A\left(s+\frac{1}{e\gamma}\right),
\frac{4\sqrt{240}\capac s}{\gamma_{\gr}^2}e^{3s\gamma}
\right\}.
$$
Thus, for all $L>0$,
$$
\rate(\algname)-\rate(\sched^*)
\le
\kappa e^{-\gamma L}.
$$
This proves~\Cref{thm:poisson-upper}.
\end{proof}

\subsection{A Lower Bound on the Rate of Lookahead Algorithms (\Cref{thm:poisson-lb})}
\label{app:poisson-lookahead-lower}
We finally prove the lower bound for finite-lookahead schedulers in the Poisson arrival setting, establishing~\Cref{thm:poisson-lb}. The proof isolates a local obstruction that occurs with probability $\exp(-O(L))$ inside a regenerative cycle tailored to the offline optimum. On this event, the scheduler must make a batching decision before it can see the arrivals that determine which decision is correct. This causes a constant excess flow-time loss relative to the offline optimum. A renewal argument then converts this one-cycle loss into a rate lower bound.

\subsubsection{Offline-Certifying Cycles and the Hard Event}
\label{app:poisson-lookahead-lb-setup}

We first define the cycles used in the lower bound. These cycles are intrinsic to the offline optimum rather than to any particular online or lookahead algorithm.

\begin{definition}[Offline-certifying cycle]
\label{def:offline-certifying-cycle}
Let $0=r_0<r_1<r_2<\cdots$ be an infinite rate-$\lambda$ Poisson arrival sequence, shifted so that the first arrival of the cycle occurs at time $0$. For each $m\ge0$, let $\inst_m:=\{r_0,r_1,\ldots,r_m\}$, and let $r_{m+1}$ be the next arrival after this prefix. Define
$$
M_m^*
:=
\min\left\{
\makespan(\sched\mid \inst_m):
\sched \text{ is an offline-optimal schedule for } \inst_m
\right\}.
$$
We say that the prefix $\inst_m$ is \emph{offline-certifying} if $M_m^*<r_{m+1}$. The offline-certifying cycle $\inst_\off$ is the minimal such prefix: formally, we write $\inst_\off:=\inst_K$ for
$$
K:=\min\{m\ge0:M_m^*<r_{m+1}\}.
$$
Finally, we also write $N(\inst_\off):=|\inst_\off|\ge 1$ for the number of jobs in the cycle and $\opt(\inst_\off)$ for its offline-optimal flow time.
\end{definition}

Since the first excluded arrival index is a stopping time with respect to the interarrival sequence, the strong Markov property implies that the arrivals following it are independent of the past and distributed as a fresh Poisson process. Thus, successive offline-certifying cycles are i.i.d.

The strict inequality in the definition ensures that the cycle can be scheduled optimally and cleared before the next arrival. Thus offline-optimal schedules for consecutive cycles can be concatenated without interference. This is the decomposition that will allow us to compare every finite-lookahead scheduler to the offline optimum on a single representative cycle.

We now describe the hard event inside one offline-certifying cycle. Since we only seek a lower bound up to constants in the exponent, it suffices to prove the result for lookahead values of the form $L=ks$, where $k\ge5$ is an integer: for an arbitrary $L>0$, take $k=\max\{5,\lceil L/s\rceil\}$, so that $ks\ge L$, and the resulting bound $\exp(-O_{\lambda,\capac,s}(k))$ is the same as $\exp(-O_{\lambda,\capac,s}(L))$ after adjusting constants. We therefore fix such an integer $k$ and consider lookahead $L=ks$. Assume throughout this construction that $\capac\ge2$, and define
$$
u:=\left\lceil \frac{\capac}{2}\right\rceil,
\qquad
v:=\left\lfloor \frac{\capac}{2}\right\rfloor .
$$
Thus $u+v=\capac$, $1\le v\le u$, and
$$
\frac13\le \frac v\capac\le \frac12 .
$$
Fix the absolute constants\footnote{For ease of presentation and due to the tedious computations involved, we do not optimize these constants and simply demonstrate the desired dependency of the lookahead window $L$.}
$$
\theta:=10^{-3},
\qquad
\eta_0:=10^{-2},
\qquad
\rho:=10^{-4},
\qquad
\alpha:=10^{-6},
\qquad
\beta:=10^{-8}.
$$
For the fixed integer $k\ge5$, define
$$
\delta_k
:=
\left(\frac v\capac-\frac{\alpha}{k+1}\right)s,
\qquad
\varepsilon_k
:=
\frac{\beta s}{k+1}.
$$
The choices above ensure that, uniformly over every $\capac\ge2$ and every $k\ge5$,
$$
0<\varepsilon_k<\rho s<\theta s<\eta_0 s<\delta_k<s.
$$

Define the initial windows
$$
I_0:=(0,\varepsilon_k],
\qquad
D_k:=[\delta_k-\varepsilon_k,\delta_k+\varepsilon_k],
$$
and, for each $i=1,\ldots,k$, define the visible full-burst window
$$
J_i:=[is-\rho s,is].
$$
The two hidden-continuation windows are
$$
U_E:=
[ks+\delta_k+\eta_0s,\ (k+1)s-2\theta s],
$$
and
$$
U_W:=
[ks+\delta_k+s-2\theta s,\ ks+\delta_k+s-\theta s].
$$
Finally, let
$$
H_k:=(3k+10)s.
$$

We can now define our hard event over offline-certifying cycles.

\begin{definition}[Hard event]
\label{def:hard-offline-cycle-event}
Let $\mathcal E_k^E$ be the event that, after shifting the first arrival of an offline-certifying cycle to time $0$, the following pattern occurs:
\begin{enumerate}
    \item exactly $u-1$ additional arrivals occur in $I_0$;
    \item exactly $v$ arrivals occur in $D_k$;
    \item for every $i=1,\ldots,k$, exactly $\capac$ arrivals occur in $J_i$;
    \item exactly $u$ arrivals occur in $U_E$;
    \item no other arrivals occur in $(0,H_k]$.
\end{enumerate}
Similarly, let $\mathcal E_k^W$ be the event obtained by replacing $U_E$ with $U_W$ in requirement (4). Define
$$
\mathcal E_k:=\mathcal E_k^E\cup \mathcal E_k^W.
$$
\end{definition}

% Preamble:
% \usepackage{tikz}
% \usetikzlibrary{decorations.pathreplacing}

\begin{figure}[h!]
\centering
\begin{tikzpicture}[
    x=1cm,
    y=1cm,
    tick/.style={black, line width=0.8pt},
    axis/.style={black, line width=0.8pt},
    arr/.style={font=\small},
    timelab/.style={font=\small},
    title/.style={font=\large},
    prefixstyle/.style={draw=red!70, text=red!70},
    early/.style={draw=green!55!black, text=green!55!black},
    late/.style={draw=blue!70, text=blue!70},
    hidden/.style={line width=2.2pt},
    window/.style={line width=2.0pt},
    brace/.style={decorate, decoration={brace, amplitude=5pt}, line width=0.8pt},
    dashedmark/.style={black!55, dashed, line width=0.7pt}
]

% y-coordinates
\def\yE{0}
\def\yW{-3.45}

% x-coordinates
\def\xzero{1.0}
\def\xdelta{1.75}
\def\xs{2.75}
\def\xtwos{3.85}
\def\xmid{5.35}
\def\xks{7.35}
\def\xlook{8.05}
\def\xEhid{8.80}
\def\xkpone{9.80}
\def\xWhid{10.35}
\def\xH{13.00}

%%%%%%%%%%%%%%%%%%%%%%%%%%%%%%%%%%%%%%%%%%%%%%%%%%%%%%%%%%%%
% Early continuation
%%%%%%%%%%%%%%%%%%%%%%%%%%%%%%%%%%%%%%%%%%%%%%%%%%%%%%%%%%%%

\node[title, anchor=west] at (-0.55,{\yE+1.85}) {\small Early continuation $\mathcal E_k^E$:};

\draw[axis] (0.70,\yE) -- (13.35,\yE);

\node[arr, red!70, anchor=east] at (0.62,{\yE+0.36}) {arrivals};
\node[timelab, anchor=east] at (0.62,{\yE-0.35}) {time};

% Main ticks
\foreach \x/\lab in {
    \xzero/{\scriptsize $0$},
    \xdelta/{\scriptsize $\delta_k$},
    \xs/{\scriptsize $s$},
    \xtwos/{\scriptsize $2s$},
    \xks/{\scriptsize $ks$},
    \xkpone/{\scriptsize $(k+1)s$},
    \xH/{\scriptsize $H_k$}
}{
    \draw[tick] (\x,{\yE-0.12}) -- (\x,{\yE+0.12});
    \node[timelab, below=3pt] at (\x,\yE) {\lab};
}

% Lookahead dashed mark and callout
\draw[dashedmark] (\xlook,{\yE-0.55}) -- (\xlook,{\yE+1.10});
\node[timelab, black!60, align=center] at (\xlook,{\yE-0.92})
    {\tiny lookahead ends at\\[-2pt]\tiny $\le ks+\delta_k+\varepsilon_k$};

% Ellipses
\node at (\xmid,{\yE+0.24}) {\large $\cdots$};
\node at (11.55,{\yE+0.24}) {\large $\cdots$};

% Common-prefix arrivals
\node[arr, red!70] at (\xzero,{\yE+0.45}) {\scriptsize $u$};
\node[arr, red!70] at (\xdelta,{\yE+0.45}) {\scriptsize $v$};
\node[arr, red!70] at (\xs,{\yE+0.45}) {\scriptsize $\capac$};
\node[arr, red!70] at (\xtwos,{\yE+0.45}) {\scriptsize $\capac$};
\node[arr, red!70] at (\xks,{\yE+0.45}) {\scriptsize $\capac$};

% Interval bars for I_0, D_k, J_i
\draw[window, red!70] ({\xzero-0.09},{\yE+0.18}) -- ({\xzero+0.09},{\yE+0.18});
\draw[window, red!70] ({\xdelta-0.15},{\yE+0.18}) -- ({\xdelta+0.15},{\yE+0.18});
\draw[window, red!70] ({\xs-0.22},{\yE+0.18}) -- (\xs,{\yE+0.18});
\draw[window, red!70] ({\xtwos-0.22},{\yE+0.18}) -- (\xtwos,{\yE+0.18});
\draw[window, red!70] ({\xks-0.22},{\yE+0.18}) -- (\xks,{\yE+0.18});

% Early hidden block
\draw[hidden, green!55!black] ({\xEhid-0.18},{\yE+0.18}) -- ({\xEhid+0.18},{\yE+0.18});
\node[arr, green!55!black] at (\xEhid,{\yE+0.50}) {\scriptsize $u$};
\node[timelab, green!55!black] at (\xEhid,{\yE-0.43}) {\tiny $U_E$};

% Braces
\draw[brace, prefixstyle] (\xzero,{\yE+1.08}) -- (\xks,{\yE+1.08})
    node[midway, above=6pt] {\scriptsize common prefix};

\draw[brace, early] ({\xEhid-0.18},{\yE+1.08}) -- ({\xEhid+0.18},{\yE+1.08})
    node[midway, above=6pt] {\scriptsize hidden};

\draw[brace, early] ({\xEhid+0.25},{\yE+0.82}) -- (\xH,{\yE+0.82})
    node[midway, above=6pt] {\scriptsize no arrivals until $H_k$};

%%%%%%%%%%%%%%%%%%%%%%%%%%%%%%%%%%%%%%%%%%%%%%%%%%%%%%%%%%%%
% Late continuation
%%%%%%%%%%%%%%%%%%%%%%%%%%%%%%%%%%%%%%%%%%%%%%%%%%%%%%%%%%%%

\node[title, anchor=west] at (-0.55,{\yW+1.85}) {\small Late continuation $\mathcal E_k^W$:};

\draw[axis] (0.70,\yW) -- (13.35,\yW);

\node[arr, red!70, anchor=east] at (0.62,{\yW+0.36}) {arrivals};
\node[timelab, anchor=east] at (0.62,{\yW-0.35}) {time};

% Main ticks
\foreach \x/\lab in {
    \xzero/{\scriptsize $0$},
    \xdelta/{\scriptsize $\delta_k$},
    \xs/{\scriptsize $s$},
    \xtwos/{\scriptsize $2s$},
    \xks/{\scriptsize $ks$},
    \xkpone/{\scriptsize $(k+1)s$},
    \xH/{\scriptsize $H_k$}
}{
    \draw[tick] (\x,{\yW-0.12}) -- (\x,{\yW+0.12});
    \node[timelab, below=3pt] at (\x,\yW) {\lab};
}

% Lookahead dashed mark and callout
\draw[dashedmark] (\xlook,{\yW-0.55}) -- (\xlook,{\yW+1.10});
\node[timelab, black!60, align=center] at (\xlook,{\yW-1.00})
    {\tiny lookahead ends at\\[-2pt]\tiny $\le ks+\delta_k+\varepsilon_k$};

% Ellipses
\node at (\xmid,{\yW+0.24}) {\large $\cdots$};
\node at (11.65,{\yW+0.24}) {\large $\cdots$};

% Common-prefix arrivals
\node[arr, red!70] at (\xzero,{\yW+0.45}) {\scriptsize $u$};
\node[arr, red!70] at (\xdelta,{\yW+0.45}) {\scriptsize $v$};
\node[arr, red!70] at (\xs,{\yW+0.45}) {\scriptsize $\capac$};
\node[arr, red!70] at (\xtwos,{\yW+0.45}) {\scriptsize $\capac$};
\node[arr, red!70] at (\xks,{\yW+0.45}) {\scriptsize $\capac$};

% Interval bars for I_0, D_k, J_i
\draw[window, red!70] ({\xzero-0.09},{\yW+0.18}) -- ({\xzero+0.09},{\yW+0.18});
\draw[window, red!70] ({\xdelta-0.15},{\yW+0.18}) -- ({\xdelta+0.15},{\yW+0.18});
\draw[window, red!70] ({\xs-0.22},{\yW+0.18}) -- (\xs,{\yW+0.18});
\draw[window, red!70] ({\xtwos-0.22},{\yW+0.18}) -- (\xtwos,{\yW+0.18});
\draw[window, red!70] ({\xks-0.22},{\yW+0.18}) -- (\xks,{\yW+0.18});

% Late hidden block
\draw[hidden, blue!70] ({\xWhid-0.18},{\yW+0.18}) -- ({\xWhid+0.18},{\yW+0.18});
\node[arr, blue!70] at (\xWhid,{\yW+0.50}) {\scriptsize $u$};
\node[timelab, blue!70] at (\xWhid,{\yW-0.48}) {\tiny $\;\;\; U_W$};

% Braces
\draw[brace, prefixstyle] (\xzero,{\yW+1.08}) -- (\xks,{\yW+1.08})
    node[midway, above=6pt] {\scriptsize common prefix};

\draw[brace, late] ({\xWhid-0.18},{\yW+1.08}) -- ({\xWhid+0.18},{\yW+1.08})
    node[midway, above=6pt] {\scriptsize hidden};

\draw[brace, late] ({\xWhid+0.25},{\yW+0.82}) -- (\xH,{\yW+0.82})
    node[midway, above=6pt] {\scriptsize no arrivals until $H_k$};

\end{tikzpicture}
\caption{The hard event for the Poisson lookahead lower bound. The common prefix is identical in the early and late continuations and is visible to the scheduler. The hidden block contains $u$ arrivals in both continuations; it lies just beyond the lookahead window at the initial decision time, but determines which first batching commitment is optimal.}
\label{fig:poisson-lookahead-hard-event}
\end{figure}

A pictorial depiction of the arrival sequence in an offline-certifying cycle conditioned on the event $\mathcal E_k$ is shown in~\Cref{fig:poisson-lookahead-hard-event}. Let us briefly justify that the arrivals specified in $\mathcal E_k$ indeed form one offline-certifying cycle. No prefix ending before the completion of one of the windows $D_k,J_1,\ldots,J_k$, or before the completion of the active hidden window, can be certifying, because another specified arrival occurs within less than one service time. After all arrivals in $D_k$ and after each visible burst $J_i$ with $i<k$, the next visible burst begins before any schedule could process all jobs in the current prefix: the first full batch is not available before time $\delta_k-\varepsilon_k$, and the next burst begins at time $(i+1)s-\rho s$. Similarly, after the common prefix through $J_k$, the hidden block arrives before the prefix could be cleared. Finally, once the hidden block has arrived, there are no further arrivals until $H_k=(3k+10)s$, while all specified jobs can be processed in at most $k+2$ batches and cleared by time less than $(k+3)s$. Hence the minimal offline-certifying prefix is exactly the set of jobs specified by $\mathcal E_k$. A detailed argument is given in the proof of~\Cref{lemma:hard-event-probability} below.

The event $\mathcal E_k$ creates the desired ambiguity. The arrivals at time $0$ and in $I_0$ form an initial group of $u$ jobs, while the arrivals in $D_k$ complete this group to a full batch of $\capac$ jobs. Thus the scheduler must choose whether to process the first $u$ jobs early, or wait until approximately time $\delta_k$ and process the first full batch. The visible bursts $J_1,\ldots,J_k$ keep these two decisions nearly balanced over the entire lookahead window. The decisive information lies in the hidden continuation: on $\mathcal E_k^E$, the hidden $u$ jobs arrive early enough to complete the final $v$-job offset created by the early commitment, while on $\mathcal E_k^W$, they arrive too late and waiting is better. Since the hidden continuation lies beyond the scheduler's lookahead horizon at the time of the first consequential decision, the scheduler cannot distinguish these two cases when making that decision.

The proof reduces to three claims. First, offline-certifying cycles induce a renewal decomposition for the offline optimum, and any uniform one-cycle excess loss on such cycles yields a rate separation. Formally:

\begin{lemma}
\label{lemma:poisson-lb-renewal-reduction}
Fix an $L=ks$-lookahead scheduler $\sched$ and an event $\mathcal E$ defined over one offline-certifying cycle. Suppose there is a number $\Delta_{\mathcal E}\ge0$ such that, for every possible history and internal state of $\sched$ at the beginning of a cycle,
$$
\expect{
\sched(\inst_\off)-\opt(\inst_\off)
\mid
\mathcal E
}
\ge
\Delta_{\mathcal E},
$$
where $\sched(\inst_\off)$ denotes the flow-time contribution of the jobs of that cycle under the actual execution of $\sched$. Then
$$
\rate(\sched)-\rate(\sched^*)
\ge
\Omega_{\lambda,\capac,s}(1)\cdot
\prob{\mathcal E}\cdot
\Delta_{\mathcal E}.
$$
\end{lemma}

\noindent Then, we show that although hard cycles are rare, they still appear with exponential probability.
\begin{lemma}
\label{lemma:hard-event-probability}
The hard event satisfies
$$
\prob{\mathcal E_k}
\ge
\exp\bigl(-O_{\lambda,\capac,s}(k)\bigr).
$$
\end{lemma}

\noindent Finally, by exploiting the ambiguity over the first action built inside the definition of the hard event, we prove the following result:

\begin{lemma}
\label{lemma:hard-event-constant-loss}
For every $L=ks$-lookahead scheduler $\sched$, uniformly over every possible history and internal state of $\sched$ at the beginning of the cycle,
$$
\expect{
\sched(\inst_\off)-\opt(\inst_\off)
\mid
\mathcal E_k
}
\ge
\Omega_{\lambda,\capac,s}(1).
$$
\end{lemma}

Assuming these three claims, the theorem follows immediately. Applying~\Cref{lemma:poisson-lb-renewal-reduction} with $\mathcal E=\mathcal E_k$, and then using~\Cref{lemma:hard-event-probability,lemma:hard-event-constant-loss}, gives
$$
\rate(\sched)-\rate(\sched^*)
\ge
\exp\bigl(-O_{\lambda,\capac,s}(k)\bigr).
$$
Since arbitrary lookahead values were already absorbed into the choice $k=\max\{5,\lceil L/s\rceil\}$, this is exactly
$$
\rate(\sched)-\rate(\sched^*)
\ge
\exp\bigl(-O_{\lambda,\capac,s}(L)\bigr),
$$
which is equivalent to the statement of~\Cref{thm:poisson-lb}. It remains to prove the three claims, which we respectively do in~\Cref{app:poisson-lb-renewal-reduction},~\Cref{app:poisson-lb-hard-probability} and~\Cref{app:poisson-lb-hard-loss}.

\subsubsection{Proof of the Renewal Reduction (\Cref{lemma:poisson-lb-renewal-reduction})}
\label{app:poisson-lb-renewal-reduction}

We first verify that the offline-certifying cycles have finite first moments. Couple the construction of $\inst_\off$ with an execution of the greedy scheduler $\sched_1^\on$, started with one job at time $0$ and run on the same Poisson arrival sequence. At the end of every greedy regenerative cycle, the system is empty. Let $G$ be the idle time from such a clearing time until the next arrival. By the memoryless property of the Poisson process,
$$
G\sim \mathrm{Exp}(\lambda),
$$
independently of the past, and therefore
$$
\prob{G>2s}=e^{-2\lambda s}.
$$

Any greedy clearing time followed by an idle gap longer than $2s$ certifies an offline-certifying prefix. Indeed, let $T$ be such a clearing time, and let $\inst_T$ be the prefix consisting of all jobs that arrive up to time $T$. Since $\sched_1^\on$ clears the system at time $T$, there is a feasible schedule for $\inst_T$ with makespan $T$. By~\Cref{lemma:opt-makespan}, there exists an offline-optimal schedule for $\inst_T$ whose makespan is at most $T+2s$. Since the next arrival occurs strictly after $T+2s$, the prefix $\inst_T$ is offline-certifying.

Thus $\inst_\off$ is contained in the union of the greedy regenerative cycles up to and including the first greedy cycle whose following idle gap is longer than $2s$. Let $M$ denote the index of this first certifying greedy cycle. Then $M$ is geometric with success probability $e^{-2\lambda s}$, independent of the greedy cycle sizes, and
$$
\expect{M}=e^{2\lambda s}.
$$
Let $N_1^\on,N_2^\on,\ldots$ denote the numbers of jobs in successive greedy regenerative cycles. Then
$$
N(\inst_\off)
\le
\sum_{i=1}^{M}N_i^\on.
$$
Taking expectations and applying Wald's identity,
$$
\expect{N(\inst_\off)}
\le
\expect{M}\cdot \expect{\cycjobs(\sched_1^\on)}.
$$
A greedy regenerative cycle consists of the initial job at time $0$ together with the Poisson arrivals during its duration. Hence, by the compensator identity,
$$
\expect{\cycjobs(\sched_1^\on)}
=
1+\lambda\expect{\cycdur(\sched_1^\on)}.
$$
Using the greedy-cycle moment bound from~\Cref{lemma:s1-cycle-tail},
$$
\expect{\cycdur(\sched_1^\on)}
\le
\frac{5s}{2\gamma_{\gr}},
$$
we get
$$
\expect{N(\inst_\off)}
\le
D_{\lambda,\capac,s}
:=
e^{2\lambda s}
\left(
1+
\frac{5\lambda s}{2\gamma_{\gr}}
\right)
<\infty.
$$
The same domination, together with the greedy-cycle tail bound, also implies that
$$
\expect{N(\inst_\off)^2}<\infty
\qquad\text{and}\qquad
\expect{\opt(\inst_\off)^2}<\infty.
$$
Indeed, $\opt(\inst_\off)$ is bounded by the greedy flow accumulated over the first $M$ greedy cycles, and $M$ has finite second moment.

We now derive the renewal comparison. Repeatedly applying~\Cref{def:offline-certifying-cycle} partitions the infinite Poisson sequence into i.i.d. cycles
$$
\inst_\off^1,\inst_\off^2,\ldots .
$$
Let
$$
N_i:=N(\inst_\off^i),
\qquad
O_i:=\opt(\inst_\off^i).
$$
For the offline optimum, superadditivity (\Cref{lemma:superadditive}) gives the lower bound
$$
\opt(\inst(n))
\ge
\sum_{i=1}^{K(n)}O_i,
\qquad
K(n):=
\max\left\{
k:
\sum_{i=1}^k N_i\le n
\right\}.
$$
Taking expectations, dividing by $n$, and applying the renewal-reward theorem gives
$$
\liminf_{n\to\infty}
\frac{\expect{\opt(\inst(n))}}{n}
\ge
\frac{\expect{O_1}}{\expect{N_1}}.
$$
Conversely, by the certifying property of the cycles, there is an offline-optimal schedule for each cycle that completes before the first arrival of the next cycle. Concatenating these schedules gives a feasible offline schedule that solves the cycles independently. The finite second-moment bounds above imply that the incomplete residual cycle contributes $o(n)$ in expectation. Therefore
$$
\rate(\sched^*)
=
\frac{\expect{\opt(\inst_\off)}}{\expect{N(\inst_\off)}}.
$$

Now fix an $L=ks$-lookahead scheduler $\sched$. For each cycle $i$, let $S_i$ be the total flow-time contribution, in the actual execution of $\sched$, of the jobs belonging to $\inst_\off^i$. Restricting the actual schedule to only the jobs of cycle $i$ gives a feasible schedule for $\inst_\off^i$; batches that also contain jobs from other cycles can simply be viewed as partially filled batches. Hence, pointwise,
$$
S_i\ge O_i.
$$
Moreover, by the hypothesis of the lemma, whenever the event $\mathcal E$ occurs in cycle $i$, the conditional expected excess contribution of that cycle is at least $\Delta_{\mathcal E}$, uniformly over the history and internal state of the scheduler before the cycle begins. Since the next offline-certifying cycle is generated from a fresh Poisson process after the previous certifying gap, the event $\mathcal E$ has probability $\prob{\mathcal E}$ independently of the past. Thus, for every $i$,
$$
\expect{S_i-O_i}
\ge
\prob{\mathcal E}\Delta_{\mathcal E}.
$$

The total flow time of $\sched$ on the first $n$ arrivals is at least the contribution of the completed offline-certifying cycles:
$$
\sched(\inst(n))
\ge
\sum_{i=1}^{K(n)}S_i .
$$
Combining the preceding display with the renewal-reward theorem gives
$$
\liminf_{n\to\infty}
\frac{\expect{\sched(\inst(n))}}{n}
\ge
\frac{\expect{O_1}+\prob{\mathcal E}\Delta_{\mathcal E}}{\expect{N_1}}.
$$
Using the offline identity and $\expect{N_1}\le D_{\lambda,\capac,s}$, we conclude that
$$
\rate(\sched)-\rate(\sched^*)
\ge
\frac{\prob{\mathcal E}\Delta_{\mathcal E}}{\expect{N(\inst_\off)}}
\ge
\frac{\prob{\mathcal E}\Delta_{\mathcal E}}{D_{\lambda,\capac,s}}.
$$
This proves~\Cref{lemma:poisson-lb-renewal-reduction}.

\subsubsection{Probability of the Hard Cycle (\Cref{lemma:hard-event-probability})}
\label{app:poisson-lb-hard-probability}

We first compute the probability of the corresponding raw Poisson pattern after the first arrival has been shifted to time $0$, and then verify that this raw pattern is exactly one offline-certifying cycle.

For the common visible prefix, independence of Poisson increments gives probability at least
$$
e^{-\lambda ks}
\cdot
\frac{(\lambda\varepsilon_k)^{u-1}}{(u-1)!}
\cdot
\frac{(2\lambda\varepsilon_k)^v}{v!}
\cdot
\left(
\frac{(\lambda\rho s)^\capac}{\capac!}
\right)^k .
$$
Since $u+v=\capac$ and $\varepsilon_k=\beta s/(k+1)$, the polynomial factor in $k$ can be absorbed into $\exp(-O_{\capac}(k))$. Therefore the common visible prefix has probability at least
$$
\exp(-O_{\lambda,\capac,s}(k)).
$$

Conditioned on this common prefix, either one of the two hidden continuations has probability at least $\exp(-O_{\lambda,\capac,s}(k))$. Indeed, $U_W$ has length $\theta s$, while $U_E$ has length
$$
s-\delta_k-\eta_0s-2\theta s
\ge
\frac{s}{3}
$$
for the constants chosen above. The requirement that there are no other arrivals in $(ks,H_k]$ contributes only $\exp(-O_{\lambda,s}(k))$. Hence
$$
\prob{\mathcal E_k}
\ge
\exp(-O_{\lambda,\capac,s}(k)).
$$

It remains to verify that, on this raw event, the specified jobs form exactly one offline-certifying cycle. No prefix ending before all arrivals in $D_k$ can be certifying, because the next specified arrival occurs less than one service time later. Similarly, no prefix ending before all arrivals in one of the windows $J_i$, $U_E$, or $U_W$ can be certifying, since each such window has length less than $s$.

Now consider a prefix ending after all arrivals in $D_k$ or after all arrivals in $J_i$ for some $i<k$. Such a prefix contains $(i+1)\capac$ jobs, where we interpret $i=0$ as the prefix ending after $D_k$. The first time at which a full batch of $\capac$ jobs is available is at least $\delta_k-\varepsilon_k$, so every feasible schedule for this prefix has makespan at least
$$
\delta_k-\varepsilon_k+(i+1)s.
$$
The next visible burst begins at time $(i+1)s-\rho s$. Since
$$
\delta_k-\varepsilon_k+\rho s>0,
$$
the next burst begins before this prefix can be completed. Hence no such prefix is offline-certifying.

Next consider the common prefix through $J_k$. This prefix contains $(k+1)\capac$ jobs, so every feasible schedule has makespan at least
$$
\delta_k-\varepsilon_k+(k+1)s.
$$
On $\mathcal E_k^E$, the entire hidden window $U_E$ ends by time $(k+1)s-2\theta s$, which is strictly before this lower bound. On $\mathcal E_k^W$, the entire hidden window $U_W$ ends by time $ks+\delta_k+s-\theta s$, which is also strictly before this lower bound because $\varepsilon_k<\theta s$. Thus, in both continuations, the hidden block arrives before the common prefix could be certified, and the hidden jobs belong to the same offline-certifying cycle.

Finally, after the hidden block there are no arrivals until time $H_k=(3k+10)s$. The total number of specified jobs is $(k+1)\capac+u$, so these jobs can be processed in at most $k+2$ batches. Starting when the first full batch is available gives a feasible schedule clearing all specified jobs by time at most
$$
\delta_k+\varepsilon_k+(k+2)s
<
(k+3)s
<
H_k.
$$
Therefore the prefix containing all specified jobs is offline-certifying. Since no proper prefix was offline-certifying, the minimal offline-certifying prefix is exactly the specified arrival pattern. This proves~\Cref{lemma:hard-event-probability}.

\subsubsection{Constant Loss on the Hard Cycle (\Cref{lemma:hard-event-constant-loss})}
\label{app:poisson-lb-hard-loss}

Condition on $\mathcal E_k$ and on an arbitrary history and internal state of the scheduler before the cycle begins. We may assume without loss of generality that the scheduler is work-conserving on the true pending queue: whenever it starts a batch, it includes as many pending jobs as possible up to capacity, and whenever the machine is idle with at least $\capac$ pending jobs, it starts a full batch immediately. Enforcing this property can only improve the scheduler's flow time.

Let $\tau$ be the last arrival time in $\{0\}\cup I_0$, and let $r$ be the last arrival time in $D_k$. Then
$$
0\le \tau\le \varepsilon_k,
\qquad
\delta_k-\varepsilon_k\le r\le \delta_k+\varepsilon_k.
$$
At time $r$, the hidden continuation is not visible to a scheduler with lookahead $ks$. Indeed,
$$
r+ks
\le
ks+\delta_k+\varepsilon_k
<
ks+\delta_k+\eta_0s,
$$
which is strictly before the earliest possible hidden arrival in the early continuation. Moreover, if the scheduler delays its first relevant commitment until time $r+\eta_0s/4$, then the first $\capac$ jobs of the cycle, all of which have arrived by time $r$, have already accumulated an additional constant waiting cost. This will itself be enough for the desired lower bound. Hence, except for a case that already gives constant loss, the scheduler must make its first consequential batching decision before it can distinguish $\mathcal E_k^E$ from $\mathcal E_k^W$.

We compare two canonical first commitments. The \emph{early commitment} starts a batch at time $\tau$ containing the $u$ jobs in $\{0\}\cup I_0$. The \emph{waiting commitment} waits until time $r$ and starts a full batch containing the first $\capac$ jobs. After this first commitment, reveal the entire cycle and complete it optimally subject to that commitment. Let $\opt_{\early}$ and $\opt_{\wait}$ be the corresponding optimal costs.

We first compare the two commitments on the common visible prefix, ignoring the hidden block. Since release times are the same, it suffices to compare sums of completion times. Under the waiting commitment, the first full batch completes at time $r+s$, and the $k$ full bursts $J_1,\ldots,J_k$ complete at times
$$
r+2s,\ r+3s,\ldots,\ r+(k+1)s.
$$
Under the early commitment, the first $u$ jobs complete at time $\tau+s$. This leaves a $v$-job offset. Each subsequent full batch contains the $v$ leftover jobs from the previous group together with $u$ jobs from the next burst, leaving $v$ jobs to be carried forward. After $J_k$, these final $v$ jobs remain.

A direct summation gives
\begin{equation}\label{eq:poisson-lb-common-prefix-comparison}
\opt_{\wait}^{\mathrm{pref}}-\opt_{\early}^{\mathrm{pref}}
=
(k+1)\bigl(\capac(r-\tau)-vs\bigr).
\end{equation}
Using the definitions of $\delta_k$ and $\varepsilon_k$, and the choices $\alpha=10^{-6}$ and $\beta=10^{-8}$, we have
\begin{equation}\label{eq:poisson-lb-common-prefix-small}
-2\capac\alpha s
\le
(k+1)\bigl(\capac(r-\tau)-vs\bigr)
\le
-\frac12\capac\alpha s .
\end{equation}
Thus the common visible prefix favors waiting, but only by a fixed constant of order $\capac\alpha s$.

Now consider the early continuation $\mathcal E_k^E$. The hidden $u$ jobs arrive after the last visible-prefix batch of the waiting commitment has started, but before the final $v$ offset jobs of the early commitment must be processed. Thus the early commitment can batch these hidden $u$ jobs together with the final $v$ offset jobs, while the waiting commitment must process the hidden jobs in an additional batch. Relative to~\eqref{eq:poisson-lb-common-prefix-comparison}, this improves the early commitment by at least $u(r-\tau)$. Since
$$
r-\tau\ge \frac{v}{2\capac}s
$$
for the constants above, we get
$$
\opt_{\wait}-\opt_{\early}
\ge
\frac{uv}{2\capac}s-2\capac\alpha s
\ge
\frac{uv}{4\capac}s.
$$
Define
$$
\Delta_E:=\frac{uv}{4\capac}s>0.
$$
Then
\begin{equation}\label{eq:poisson-lb-early-gap}
\opt_{\wait}-\opt_{\early}
\ge
\Delta_E
\qquad
\text{on } \mathcal E_k^E .
\end{equation}

Next consider the late continuation $\mathcal E_k^W$. Let
$$
t_0:=\tau+(k+1)s
$$
be the time at which the early commitment has completed its last visible full batch and is left with the final $v$ offset jobs. Let $x$ be the last hidden arrival time in $U_W$. Then
$$
x-t_0
\in
[r-\tau-O(\theta s),\, r-\tau],
$$
more precisely,
$$
x-t_0\ge r-\tau-3\theta s.
$$
The waiting commitment has no offset at this point and processes the hidden $u$ jobs after its last visible batch. The early commitment, however, must handle both the $v$ offset jobs and the hidden $u$ jobs. Since their total number is exactly $\capac$, the best completion after time $t_0$ is one of two possibilities: either wait until all hidden jobs arrive and process the $\capac$ jobs together, or process the $v$ offset jobs first and then process the hidden $u$ jobs. Comparing these possibilities to the waiting commitment gives an additional late-continuation penalty of at least
$$
\min\left\{
\capac(x-t_0)-u(r-\tau),
\ u(s-(r-\tau))
\right\}.
$$
Using $x-t_0\ge r-\tau-3\theta s$, $r-\tau\ge \frac{v}{2\capac}s$, $r-\tau\le \frac{v}{\capac}s+o(s)$, and $\theta=10^{-3}$, this minimum is at least
$$
\frac{v^2}{4\capac}s.
$$
Together with~\eqref{eq:poisson-lb-common-prefix-small}, this gives
$$
\opt_{\early}-\opt_{\wait}
\ge
\frac{v^2}{4\capac}s
\qquad
\text{on } \mathcal E_k^W .
$$
Define
$$
\Delta_W:=\frac{v^2}{4\capac}s>0.
$$

We now relate the actual scheduler to these commitments. First suppose that, by time $r+\eta_0s/4$, the scheduler has not started any batch containing a job from the current cycle in a way that serves the first $\capac$ cycle jobs. Since all first $\capac$ cycle jobs have arrived by time $r$, delaying them until after $r+\eta_0s/4$ increases their total flow time by at least
$$
\frac{\capac\eta_0s}{4}
$$
relative to the waiting commitment. This is already a positive constant depending only on $\capac$ and $s$, and therefore gives the desired loss after decreasing constants.

We may therefore assume that the scheduler makes a relevant first commitment by time $r+\eta_0s/4$. Since
$$
r+\frac{\eta_0s}{4}+ks
<
ks+\delta_k+\eta_0s,
$$
this commitment is made before the hidden continuation is visible.

If the scheduler starts no batch before time $r$, then by work-conservation its relevant first commitment is a full batch containing the first $\capac$ jobs of the cycle, possibly delayed by inherited machine occupation. Moving this batch earlier to time $r$ and then revealing the rest of the cycle can only improve its cost. Hence its cost is at least $\opt_{\wait}$, and by~\eqref{eq:poisson-lb-early-gap} it incurs loss at least $\Delta_E$ on $\mathcal E_k^E$.

If the scheduler starts a batch before time $r$ containing exactly the $u$ early jobs from $\{0\}\cup I_0$ and no jobs from $D_k$, then, after ignoring any jobs inherited from previous cycles, moving this batch to time $\tau$ and completing the rest of the cycle optimally can only improve the cost contribution of the current-cycle jobs. Hence its cost is at least $\opt_{\early}$, and it incurs loss at least $\Delta_W$ on $\mathcal E_k^W$.

It remains to consider any other relevant batch started before time $r$. Such a batch contains either fewer than the $u$ early jobs, or it contains some but not all of the $v$ jobs in $D_k$. In either case, restricted to the jobs of the current cycle, it creates an offset different from the balanced $v$-job offset of the early commitment. The same completion-time summation as in~\eqref{eq:poisson-lb-common-prefix-comparison} shows that this imbalance accumulates over the $k$ visible bursts and creates an additional visible-prefix cost of order at least $ks-O_{\capac,s}(1)$. Since $k\ge5$ and the constants above are fixed, this is larger than $\min\{\Delta_E,\Delta_W,\capac\eta_0s/4\}$ after decreasing constants if necessary. Thus this case also incurs constant loss, regardless of which hidden continuation occurs.

Finally, conditioned on $\mathcal E_k$, both continuations $\mathcal E_k^E$ and $\mathcal E_k^W$ have probability bounded below by a positive constant depending only on $\capac$ and $s$. Indeed, both hidden windows have lengths bounded below by positive multiples of $s$, and the same no-arrival condition outside the hidden window appears in both continuations. Therefore, uniformly over the scheduler's entering history and internal state,
$$
\expect{
\sched(\inst_\off)-\opt(\inst_\off)
\mid
\mathcal E_k
}
\ge
\Omega_{\capac,s}(1).
$$
This proves~\Cref{lemma:hard-event-constant-loss}.

\section{Lookahead Predictions with Errors
(Omitted Proofs from~\Cref{sec:prediction_error})}
\label{app:prediction-error-proofs}

In this appendix, we prove \Cref{thm:imperfect-prediction}. We separate the analysis into two steps. First, we establish consistency bounds for prediction-assisted implementations of $\advname$ and $\algname$ under sufficiently accurate predictions. We then obtain robustness by comparing these bounds with the corresponding prediction-free guarantees and selecting the better algorithm based only on the known parameters $L$ and $\epsilon$.

Let $\advname_{\pred}$ denote the guarded prediction-assisted implementation of $\advname$ described in \Cref{app:error-random-shift}, and let $\algname_{\pred}$ denote the margin-certified prediction-assisted implementation of $\algname$ described in \Cref{app:error-prefix-lookahead}. The two consistency bounds that drive the proof are stated below.

\begin{lemma}
\label{lemma:error-random-shift-consistency}
Assume $\capac=\infty$, $\epsilon\in[0,s]$, and $L>2\epsilon$. Then, for every true instance $\inst$ and predicted instance $\inst'$ satisfying $|r_j-r'_j|\le\epsilon$ for every $j\in [n]$, the guarded prediction-assisted implementation of $\advname$ satisfies
\begin{equation}
\label{eq:error-random-shift-consistency}
\frac{\expect{\advname_{\pred}(\inst)}}{\opt(\inst)}
\le
1
+
4\frac{\epsilon}{s}
+
\frac{(3s+2\epsilon)(2s+2\epsilon)}
{s(L+2s)}.
\end{equation}
\end{lemma}

\begin{lemma}
\label{lemma:error-prefix-lookahead-consistency}
Assume $0<\lambda s<\capac$. There exist constants $\kappa,\gamma>0$, depending only on $\lambda$, $s$, and $\capac$, such that, for every $L\ge4s$ and $\epsilon\in[0,s/8]$, the margin-certified prediction-assisted implementation of $\algname$ satisfies
\begin{equation}
\label{eq:error-prefix-lookahead-consistency}
\rate(\algname_{\pred})
\le
\rate(\sched^*)
+
\kappa e^{-\gamma L}
+
4\epsilon.
\end{equation}
\end{lemma}

Before proving these lemmas, we show how they imply the full consistency and robustness guarantees of \Cref{thm:imperfect-prediction}, which we restate for the reader's convenience.

\predictionerror*

\begin{proof}
We begin with the uncapacitated adversarial setting. Define
\begin{equation}
\label{eq:error-random-shift-bound-definition}
A_{\mathrm{RS}}(L,\epsilon)
:=
4\frac{\epsilon}{s}
+
\frac{(3s+2\epsilon)(2s+2\epsilon)}
{s(L+2s)}.
\end{equation}
The final error-aware implementation $\mathrm{Error}\advname$ runs $\advname_{\pred}$ when
$$
\epsilon\le s,
\qquad
L>2\epsilon,
\qquad\text{and}\qquad
A_{\mathrm{RS}}(L,\epsilon)\le1,
$$
and otherwise ignores the predictions and runs the greedy scheduler $\greedy$ on the true instance.

When the prediction-assisted implementation is used, \Cref{lemma:error-random-shift-consistency} gives competitive ratio at most
$1+A_{\mathrm{RS}}(L,\epsilon)$. When greedy is used, \Cref{thm:adv-online-bounds} gives competitive ratio at most $2$. Moreover, whenever $\epsilon\le s$,
$$
A_{\mathrm{RS}}(L,\epsilon)
\le
4\frac{\epsilon}{s}
+
20\frac{s}{L+2s}
\le
4\frac{\epsilon}{s}
+
20\frac{s}{L}.
$$
If the assumptions of \Cref{lemma:error-random-shift-consistency} hold but
$A_{\mathrm{RS}}(L,\epsilon)>1$, then
$$
1
<
A_{\mathrm{RS}}(L,\epsilon)
\le
20\left(
\frac{s}{L}
+
\frac{\epsilon}{s}
\right),
$$
so the additive loss of $1$ in the greedy guarantee is absorbed by the
$O(s/L+\epsilon/s)$ term.

The remaining parameter regimes are also covered by the constant robust guarantee. If $\epsilon>s$, then
$$
\frac{s}{L}+\frac{\epsilon}{s}>1.
$$
If instead $L\le2\epsilon$, then
$$
\frac{s}{L}+\frac{\epsilon}{s}
\ge
\frac{s}{L}+\frac{L}{2s}
\ge
\sqrt{2}.
$$
Thus, after increasing the absolute constant hidden by the $O(\cdot)$ notation, the prediction-free ratio $2$ is absorbed whenever the prediction-assisted implementation is not used. Consequently,
$$
\frac{\expect{\mathrm{Error}\advname(\inst)}}{\opt(\inst)}
\le
1+
\min\left\{
O\left(
\frac{s}{L}
+
\frac{\epsilon}{s}
\right),
1
\right\}.
$$

We next consider the capacitated Poisson setting. Let
$$
k^*
\in
\mathrm{argmin}_{k\in[\capac]}
\rate(\sched_k^{\on})
$$
be the best prediction-free threshold policy, whose optimality among online policies follows from \Cref{thm:online-thresholds}, and define its gap from the offline optimum by
$$
B_{\lambda,s,\capac}
:=
\rate(\sched_{k^*}^{\on})
-
\rate(\sched^*).
$$
Since $0<\lambda s<\capac$, this is a finite constant depending only on $\lambda$, $s$, and $\capac$.

The final error-aware implementation $\mathrm{Error}\algname$ runs $\algname_{\pred}$ when
$$
L\ge4s,
\qquad
\epsilon\le\frac{s}{8},
\qquad\text{and}\qquad
\kappa e^{-\gamma L}+4\epsilon
\le
B_{\lambda,s,\capac},
$$
and otherwise ignores the predictions and runs $\sched_{k^*}^{\on}$. Hence, whenever the assumptions of \Cref{lemma:error-prefix-lookahead-consistency} hold,
\begin{equation}
\label{eq:error-prefix-min-under-assumptions}
\rate(\mathrm{Error}\algname)
\le
\rate(\sched^*)
+
\min\left\{
\kappa e^{-\gamma L}+4\epsilon,
B_{\lambda,s,\capac}
\right\}.
\end{equation}

It remains to absorb the parameter regimes outside the consistency lemma into the robust term. Choose constants $\overline{\kappa}\ge\kappa$ and $c\ge4$ sufficiently large that
$$
\overline{\kappa}e^{-4\gamma s}
\ge
B_{\lambda,s,\capac}
\qquad\text{and}\qquad
c\frac{s}{8}
\ge
B_{\lambda,s,\capac}.
$$
If $L<4s$, then
$$
\overline{\kappa}e^{-\gamma L}
\ge
\overline{\kappa}e^{-4\gamma s}
\ge
B_{\lambda,s,\capac}.
$$
If $\epsilon>s/8$, then
$$
c\epsilon
>
c\frac{s}{8}
\ge
B_{\lambda,s,\capac}.
$$
Therefore, in either excluded regime, the first argument of the minimum below is at least the prediction-free gap, and the algorithm obtains the robust guarantee by running $\sched_{k^*}^{\on}$. Combining this observation with \eqref{eq:error-prefix-min-under-assumptions} gives, for all $L$ and $\epsilon$,
$$
\rate(\mathrm{Error}\algname)
\le
\rate(\sched^*)
+
\min\left\{
\overline{\kappa}e^{-\gamma L}
+
c\epsilon,
B_{\lambda,s,\capac}
\right\}.
$$
Equivalently,
$$
\rate(\mathrm{Error}\algname)
\le
\rate(\sched^*)
+
\min\left\{
O_{\lambda,s,\capac}(1)
\left(
e^{-\Omega_{\lambda,s,\capac}(L)}
+
\epsilon
\right),
O_{\lambda,s,\capac}(1)
\right\}.
$$
This proves \Cref{thm:imperfect-prediction}.
\end{proof}

It remains to prove
\Cref{lemma:error-random-shift-consistency,lemma:error-prefix-lookahead-consistency}. We first establish a general schedule-transfer lemma used by both algorithms
(\Cref{app:prediction-shifting-proof}). We then prove the consistency guarantee for the guarded implementation of $\advname$
(\Cref{app:error-random-shift}) and finally analyze the margin-certified implementation of $\algname$
(\Cref{app:error-prefix-lookahead}).

\subsection{Transferring Predicted Schedules}
\label{app:prediction-shifting-proof}

We begin with the schedule-transfer property used throughout this appendix.

\begin{lemma}
\label{lemma:prediction-shifting}
Let $\inst=\{r_j\}_{j=1}^n$ and $\inst'=\{r'_j\}_{j=1}^n$ be any two ordered sequences of arrival times satisfying
$$
|r_j-r'_j|\le\epsilon
\qquad
\text{for every }j\in[n].
$$
Let $\sched'$ be any feasible schedule for $\inst'$. Construct a schedule $\sched$ for $\inst$ by preserving the same ordered partition of ranks into batches and starting each batch as soon as the machine is free and all of its corresponding true jobs have arrived. Then every batch under $\sched$ completes at most $\epsilon$ later than the corresponding batch under $\sched'$, and
$$
\sched(\inst)
\le
\sched'(\inst')+2n\epsilon.
$$
Consequently,
$$
|\opt(\inst)-\opt(\inst')|
\le
2n\epsilon.
$$
\end{lemma}

\begin{proof}
Let $Q_1,\ldots,Q_m$ be the batches of $\sched'$, in the order in which they are processed. Let $a'_\ell$ be the start time of $Q_\ell$ under $\sched'$, and let $a_\ell$ be the start time of the corresponding batch under $\sched$. We prove by induction that
$$
a_\ell\le a'_\ell+\epsilon
$$
for every $\ell\in[m]$.

Since $\sched'$ is feasible,
$$
a'_\ell\ge\max_{j\in Q_\ell}r'_j.
$$
Therefore, every corresponding true job has arrived by time
$$
\max_{j\in Q_\ell}r_j
\le
\max_{j\in Q_\ell}r'_j+\epsilon
\le
a'_\ell+\epsilon.
$$
For the first batch, this proves the claim. For $\ell>1$, the induction hypothesis also gives
$$
a_{\ell-1}+s
\le
a'_{\ell-1}+s+\epsilon
\le
a'_\ell+\epsilon,
$$
so the machine is free by time $a'_\ell+\epsilon$. Hence $Q_\ell$ can start by that time.

Now fix a job $j\in Q_\ell$, and let $T_j$ and $T'_j$ denote its completion times under $\sched$ and $\sched'$, respectively. The start-time bound gives
$$
T_j\le T'_j+\epsilon.
$$
Since $r_j\ge r'_j-\epsilon$, it follows that
$$
T_j-r_j
\le
T'_j-r'_j+2\epsilon.
$$
Summing over all jobs yields
$$
\sched(\inst)
\le
\sched'(\inst')+2n\epsilon.
$$
Applying the same argument to an offline-optimal schedule for $\inst'$ gives
$$
\opt(\inst)
\le
\opt(\inst')+2n\epsilon.
$$
Interchanging $\inst$ and $\inst'$ gives the reverse inequality and proves the final claim.
\end{proof}

\subsection{Error-Aware Random Shift
(Proof of~\Cref{lemma:error-random-shift-consistency})}
\label{app:error-random-shift}

Fix a true instance $\inst$, a predicted instance $\inst'$, and a realization of the random offset used by $\advname$. Let
$$
H:=L+2s
$$
denote the period length. The guarded implementation retains the same periodic decomposition as $\advname$, but replaces each original active interval
$$
[b,b+L]
$$
by the active core
$$
[b+\epsilon,b+L-\epsilon].
$$
All jobs whose predicted arrivals lie in an active core are called active, and all remaining jobs are called ignored. Since $L>2\epsilon$, every active core has positive length $L-2\epsilon$, and the gap between consecutive active cores has length
$$
H-(L-2\epsilon)=2s+2\epsilon.
$$

At the beginning $b$ of an original active interval, the algorithm's
$L$-lookahead window reveals every predicted arrival in its active core. The algorithm computes an offline-optimal schedule for these predicted active jobs. Ignored predicted jobs are handled using the stitching rule of $\advname$: each ignored job is inserted into the first batch after its arrival generated by the next active-core schedule. If no such batch starts within the first $s$ time units of the next active core, the algorithm starts a fresh stitching batch at the beginning of that core.

We first verify that the predicted schedules constructed for different active cores do not interfere after being transferred to the true instance. Fix an active core
$$
[b+\epsilon,b+L-\epsilon].
$$
Every predicted job in this core arrives by time $b+L-\epsilon$. As in the perfect-prediction analysis of $\advname$, an offline-optimal schedule for these jobs completes by time
$$
b+L-\epsilon+2s
=
b+H-\epsilon.
$$
Indeed, in the uncapacitated setting, the last batch started before the final arrival completes within $s$ time units of that arrival, after which all remaining jobs can be processed together in one additional batch.

By \Cref{lemma:prediction-shifting}, transferring this schedule to the corresponding true jobs delays every batch completion by at most $\epsilon$. The transferred schedule therefore completes by time
$$
b+H,
$$
which is the beginning of the next original active interval. Hence the transferred schedules associated with consecutive active cores remain disjoint.

We now analyze the flow time of the schedule constructed on the predicted instance. Let $\inst'_{\act}$ and $\inst'_{\ign}$ denote the subinstances of active and ignored predicted jobs, respectively, and let $\sched'$ denote the complete predicted schedule constructed by the guarded implementation.

For each active core, $\sched'$ uses an offline-optimal schedule for the jobs in that core. These schedules are mutually nonoverlapping, and inserting ignored jobs into their batches does not affect the completion times of active jobs because $\capac=\infty$. Consequently, the flow time paid by $\sched'$ on active jobs is at most
$$
\opt(\inst'_{\act}).
$$

Every ignored predicted job incurs flow time at most
$$
4s+2\epsilon.
$$
Indeed, such a job waits for at most $2s+2\epsilon$ until the next active core begins, for at most an additional $s$ until either an active batch or a stitching batch starts, and then for one service time $s$. Therefore,
$$
\sched'(\inst')
\le
\opt(\inst'_{\act})
+
(4s+2\epsilon)|\inst'_{\ign}|.
$$
By superadditivity of the offline optimum (\Cref{lemma:superadditive}), we have
$$
\opt(\inst')
\ge
\opt(\inst'_{\act})
+
\opt(\inst'_{\ign}).
$$
Since every job incurs flow time at least $s$ under every feasible schedule,
$$
\opt(\inst'_{\ign})
\ge
s|\inst'_{\ign}|.
$$
Combining these inequalities gives
\begin{align}
\sched'(\inst')
&\le
\opt(\inst')
+
(4s+2\epsilon)|\inst'_{\ign}|
-
\opt(\inst'_{\ign}) \notag\\
&\le
\opt(\inst')
+
(3s+2\epsilon)|\inst'_{\ign}|.
\label{eq:error-random-shift-predicted-bound}
\end{align}

The schedule executed by $\advname_{\pred}$ on the true instance is obtained by transferring the ordered predicted batch structure of $\sched'$ to the corresponding true jobs. Applying \Cref{lemma:prediction-shifting} and then comparing the two offline optima yields
\begin{align}
\advname_{\pred}(\inst)
&\le
\sched'(\inst')+2n\epsilon \notag\\
&\le
\opt(\inst')
+
(3s+2\epsilon)|\inst'_{\ign}|
+
2n\epsilon \notag\\
&\le
\opt(\inst)
+
(3s+2\epsilon)|\inst'_{\ign}|
+
4n\epsilon.
\label{eq:error-random-shift-fixed-offset}
\end{align}

It remains to average over the random offset. Since the offset is uniform over a period of length $H=L+2s$, the phase of every fixed predicted arrival is uniform over that period. The total portion of each period lying outside the active core has length $2s+2\epsilon$. Hence, for every predicted job,
$$
\Pr[\text{the job is ignored}]
=
\frac{2s+2\epsilon}{L+2s},
$$
and therefore
$$
\expect{|\inst'_{\ign}|}
=
n\frac{2s+2\epsilon}{L+2s}.
$$
Taking expectations in
\eqref{eq:error-random-shift-fixed-offset} gives
$$
\expect{\advname_{\pred}(\inst)}
\le
\opt(\inst)
+
4n\epsilon
+
n(3s+2\epsilon)
\frac{2s+2\epsilon}{L+2s}.
$$
Finally, every job contributes at least $s$ to the offline-optimal flow time, so
$$
\opt(\inst)\ge ns.
$$
Dividing by $\opt(\inst)$ yields
$$
\frac{\expect{\advname_{\pred}(\inst)}}{\opt(\inst)}
\le
1
+
4\frac{\epsilon}{s}
+
\frac{(3s+2\epsilon)(2s+2\epsilon)}
{s(L+2s)},
$$
which proves
\Cref{lemma:error-random-shift-consistency}.

\subsection{Error-Aware Prefix Lookahead
(Proof of~\Cref{lemma:error-prefix-lookahead-consistency})}
\label{app:error-prefix-lookahead}

Assume throughout that $0<\lambda s<\capac$, $L\ge4s$, and
$\epsilon\in[0,s/8]$. The predicted arrival sequence need not itself be
Poisson, so all probabilistic arguments below are applied to the true
Poisson arrival sequence.

\subsubsection{The Error-Aware Algorithm}

Consider a cycle beginning with a true arrival at time $a_0$, and shift
time so that $a_0=0$. The error-aware implementation examines predicted
arrivals only within the first
$$
L-2\epsilon
$$
units of the lookahead window. For every predicted prefix beginning with
the rank of the arrival at time $0$, let $T'_j$ be the minimum makespan
among offline-optimal schedules for that predicted prefix, and let
$r'_{j+1}$ be the next predicted arrival. The prefix is accepted only if
$r'_{j+1}$ is visible within the shortened horizon and
\begin{equation}
\label{eq:error-prefix-certificate}
T'_j+2\epsilon<r'_{j+1}.
\end{equation}
The algorithm then transfers an offline-optimal schedule attaining
makespan $T'_j$ to the corresponding true jobs, using the execution rule
of \Cref{lemma:prediction-shifting}.

Every accepted prefix is a valid certificate for the true instance.
Indeed, the transferred schedule completes by time at most
$T'_j+\epsilon$, whereas
$$
r_{j+1}\ge r'_{j+1}-\epsilon.
$$
Therefore, \eqref{eq:error-prefix-certificate} implies
$$
T'_j+\epsilon
<
r'_{j+1}-\epsilon
\le
r_{j+1}.
$$
Thus the transferred schedule clears the true system before the next
true arrival.

If no prefix satisfies
\eqref{eq:error-prefix-certificate}, the algorithm follows the fallback
rule of $\algname$: it runs $\sched^{\on}_1$ on the true arrivals until
all jobs represented in the shortened search window have arrived and
the system has subsequently become empty.

\subsubsection{Probability of Fallback}

We next show that the probability of entering fallback mode remains
exponentially small in $L$. Run $\sched^{\on}_1$ on the true arrival
sequence, starting with the first job at time $0$. Consider the event
that $\sched^{\on}_1$ clears the system at some time
\begin{equation}
\label{eq:error-prefix-clearing-horizon}
T\le L-3s-3\epsilon
\end{equation}
and that the next true arrival, at time $R$, satisfies
\begin{equation}
\label{eq:error-prefix-certifying-gap}
T+\frac{5s}{2}<R\le T+3s.
\end{equation}
We claim that this event guarantees that the error-aware algorithm finds
a valid predicted certificate.

Let $\inst_T$ be the true prefix consisting of all jobs arriving up to
time $T$, and let $\inst'_T$ be the corresponding predicted prefix.
Since $\sched^{\on}_1$ clears $\inst_T$ at time $T$, it gives a feasible
schedule for $\inst_T$ with makespan $T$. Transferring this schedule to
$\inst'_T$ using \Cref{lemma:prediction-shifting} gives a feasible
predicted schedule with makespan at most $T+\epsilon$. By the makespan
comparison of \Cref{lemma:opt-makespan}, there is therefore an offline-optimal schedule
for $\inst'_T$ with makespan at most
$$
T+2s+\epsilon.
$$
Consequently, the minimum makespan $T'_j$ among offline-optimal
schedules for this predicted prefix satisfies
$$
T'_j\le T+2s+\epsilon.
$$

The next predicted arrival satisfies
$$
r'_{j+1}\ge R-\epsilon.
$$
Since $\epsilon\le s/8$, we have
$$
2s+4\epsilon\le\frac{5s}{2}.
$$
Using \eqref{eq:error-prefix-certifying-gap}, it follows that
\begin{align*}
T'_j+2\epsilon
&\le
T+2s+3\epsilon\\
&<
R-\epsilon\\
&\le
r'_{j+1}.
\end{align*}
Hence the predicted prefix satisfies
\eqref{eq:error-prefix-certificate}.

It remains to verify that the next predicted arrival is visible within
the shortened search horizon. By
\eqref{eq:error-prefix-clearing-horizon} and
\eqref{eq:error-prefix-certifying-gap},
$$
r'_{j+1}
\le
R+\epsilon
\le
T+3s+\epsilon
\le
L-2\epsilon.
$$
Thus the entire predicted certificate is visible to the algorithm.

We now bound the probability that no such certificate occurs. At every
clearing time of $\sched^{\on}_1$, the idle time $G$ until the next true
arrival is exponentially distributed with rate $\lambda$, independently
of the past. Therefore, each clearing time has an independent
certifying probability
$$
p_{\mathrm{cert}}
:=
\Pr\left[
G\in\left(\frac{5s}{2},3s\right]
\right]
=
e^{-5\lambda s/2}-e^{-3\lambda s}
>0.
$$
The proof of \Cref{thm:poisson-upper} shows that the durations of the
regenerative cycles of $\sched^{\on}_1$ have an exponential moment.
Applying the same geometric-sum argument with success probability
$p_{\mathrm{cert}}$ therefore gives constants
$c_0,\eta_0>0$, depending only on $\lambda$, $s$, and $\capac$, such
that the probability of observing no certifying clearing time by time
$x$ is at most
$$
c_0e^{-\eta_0x}.
$$

Taking
$$
x:=L-3s-3\epsilon,
$$
and using $L\ge4s$ and $\epsilon\le s/8$, we obtain
$$
x
\ge
L-\frac{27s}{8}
\ge
\frac{5L}{32}.
$$
Hence, if $\bad$ denotes the event that no valid predicted certificate
is found,
\begin{equation}
\label{eq:error-prefix-fallback-probability}
\Pr[\bad]
\le
c_0
\exp\left(
-\frac{5\eta_0}{32}L
\right).
\end{equation}
Thus there exist constants $\kappa_0,\gamma_0>0$, depending only on
$\lambda$, $s$, and $\capac$, such that
$$
\Pr[\bad]\le\kappa_0e^{-\gamma_0L}.
$$

\subsubsection{Rate Bound}

We finally compare the rate of $\algname_{\pred}$ with the
offline-optimal rate. Let $\inst_{\cyc}$ denote the true instance in one
cycle, let
$$
N(\inst_{\cyc}):=|\inst_{\cyc}|,
$$
and let $F_{\pred}(\inst_{\cyc})$ be the flow time incurred by
$\algname_{\pred}$ during that cycle.

Suppose first that a valid predicted certificate is found. Because the
transferred schedule clears the true system before the next true
arrival, the accepted prefix is exactly the cycle instance. Let
$\inst'_{\cyc}$ be the corresponding predicted prefix. Applying
\Cref{lemma:prediction-shifting} to the predicted offline-optimal
schedule gives
\begin{align}
F_{\pred}(\inst_{\cyc})
&\le
\opt(\inst'_{\cyc})
+
2\epsilon N(\inst_{\cyc}) \notag\\
&\le
\opt(\inst_{\cyc})
+
4\epsilon N(\inst_{\cyc}).
\label{eq:error-prefix-good-cycle}
\end{align}

On a fallback cycle, the algorithm runs $\sched^{\on}_1$ on the true
arrivals until the system clears. Therefore, writing
$\sched^{\on}_1(\inst_{\cyc})$ for the resulting flow time, we have
pointwise
\begin{equation}
\label{eq:error-prefix-cycle-loss}
F_{\pred}(\inst_{\cyc})
-
\opt(\inst_{\cyc})
\le
4\epsilon N(\inst_{\cyc})
+
\mathbf{1}_{\bad}
\,\sched^{\on}_1(\inst_{\cyc}).
\end{equation}

As in the proof of \Cref{thm:poisson-upper}, the cycle decomposition and
superadditivity of the offline optimum imply
\begin{align}
\rate(\algname_{\pred})-\rate(\sched^*)
&\le
\frac{
\expect{
F_{\pred}(\inst_{\cyc})
-
\opt(\inst_{\cyc})
}
}{
\expect{N(\inst_{\cyc})}
} \notag\\
&\le
4\epsilon
+
\frac{
\expect{
\mathbf{1}_{\bad}
\,\sched^{\on}_1(\inst_{\cyc})
}
}{
\expect{N(\inst_{\cyc})}
}.
\label{eq:error-prefix-rate-reduction}
\end{align}
Since every cycle contains at least one job,
$\expect{N(\inst_{\cyc})}\ge1$. By Cauchy--Schwarz,
\begin{equation}
\label{eq:error-prefix-cauchy-schwarz}
\expect{
\mathbbm{1}_{\bad}
\,\sched^{\on}_1(\inst_{\cyc})
}
\le
\sqrt{
\Pr[\bad]\,
\expect{
\bigl(
\sched^{\on}_1(\inst_{\cyc})
\bigr)^2
}
}.
\end{equation}

The same enlarged-greedy domination used in the proof of
\Cref{thm:poisson-upper} applies here. Indeed, every true job represented
in the shortened predicted horizon arrives by time at most
$$
L-2\epsilon+\epsilon<L.
$$
Consequently, there is a constant $K_0>0$, depending only on
$\lambda$, $s$, and $\capac$, such that
\begin{equation}
\label{eq:error-prefix-fallback-second-moment}
\expect{
\bigl(
\sched^{\on}_1(\inst_{\cyc})
\bigr)^2
}
\le
K_0(L+s)^2.
\end{equation}
Combining
\eqref{eq:error-prefix-fallback-probability},
\eqref{eq:error-prefix-rate-reduction},
\eqref{eq:error-prefix-cauchy-schwarz}, and
\eqref{eq:error-prefix-fallback-second-moment} gives
$$
\rate(\algname_{\pred})-\rate(\sched^*)
\le
4\epsilon
+
\sqrt{\kappa_0K_0}\,
(L+s)e^{-\gamma_0L/2}.
$$
The polynomial factor can be absorbed into the exponential. In
particular, after decreasing the exponent and increasing the
multiplicative constant, there exist constants $\kappa,\gamma>0$,
depending only on $\lambda$, $s$, and $\capac$, such that
$$
\sqrt{\kappa_0K_0}\,
(L+s)e^{-\gamma_0L/2}
\le
\kappa e^{-\gamma L}
$$
for every $L\ge4s$. Therefore,
$$
\rate(\algname_{\pred})
\le
\rate(\sched^*)
+
\kappa e^{-\gamma L}
+
4\epsilon,
$$
which proves
\Cref{lemma:error-prefix-lookahead-consistency}.

\section{Parameter Generalization}
\label{sec:parameter-generalization}

To evaluate whether the real-trace results depend on the parameter choices used
in the main experiments, we independently vary the batch capacity $\capac$ and setup time s. Figure~\ref{fig:capacity-generalization-D} fixes
$s=0.5\widetilde{\Delta}$ and varies the capacity multiplier in
$\capac=c\hat{\lambda}s$ for $c\in\{1,3,5\}$, where
$\widetilde{\Delta}$ is the median inter-arrival time and $\hat{\lambda}$ is
the empirical arrival rate of the trace instance. Figure~\ref{fig:setup-generalization-D} fixes the capacity multiplier at $c=3$ and varies the setup time over $s\in\{0.25,1,2\}\widetilde{\Delta}$. Across all six settings, the same qualitative trends hold: lookahead consistently improves over the best
fixed-wait baseline, and larger lookahead windows move the achieved flow time
closer to the offline optimum. These results show that our conclusions are not
specific to the parameterization used in the main evaluation.

\begin{figure}[!htbp]
    \centering

    \begin{subfigure}[t]{0.31\textwidth}
        \centering
        \includegraphics[width=\linewidth]{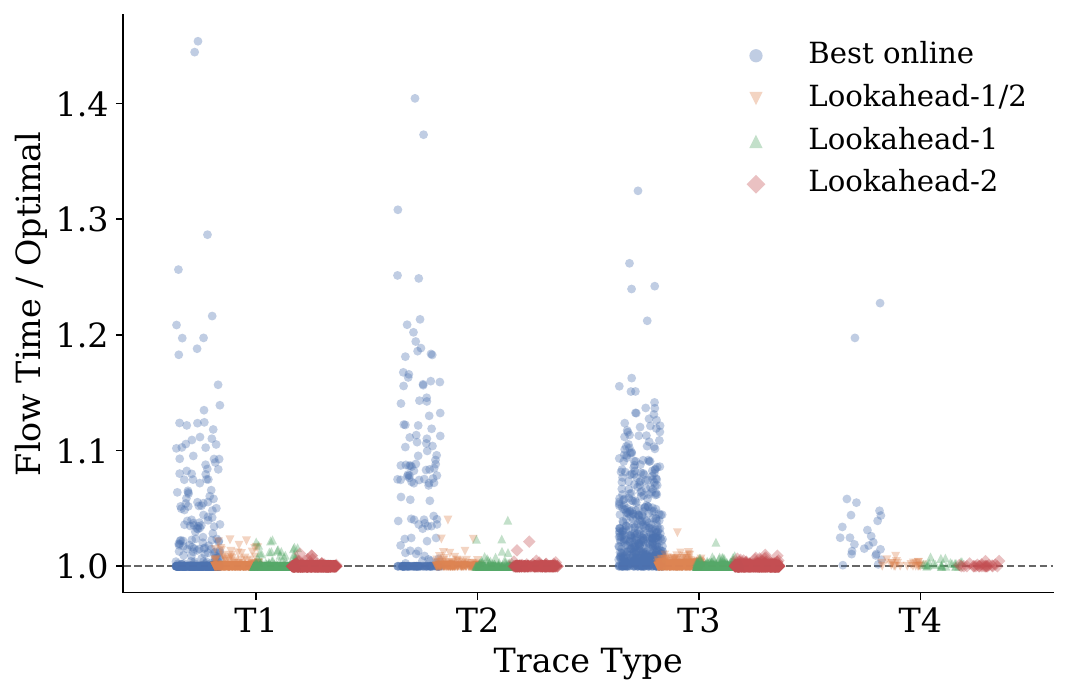}
        \caption{$\capac = 1\times \hat{\lambda}s$}
        \label{fig:capacity-1x-D}
    \end{subfigure}
    \hfill
    \begin{subfigure}[t]{0.31\textwidth}
        \centering
        \includegraphics[width=\linewidth]{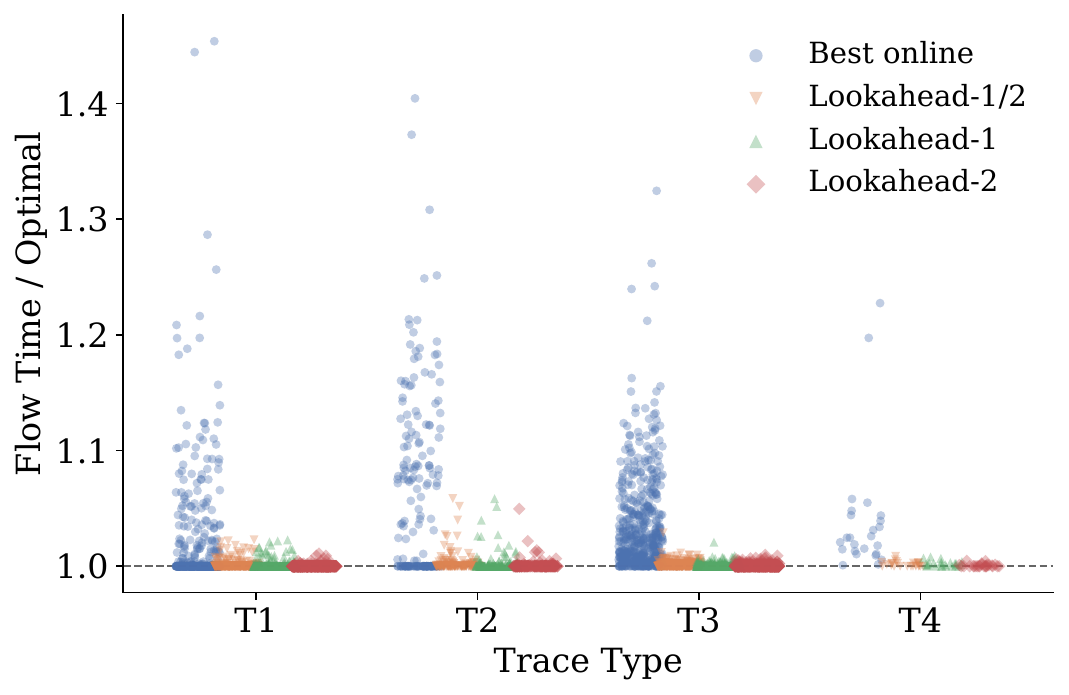}
        \caption{$\capac = 3\times \hat{\lambda}s$}
        \label{fig:capacity-3x-D}
    \end{subfigure}
    \hfill
    \begin{subfigure}[t]{0.31\textwidth}
        \centering
        \includegraphics[width=\linewidth]{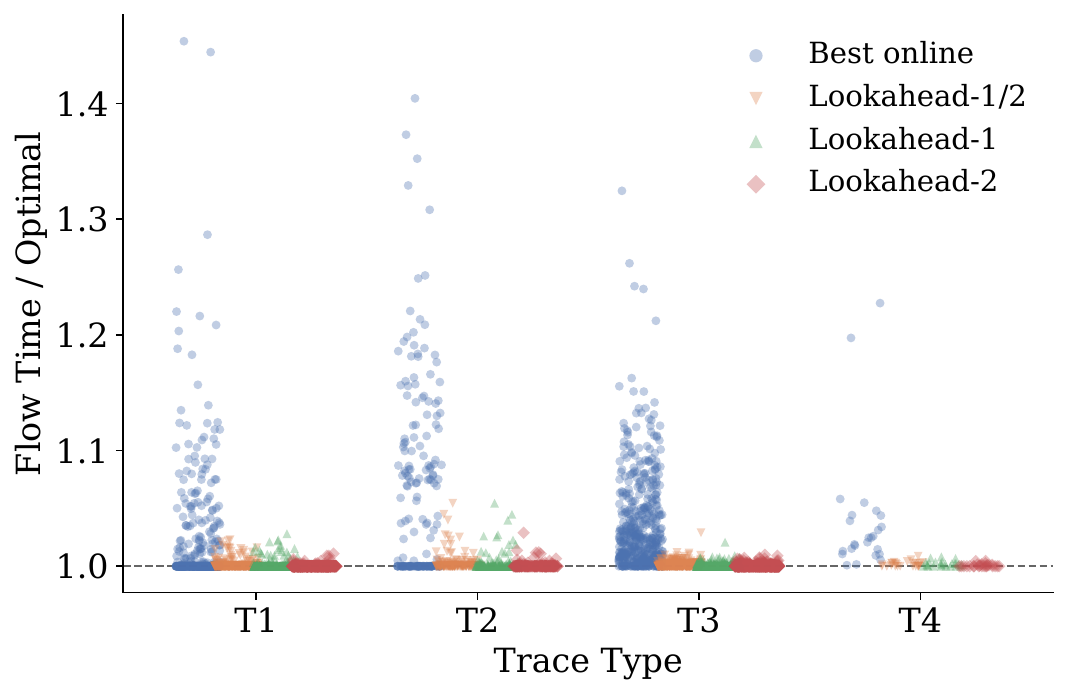}
        \caption{$\capac = 5\times \hat{\lambda}s$}
        \label{fig:capacity-5x-D}
    \end{subfigure}

    \caption{
        Sensitivity of the real-trace results to the batch-capacity parameter.
    }
    \label{fig:capacity-generalization-D}

    \medskip

    \begin{subfigure}[t]{0.31\textwidth}
        \centering
        \includegraphics[width=\linewidth]{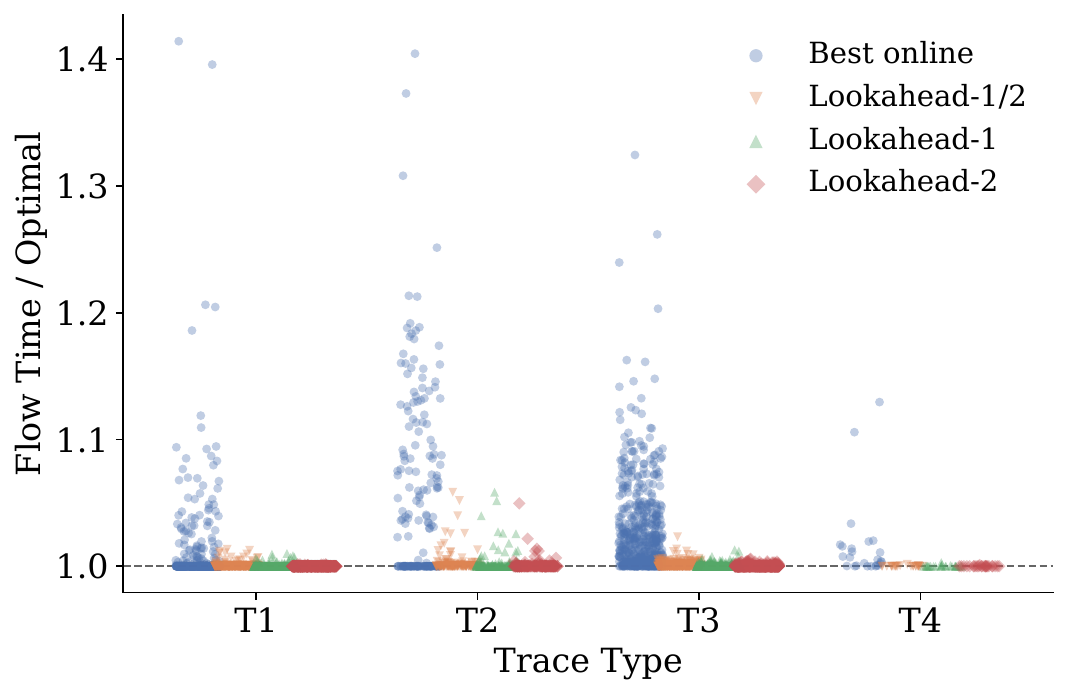}
        \caption{$s=0.25\widetilde{\Delta}$}
        \label{fig:setup-025x-D}
    \end{subfigure}
    \hfill
    \begin{subfigure}[t]{0.31\textwidth}
        \centering
        \includegraphics[width=\linewidth]{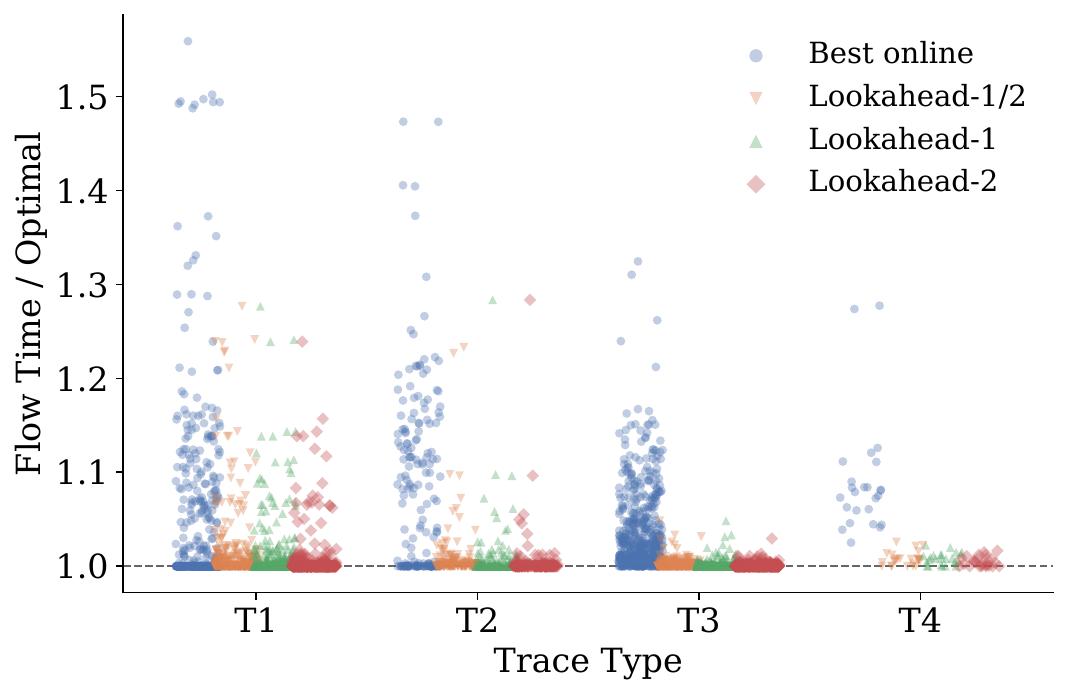}
        \caption{$s=\widetilde{\Delta}$}
        \label{fig:setup-1x-D}
    \end{subfigure}
    \hfill
    \begin{subfigure}[t]{0.31\textwidth}
        \centering
        \includegraphics[width=\linewidth]{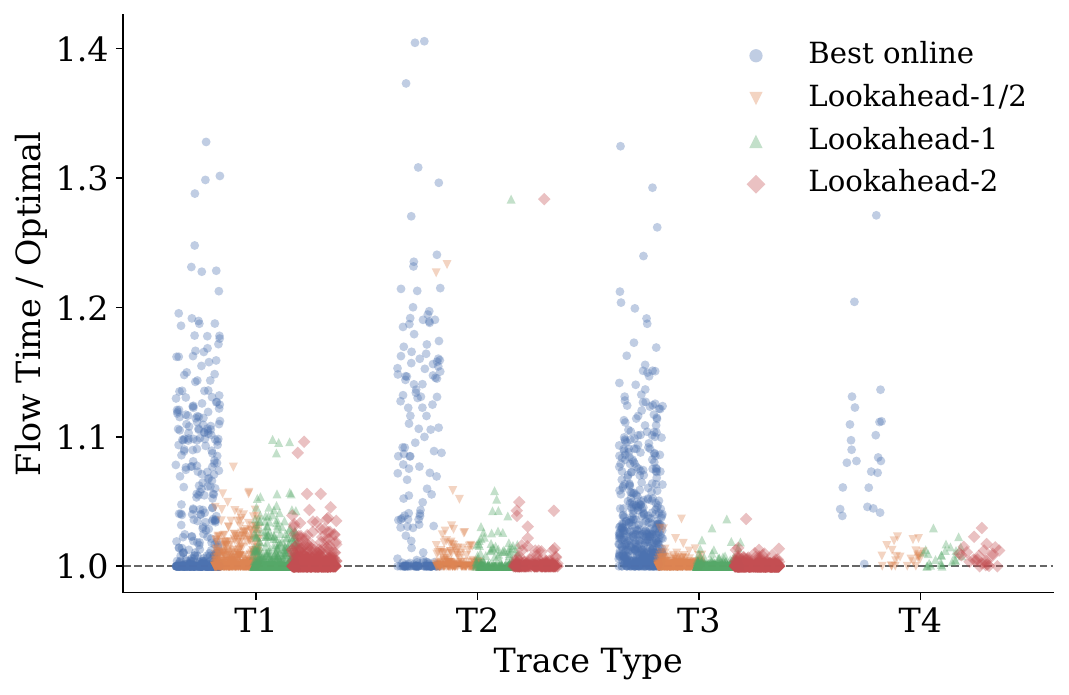}
        \caption{$s=2\widetilde{\Delta}$}
        \label{fig:setup-2x-D}
    \end{subfigure}

    \caption{
        Sensitivity of the real-trace results to the setup-time parameter.
    }
    \label{fig:setup-generalization-D}
\end{figure}

\section{Evaluation of \Cref{alg:poisson_lookahead}}
\label{sec:lookahead-complete}

We implement and evaluate $\algname$ (Algorithm~\ref{alg:poisson_lookahead}) as LookaheadComplete alongside the Lookahead implementation we introduced in Section~\ref{sec:eval} with both Poisson arrival simulations (Figure~\ref{fig:lookahead-evaluation-complete}) and real-world traces (Figure~\ref{fig:capacity-generalization-complete}--\ref{fig:setup-generalization-complete}). The conclusions in Section~\ref{sec:eval} remain true.

\medskip
\begin{figure}[h]
    \centering

    \begin{subfigure}[t]{0.3\textwidth}
        \centering
        \includegraphics[width=\linewidth]{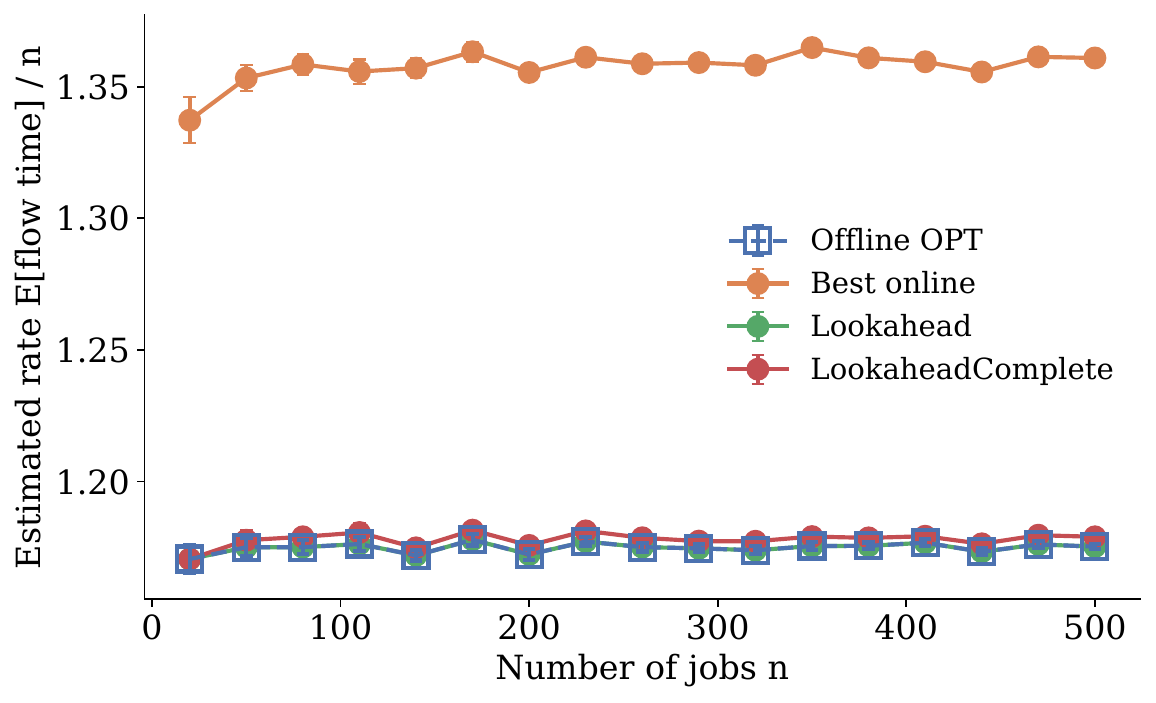}
        \caption{Rate Convergence}
        \label{fig:online-lookahead-rate}
    \end{subfigure}
    \hfill
    \begin{subfigure}[t]{0.3\textwidth}
        \centering
        \includegraphics[width=\linewidth]{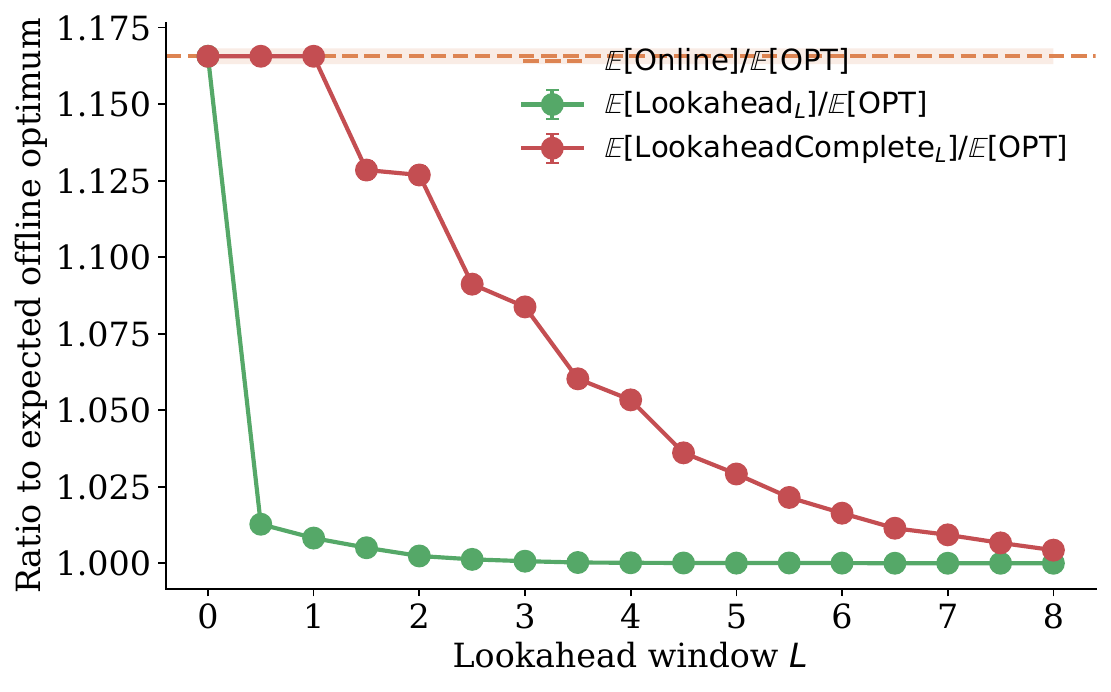}
        \caption{Lookahead Effect}
        \label{fig:lookahead-ratio-vs-L}
    \end{subfigure}
    \hfill
    \begin{subfigure}[t]{0.3\textwidth}
        \centering
        \includegraphics[width=\linewidth]{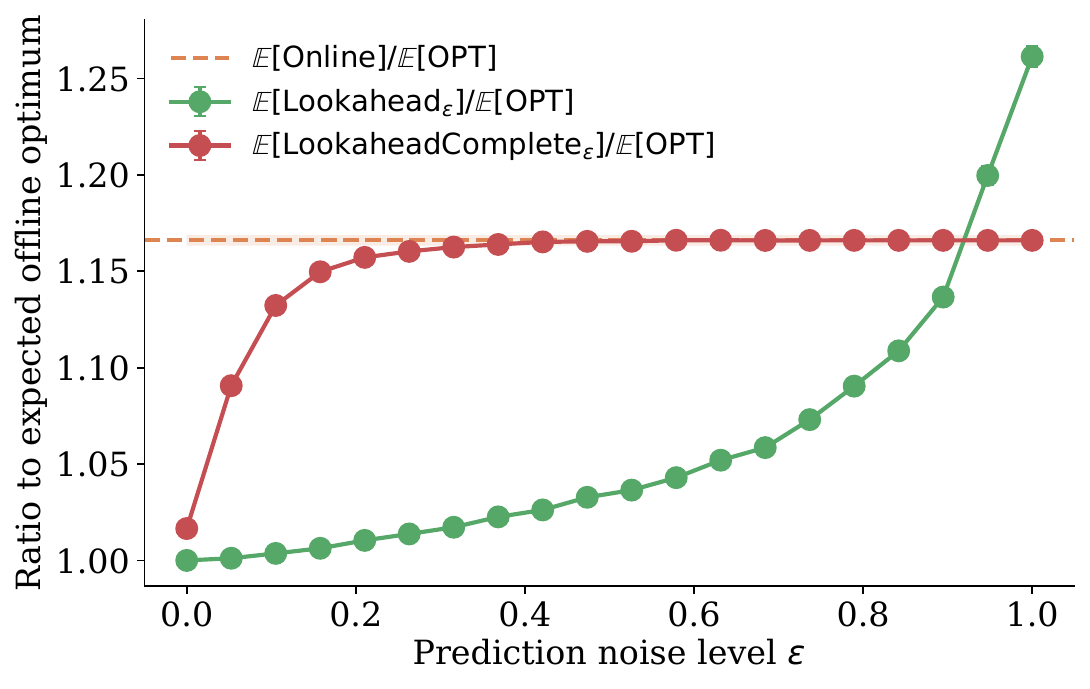}
        \caption{Error Robustness}
        \label{fig:lookahead-ratio-vs-epsilon}
    \end{subfigure}

    \caption{Experimental evaluation of online batching with both lookahead implementations.}
    \label{fig:lookahead-evaluation-complete}
\end{figure}

\begin{figure}[!htbp]
    \centering

    \begin{subfigure}[t]{0.31\textwidth}
        \centering
        \includegraphics[width=\linewidth]{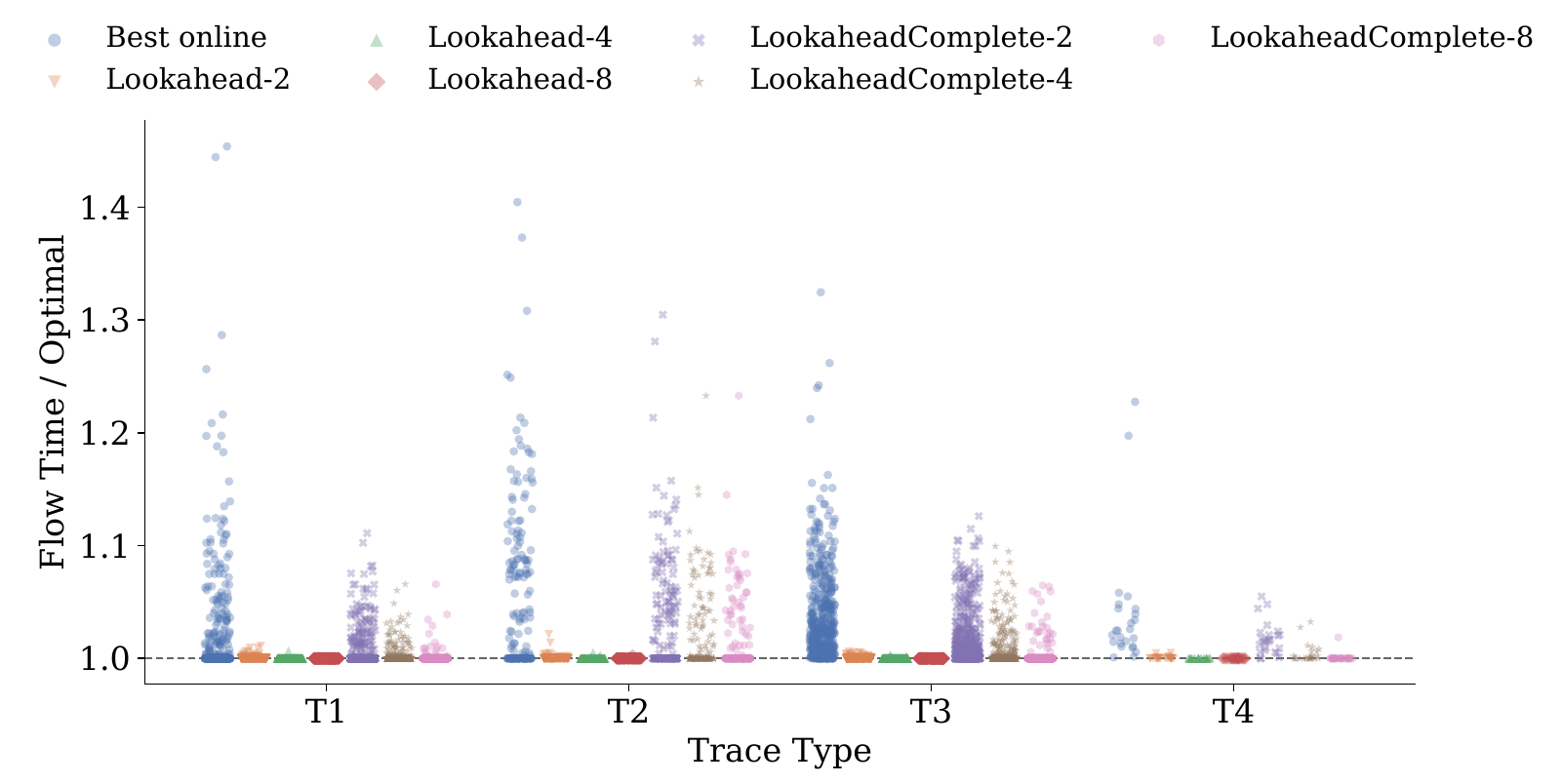}
        \caption{$\capac = 1\times \hat{\lambda}s$}
        \label{fig:capacity-1x}
    \end{subfigure}
    \hfill
    \begin{subfigure}[t]{0.31\textwidth}
        \centering
        \includegraphics[width=\linewidth]{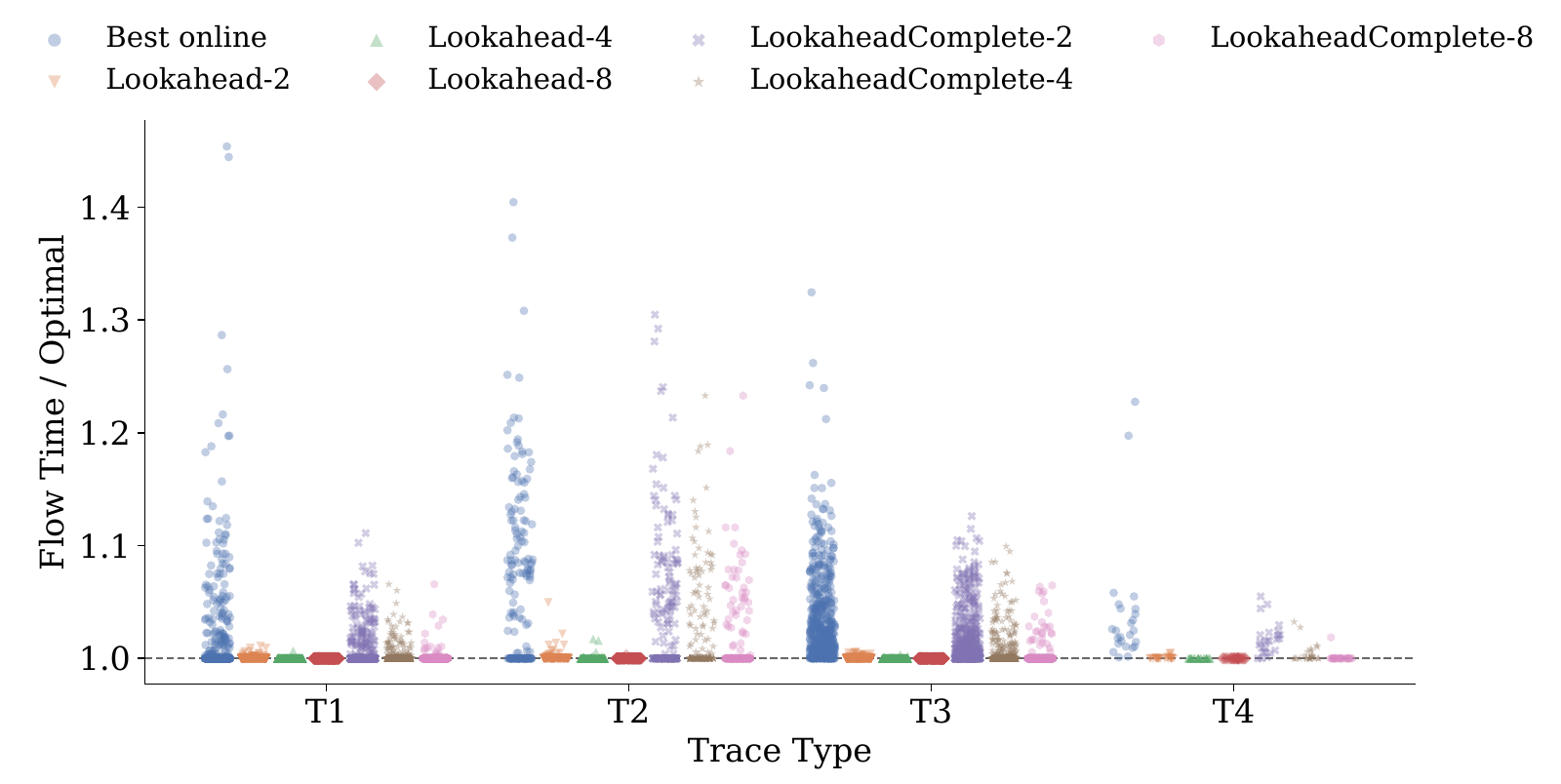}
        \caption{$\capac = 3\times \hat{\lambda}s$}
        \label{fig:capacity-3x}
    \end{subfigure}
    \hfill
    \begin{subfigure}[t]{0.31\textwidth}
        \centering
        \includegraphics[width=\linewidth]{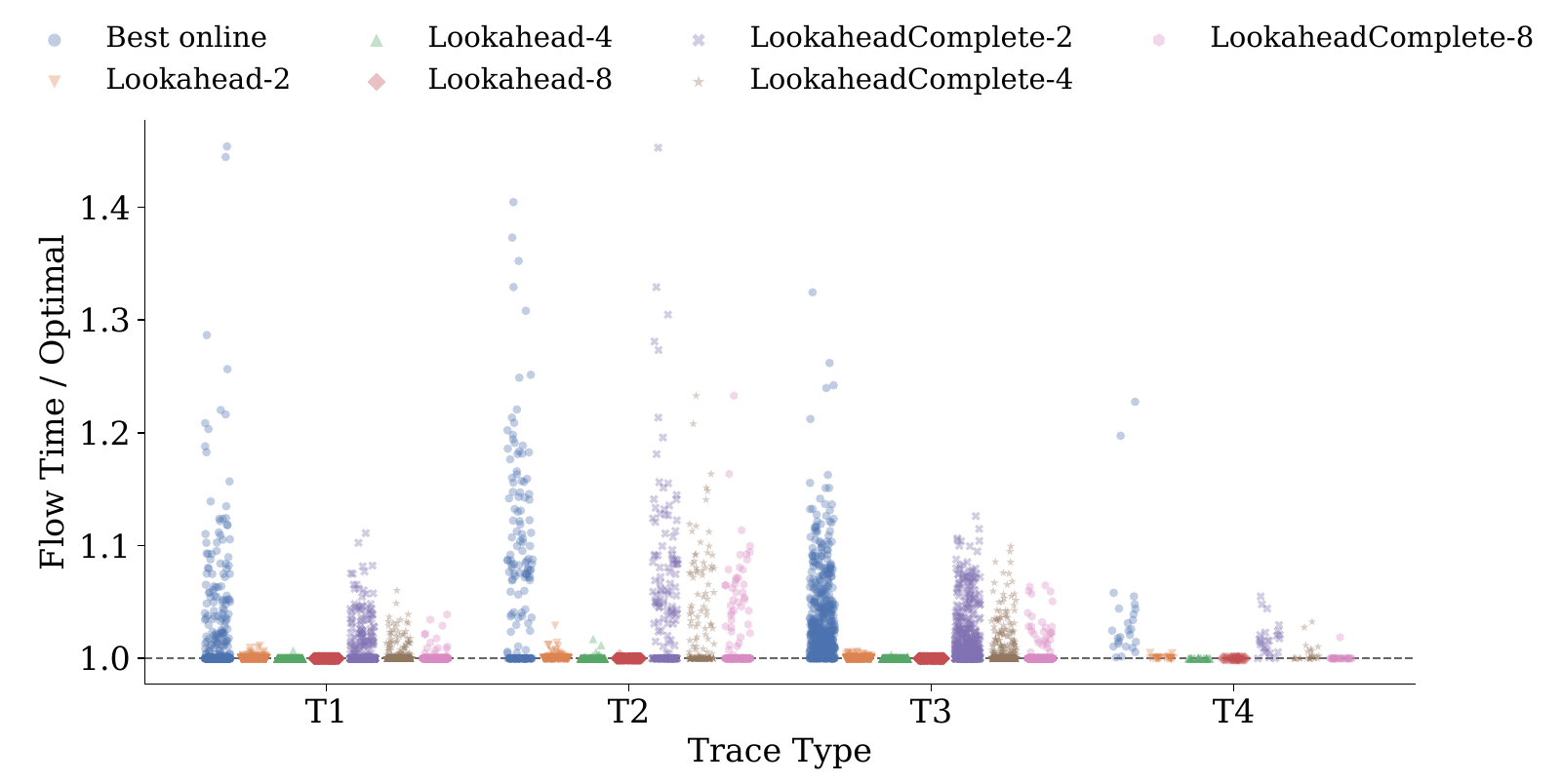}
        \caption{$\capac = 5\times \hat{\lambda}s$}
        \label{fig:capacity-5x}
    \end{subfigure}

    \caption{
        Sensitivity of the real-trace results to the batch-capacity parameter with both lookahead implementations.
    }
    \label{fig:capacity-generalization-complete}

    \medskip

    \begin{subfigure}[t]{0.31\textwidth}
        \centering
        \includegraphics[width=\linewidth]{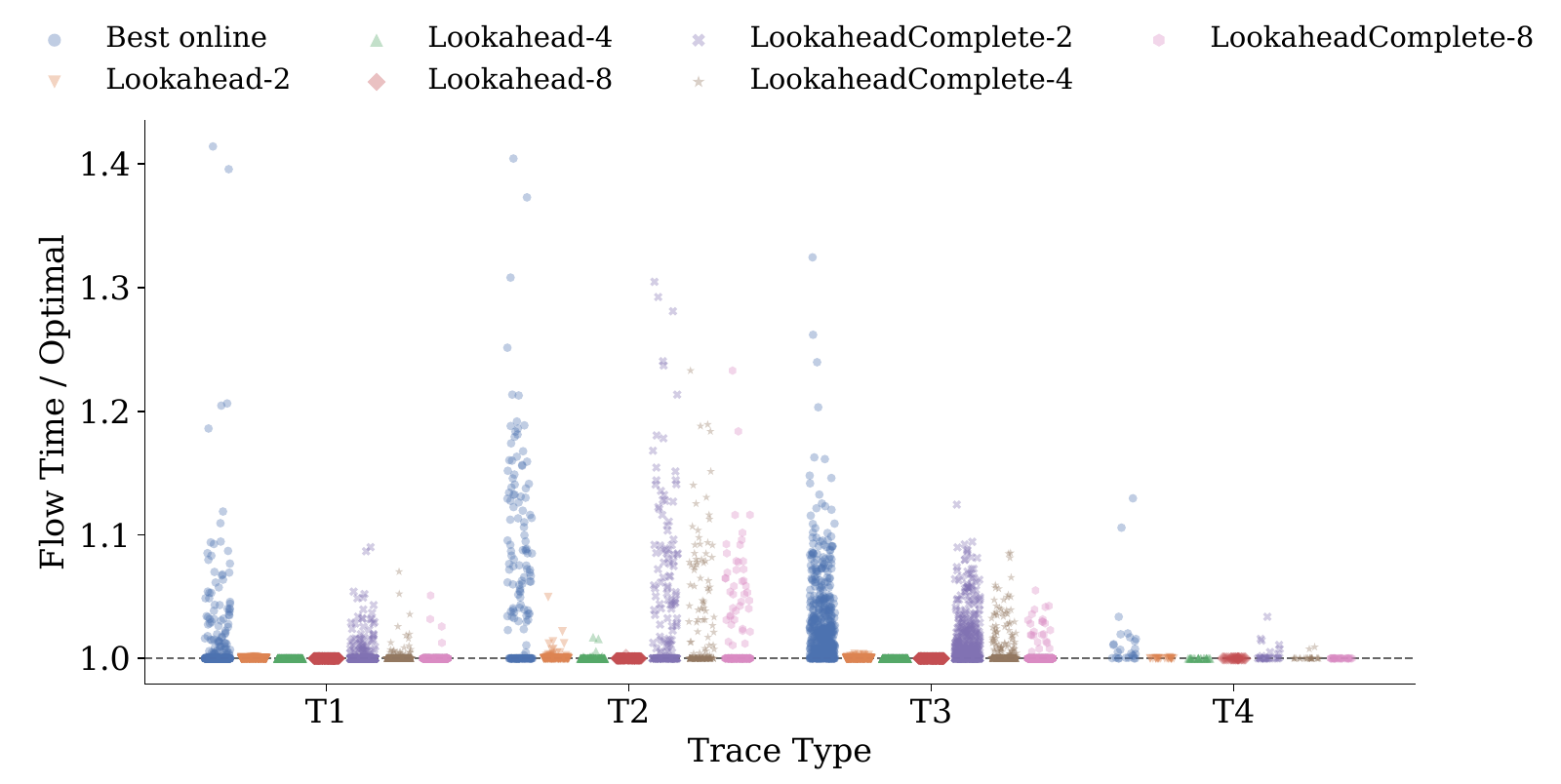}
        \caption{$s=0.25\widetilde{\Delta}$}
        \label{fig:setup-025x}
    \end{subfigure}
    \hfill
    \begin{subfigure}[t]{0.31\textwidth}
        \centering
        \includegraphics[width=\linewidth]{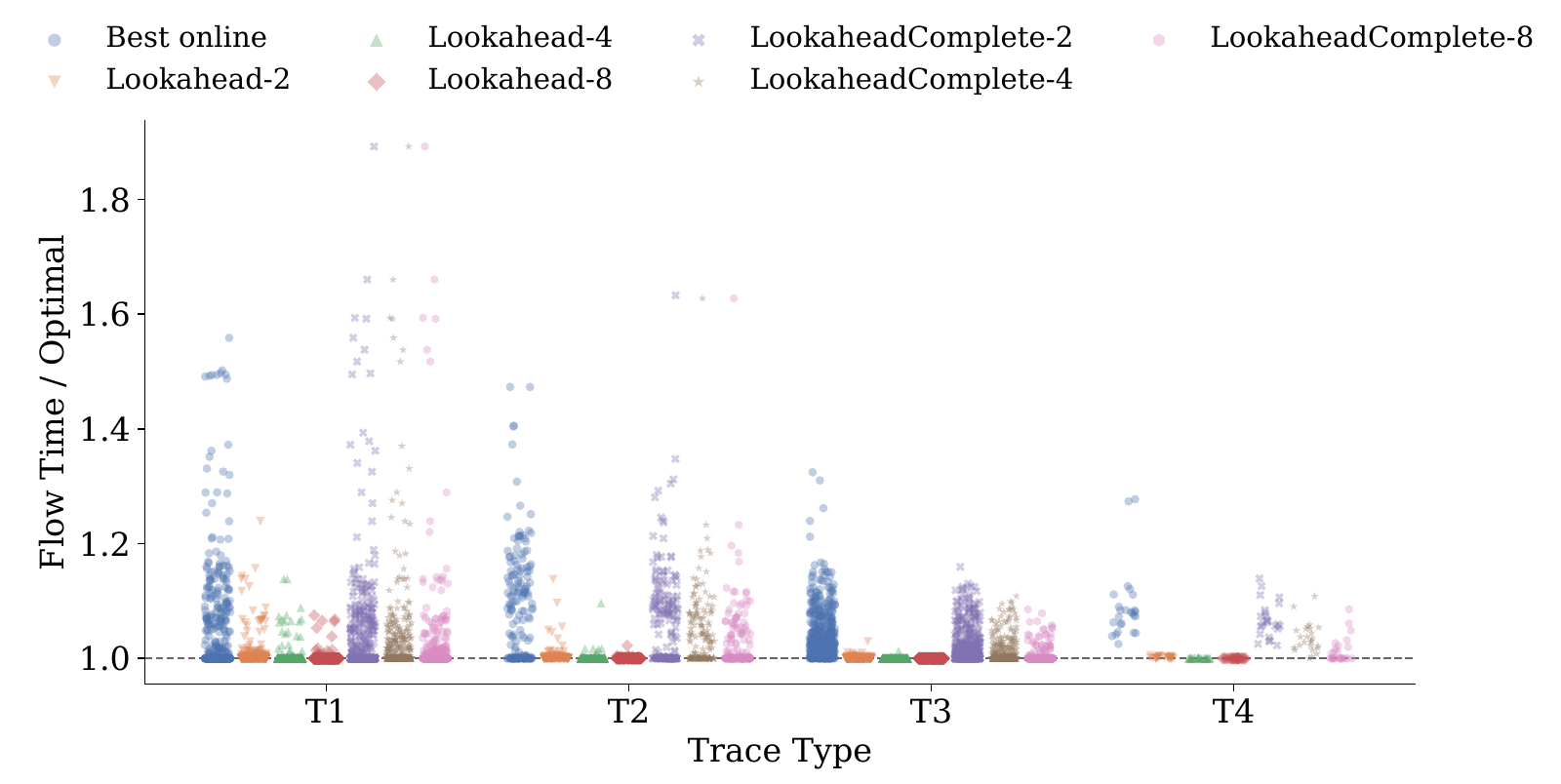}
        \caption{$s=\widetilde{\Delta}$}
        \label{fig:setup-1x}
    \end{subfigure}
    \hfill
    \begin{subfigure}[t]{0.31\textwidth}
        \centering
        \includegraphics[width=\linewidth]{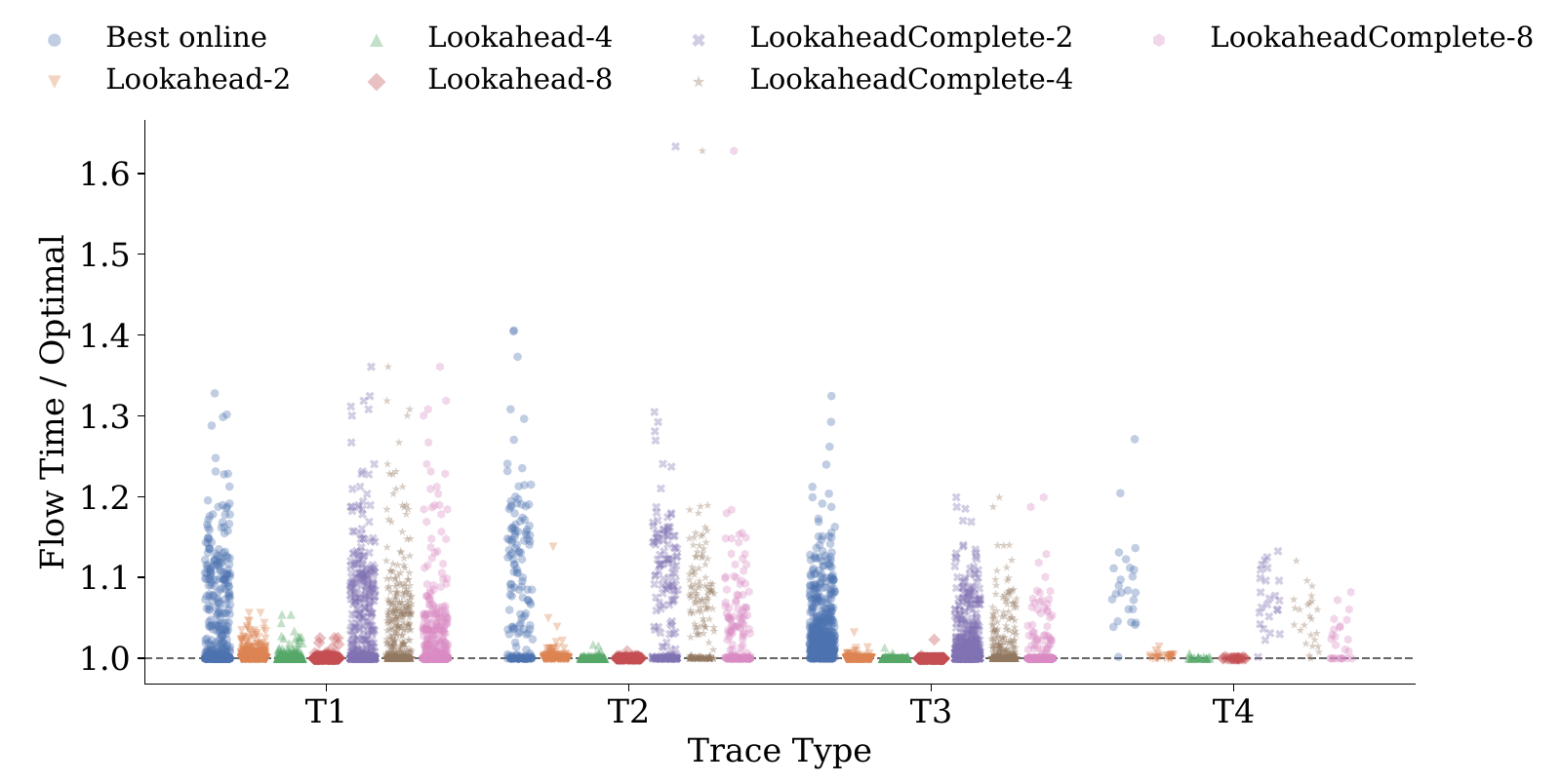}
        \caption{$s=2\widetilde{\Delta}$}
        \label{fig:setup-2x}
    \end{subfigure}

    \caption{
        Sensitivity of the real-trace results to the setup-time parameter with both lookahead implementations.
    }
    \label{fig:setup-generalization-complete}
\end{figure}

In Figure~\ref{fig:lookahead-evaluation-complete}, we use rate $\lambda=0.9$, setup time $s=1$ and batch capacity $\capac=3$ so that $\lambda s < \capac$. We use $L=8s=8$ as the default lookahead window. 
Figure~\ref{fig:lookahead-evaluation-complete}(a) shows that both Lookahead and LookaheadComplete converge to the offline OPT, and has a large gap between best online. Figure~\ref{fig:lookahead-evaluation-complete}(b)--(c) shows that Lookahead generally performs closer to OPT because LookaheadComplete is much more conservative than Lookahead. Lookahead re-plans at every idle moment using everything in its window, so every batching decision benefits from lookahead, while LookaheadComplete requires an idle gap longer than roughly the setup time inside the window.

Figure~\ref{fig:lookahead-evaluation-complete}(c) shows the other side of this trade-off. The noisy LookaheadComplete is more conservative in following the lookahead (and therefore the predictions), adding a safety margin and outperforming the simple lookahead algorithm for large values of the error. In real-world traces (Figure~\ref{fig:capacity-generalization-complete}--\ref{fig:setup-generalization-complete}), the same qualitative trends hold across all six settings: both Lookahead and LookaheadComplete consistently improve over the best fixed-wait baseline.

\clearpage

\bibliographystyle{ACM-Reference-Format}
\bibliography{sample-base}

\end{document}